\documentclass[11pt]{article}
\usepackage{graphicx}
\usepackage{color}
\usepackage{amsmath, amssymb, amsthm, ascmac, booktabs, caption, comment, enumerate, bm}
\usepackage{latexsym, mathabx, mathrsfs, mathtools, psfrag, setspace, here, subcaption}

\graphicspath{{./figs/}}

\usepackage{natbib}
\mathtoolsset{showonlyrefs=true}
\definecolor{mycolor}{cmyk}{0.9, 0.5, 0, 0}
\usepackage[colorlinks=true, linkcolor=mycolor, filecolor=mycolor, urlcolor=mycolor, citecolor=mycolor, backref=page]{hyperref}
\usepackage[margin = 2.4cm]{geometry}

\usepackage{titlesec}
\titleformat*{\section}{\Large\bfseries\sffamily}
\titleformat*{\subsection}{\large\bfseries\sffamily}
\titleformat*{\subsubsection}{\large\bfseries\sffamily}

\DeclareMathOperator*{\argmin}{argmin}

\renewcommand{\hat}{\widehat}
\renewcommand{\tilde}{\widetilde}
\renewcommand{\check}{\widecheck}
\renewcommand{\bar}{\overline}
\newcommand{\bbE}{\mathbb{E}}

\newcommand{\bbR}{\mathbb{R}}

\newcommand{\calL}{\mathcal{L}}
\newcommand{\calR}{\mathcal{R}}

\numberwithin{equation}{section}
 
\allowdisplaybreaks[2]

\theoremstyle{definition}
\newtheorem{theorem}{Theorem}[section]
\newtheorem{assumption}{Assumption}[section]
\newtheorem{definition}{Definition}[section]
\newtheorem{corollary}{Corollary}[section]
\newtheorem{example}{Example}[section]
\newtheorem{lemma}{Lemma}[section]
\newtheorem{proposition}{Proposition}[section]
\newtheorem{remark}{Remark}[section]

\numberwithin{figure}{section}
\numberwithin{table}{section}

\title{Dynamic Network Autoregressive Models for Functional Panel Data}

\author{
    Tomohiro Ando\textsuperscript{1} \and 
    Tadao Hoshino\textsuperscript{2}
}

\date{}

\begin{document}

\maketitle

\textsuperscript{1}Melbourne Business School, University of Melbourne, 200 Leicester Street, Carlton, Victoria 3053, Australia. Email: T.Ando@mbs.edu \\

\textsuperscript{2}Corresponding author. School of Political Science and Economics, Waseda University, 1-6-1 Nishi-Waseda, Shinjuku-ku, Tokyo 169-8050, Japan.
Email: thoshino@waseda.jp

\vspace{1em}

\begin{abstract}
	This study proposes a novel functional dynamic network autoregressive framework for analyzing network and dynamic interactions of functional outcomes in panel data settings.
	In this framework, an individual's outcome function is influenced by his/her previous outcome and the outcomes of others through a simultaneous equation system. 
	To estimate the functional parameters of interest, we need to cope with the endogeneity issue arising from these interactions among outcome functions. 
	We address this issue by developing a novel functional moment-based estimator. 
	We establish the consistency, convergence rate, and pointwise asymptotic normality of the proposed estimator. 
	Additionally, we discuss the estimation of marginal effects and functional network impulse responses.
	As an empirical illustration, we analyze the demand for a bike-sharing service in the U.S.
	The results reveal statistically significant spatial interactions in bike availability, with interaction patterns varying over the time of day.
\end{abstract}

\vspace{1em}

\noindent \textbf{Keywords}: functional data analysis, panel data, endogeneity, network autoregressive models, bike-sharing systems

\clearpage

\section{Introduction} \label{sec:introduction}

The availability of functional data has been rapidly expanding across all fields of research, leading to a growing need for statistical tools that appropriately account for the characteristics of each type of functional data.  
In the analysis of socioeconomic data, there are at least two key aspects that should be addressed.  
The first is that an individual's decision or behavioral pattern may be influenced by their own previous outcome and the outcomes of others through social networks -- \textit{outcome interactions} over time and across individuals. 
The second is that individuals are heterogeneous, even after controlling for observable characteristics -- \textit{unobserved heterogeneity} of individuals.  
Therefore, analyzing socioeconomic functional data requires functional models that jointly capture both of these aspects, which is the aim of this study.

More specifically, to account for outcome interactions over time and across units, we extend the dynamic network (or spatial) autoregressive (DNAR) modeling approach to a functional response model.
To address the unobserved individual heterogeneity, we introduce the functional fixed effects approach, given the availability of panel data.
When the response variable is a scalar rather than a function, there already exists a vast body of studies investigating fixed-effect DNAR models for panel data, such as
\begin{align}\label{eq:PNAR}
	Y_{it} = \alpha_0 \sum_{j=1}^n w_{i,j} Y_{jt} + \gamma_0 Y_{i,t-1} + X_{it}^\top\beta_0 + f_{0i} + \varepsilon_{it}, \;\; i = 1, \ldots, n, \; t = 1, \ldots, T 
\end{align}
and its variants (e.g., \citealp{yu2008quasi, lee2010spatial, lee2014efficient, qu2016instrumental, kuersteiner2020dynamic, zhang2025nonlinear}, among others).
Here, $Y_{it}$ is a scalar outcome, $w_{i,j}$ denotes a known weight term measuring the social or geographical proximity between units $i$ and $j$, $X_{it}$ is a vector of covariates, $f_{0i}$ represents a fixed effect specific to each $i$, and $\varepsilon_{it}$ denotes an error term. 
The term $\sum_{j=1}^n w_{i,j}Y_{jt}$ captures the local trend of the outcome variable in the neighborhood of $i$, and the term $Y_{i,t-1}$ captures temporal dependence in $i$'s own outcome.
Model \eqref{eq:PNAR} is typically applied in fields such as health, real estate, transportation, education, and municipal data.  
However, with the increasing availability of functional data in these fields, such as real-time activity recognition, real-time population mobility and congestion patterns, and regional wealth distributions, scalar models like \eqref{eq:PNAR} may fail to appropriately capture the complex nature of these interactions. 

The above discussion motivates us to extend \eqref{eq:PNAR} to the following model: for $s \in [0,1]$,
\begin{align}\label{eq:model}
    Y_{it}(s) =  \alpha_0(s) \sum_{j = 1}^n w_{i,j} A_1(Y_{jt}, s) + \gamma_0(s) A_2(Y_{i,t-1}, s) + X_{it}^\top \beta_0(s) + f_{0i}(s) + \varepsilon_{it}(s), 
\end{align}
where $Y_{it}$ represents the outcome function of interest, which may or may not be a smooth function of $s$, $\alpha_0$ is the network interaction effect function, $\gamma_0$ is the dynamic effect function, $\beta_0$ is a vector of functional coefficients, $ f_{0i}$ is the fixed effect function, and $\varepsilon_{it}$ is the functional error term with mean zero at each $s$. 
Here, $A_k(\cdot , s)$, $k \in \{1,2\}$, denotes a known functional, whose functional form may differ according to the research interest.

Since the response variable is a function, we can consider various types of interaction patterns.
The most typical form of interaction would be the "concurrent" interaction, where only the responses at the same evaluation point $s$ are influential.
In this case, $A_k(\cdot, s)$ is a point-evaluation functional at $s$: $A_k(Y_{it}, s) = Y_{it}(s)$.
When $s$ represents time, past outcomes should affect future outcomes (but the converse should not), which motivates us to employ $A_1(Y_{jt}, s) = \int_0^1 Y_{jt}(u)\nu_1(u,s)\text{d}u$, where $\nu_1(u,s)$ is a user-chosen kernel function that is non-negative, increasing in $u$ up to $s$, and $\nu_1(u,s) = 0$ for $u > s$.
By contrast, because the entire function $Y_{i,t-1}$ is realized before period $t$, the kernel $\nu_2(u,s)$ in $A_2(Y_{i,t-1},s)=\int_0^1Y_{i,t-1}(u)\nu_2(u,s)\text{d}u$ need not satisfy this restriction.
For other examples, if the responses at all evaluation points are equivalently influential, we may use $A_k(Y_{it}, s) = \int_0^1 Y_{it}(u) \text{d}u$.
In this case, $A_k(Y_{it}, s)$ is independent of $s$.
These examples can be represented as an integral operator $A_k(Y_{it}, s) = \int_0^1 Y_{it}(u) \nu_k(u, s)\text{d}u$ with some kernel weight function $\nu_k(u, s)$.
For example, in the case of the point-evaluation functional, we can set $\nu_k(u,s) = \delta(u - s)$, where $\delta$ denotes the Dirac delta function.

Here, we provide three empirical topics to which model \eqref{eq:model} can be effectively applied.

\begin{example}[Health data analysis] 
    In the health literature, researchers have increasingly focused on real-time activity data collected through wearable devices or smartphone apps (see, e.g., \cite{di2024utilizing} for a review).
    As a typical example, $Y_{it}(s)$ represents the activity level of individual $i$, measured by an accelerometer at time $s$ on day $t$.
    Now, suppose the sample consists of elderly individuals, some of whom frequently engage in fitness activities such as running or Tai Chi with others in the same neighborhood.
    We can apply our model to their activity-level data, with $w_{i,j}$ indicating whether individuals $i$ and $j$ live in the same neighborhood.
    In this setting, it would be natural to set $A_1(Y_{jt}, s) = Y_{jt}(s)$ and $A_2(Y_{i,t-1}, s) = Y_{i,t-1}(s)$.
\end{example}

\begin{example}[Demographic data analysis] 
    Demographic analysis is a major application of functional data analysis (FDA).
    For instance, functional analysis of regional age distributions (i.e., \textit{population pyramids}) has been studied extensively (e.g., \citealp{delicado2011dimensionality, hron2016simplicial, hoshino2024functional}).
    Among these studies, \cite{hoshino2024functional} considered a functional spatial autoregressive model in which $Y_i(s)$ represents the $s$-th age quantile of city $i$, allowing interactions with the age quantiles of neighboring cities.
    Another common application in FDA is the analysis of mortality and fertility rates (e.g., \citealp{hyndman2007robust, chen2012modeling}).
    In such a setting, the outcome function $Y_{it}(s)$ may represent the average number of births to women aged $s$ in city $i$ in year $t$.
    When one uses $A_1(Y_{jt}, s) = \int_0^1 Y_{jt}(u)\nu_1(u,s)\text{d}u$ such that $\nu_1(u,s) \neq 0$ for $u$ in the neighborhood of $s$, the presence of regional interactions in fertility within similar age groups would be detected.
\end{example}

\begin{example}[Transportation data analysis] 
    Functional data analysis of transportation data, such as traffic flows and demand for transportation services, has been gaining significant attention (see, e.g., \cite{ma2024network} for recent advancements).
    In the empirical application of this study, we apply our model to analyze the bike use data in a U.S. bike-sharing system.
    In our empirical analysis, $Y_{it}(s)$ represents the availability of bikes at station $i$ at time $s$ during week $t$.
    Further details are provided in Section \ref{sec:empiric}.
\end{example}

In model \eqref{eq:model}, the parameters of primary interest are $\alpha_0$, $\gamma_0$, and $\beta_0$.
With the total time periods $T$ possibly large or small, we apply a first-differencing transformation to eliminate the individual fixed effects from the model.
For the transformed model, rather than estimating $\alpha_0(s)$, $\gamma_0(s)$, and $\beta_0(s)$ pointwise at many different $s$-values separately, we approximate them using orthonormal basis expansions and estimate their entire functional forms jointly in a single estimation.
Our proposed estimator is based on the generalized method of moments (GMM). Specifically, we first derive a set of moment conditions at each $s$, in a similar manner to \cite{lin2010gmm}, and then integrate these conditions numerically over $s \in [0,1]$.
These integrated moment functions define our GMM objective function, and the resulting estimator is referred to as the integrated-GMM estimator.

Note that in model \eqref{eq:model}, the contemporaneous outcome functions appear on both the left- and right-hand sides, implying that it is formulated as a system of simultaneous functional equations. 
Depending on the true values of the functional parameters, the model may exhibit an explosive interaction process, leading to inconsistency of the proposed estimator.
Thus, we first derive conditions for the \textit{completeness} and stability of our model.
We consequently show that the magnitudes of dynamic and network interactions must lie within a certain range.
Then, under these conditions, along with some regularity conditions, we derive the convergence rates of the integrated-GMM estimators for $\alpha_0$, $\gamma_0$, and $\beta_0$.  
In addition, we prove that the estimators are asymptotically normal at each evaluation point $s$.  
These theoretical results are numerically corroborated through a series of Monte Carlo experiments.

As an empirical illustration, we apply our method to demand analysis for a bike-sharing system in the San Francisco Bay Area, U.S.
Using publicly available data from \textit{Bay Area Bike Share}, we study dynamic spatial interactions in bike availability across 70 stations from May 2014 to August 2015.
Our results reveal significantly positive spatial interactions in bike availability during the morning hours, while negative interactions emerge in the early evening.
Importantly, when we estimate a non-functional model \eqref{eq:PNAR} by averaging bike availability
over time, the estimated $\alpha_0$ is not statistically significant.
This highlights that ignoring the functional nature of the data may lead to misleading policy implications.
Furthermore, we conduct a network impulse response analysis to demonstrate how a reduction in bike availability at a given station propagates over time of day to nearby stations.
These findings underscore the importance of spatial interactions in shared mobility services and demonstrate the practical applicability of our method.

\bigskip

Our paper relates to a broad range of theoretical and empirical literature.  
From a theoretical perspective, our study contributes to both the FDA literature and the network/spatial interactions literature by proposing a new model that connects the two.  
In this sense, one of the most closely related studies to ours is \citet{zhu2022network}. 
They proposed a functional NAR model similar to a static version of \eqref{eq:model} for cross-sectional data. 
In contrast to \citet{zhu2022network}, our GMM estimator requires neither parametric assumptions nor IID conditions for the disturbance term.  
This weak requirement arises from the fact that we treat the individual effects as parameters, whereas \citet{zhu2022network} perform functional principal component analysis to control them based on some homogeneity condition.  
Moreover, they considered only a concurrent interaction case (i.e., $A_1(Y_{jt}, s) = Y_{jt}(s)$).  
As described earlier, the variable $s$ typically represents a time on some scale.\footnote{
    In this case, both $s$ and $t$ represent time, but they capture different time scales.
    Typically, $s$ is used to describe short-term cycles, such as intraday stock price movements, or daily variation in temperature.
    In contrast, $t$ captures more macro-level trends and thus corresponds to a longer time scale.
}
The concurrent interaction rules out interactions even with immediate past outcomes and allows only strictly simultaneous interactions, which should limit the interpretability of the model.  
Computationally, our estimator can recover the full functional forms of the functional parameters in a single step, while the estimator in \citet{zhu2022network} must be repeatedly applied at each evaluation point $s$.

On the empirical side, demand forecasting for bike-sharing systems has been an active topic in the data science literature (e.g., \citealp{faghih2016incorporating, lin2018predicting, eren2020review, torti2021modelling}, among others).
Among these studies, \cite{faghih2016incorporating} is most closely related to our study in that they employed a spatial panel model similar to \eqref{eq:PNAR} to analyze the spatial and temporal interaction structure for the bike-sharing system in New York City, CitiBike. 
In their approach, however, the data are not treated as functional, and thus the model parameters are not allowed to vary over the time of day.
By contrast, \cite{torti2021modelling} analyzed the flow of bikes in the bike-sharing system in Milan, BikeMi, through a functional linear model with functional coefficients; however, they did not account for the spatial interactions of mobility.
Thus, our empirical study can be viewed as combining the strengths of these two papers.

\paragraph{Paper organization}

The rest of the paper is organized as follows.
In Section \ref{sec:model}, we present our DNAR model and discuss its completeness and stability conditions.
Section \ref{sec:estimation} introduces our integrated-GMM estimator and investigates its asymptotic properties.
In Section \ref{sec:impulse}, we discuss additional topics related to our model, including the estimation of marginal effects and network impulse response analysis.
Section \ref{sec:MC} conducts a set of Monte Carlo simulations to numerically demonstrate the properties of our estimator.
Section \ref{sec:extension} discusses several extensions that should be useful in empirical applications.
Section \ref{sec:empiric} presents our empirical analysis on the U.S. bike-sharing data, and Section \ref{sec:conclusion} concludes.
Proofs of technical results, additional empirical material, and the programming code used in the numerical studies are provided in the online supplementary material.

\paragraph{Notation}

For a function $h$ defined on $[0,1]$ and $p \in [1,\infty)$, the $L^p$ norm of $h$ is written as $||h||_{L^p} \coloneqq (\int_0^1 |h(s)|^p \text{d}s)^{1/p}$, and $L^p(0,1)$ denotes the set of $h$'s such that $||h||_{L^p} < \infty$.
For a random variable $X$, the $L^p$ norm of $X$ is written as $||X||_p \coloneqq (\bbE|X|^p)^{1/p}$.
For a matrix $M$, $|| M ||$, $||M||_1$, and $||M||_\infty$ denote the Frobenius norm, the maximum absolute column sum, and the maximum absolute row sum of $M$, respectively.
If $M$ is a symmetric matrix, we use $\lambda_{\max} (M)$ and $\lambda_{\min} (M)$ to denote its largest and smallest eigenvalues, respectively.
For a positive integer $Z$, we denote $[Z] \coloneqq \{1, \ldots, Z\}$.
We use $I_Z$ to denote an identity matrix of dimension $Z$.
Finally, $X \lesssim Y$ if $X = O(Y)$ almost surely, and $X \lesssim_P Y$ if $X = O_P(Y)$.

\section{Functional Dynamic Network Autoregressive Model}\label{sec:model}
\subsection{The model}

Suppose that we have balanced panel data of size $(n, T)$: $\{(Y_{i0}, Y_{it}, X_{i0}, X_{it}, w_{i,1}, \ldots, w_{i,n}): i \in [n], \: t \in [T]\}$.
The number of time periods $T \ge 2$ can be either fixed or tending to infinity jointly with the sample size $n$.
Here, $Y_{it}: [0,1] \to \mathbb{R}$ denotes a random outcome function of interest with the common support $[0,1]$, $X_{it} = (X_{it}^1, \ldots, X_{it}^{d_x})^\top$ denotes a vector of covariates, and $w_{i,j} \in \mathbb{R}$ is the $(i,j)$-th element of an $n \times n$ time-invariant interaction matrix $W_n = (w_{i,j})$.
Throughout, we assume that the functional outcome $Y_{it}(s)$ is fully observed over $s \in [0,1]$.
In our empirical application, $Y_{it}(s)$ represents the number of available bikes at station $i$, where $s$ denotes the time of day.
The data record the timing of all pickups and returns at each station, which implies the complete observability of $Y_{it}$ as a function of $s$.
As in the case of bike availability, we do not rule out situations in which $Y_{it}(s)$ is discontinuous in $s$.
The covariates $X_{it}$ do not include a constant term and vary with both $i$ and $t$.
The value of each $w_{i,j}$ is pre-determined non-randomly.
In social network analysis, it is common to set $w_{i,j} = c_{i,j} \bm{1}\{\text{$i$ and $j$ are peers}\}$, where $c_{i,j}$ is some normalizing constant.
Similarly, if each $i$ represents a spatial unit, one may use $w_{i,j} = c_{i,j} \bm{1}\{\Delta(i,j) \le \bar \Delta\}$, where $\Delta(i,j)$ is the distance between $i$ and $j$, and $\bar \Delta$ is a given threshold. 
As is the convention, we set $w_{i,i} = 0$ for all $i$ for normalization.

As shown in \eqref{eq:model}, our working model is given as follows: for $s \in [0,1]$,
\begin{align}
    Y_{it}(s) = \alpha_0(s) A_1(\bar Y_{it}, s) + \gamma_0(s) A_2(Y_{i,t-1}, s) + X_{it}^\top \beta_0(s) + f_{0i}(s) + \varepsilon_{it}(s), \;\; i \in [n], \;\; t \in [T]
\end{align}
where $\bar Y_{it} = \sum_{j=1}^n w_{i,j} Y_{jt}$.
Considering analytical transparency and interpretational simplicity in empirical applications, for $k \in \{1,2\}$, we assume $A_k(\cdot, \cdot)$ is linear in its first argument so that we have $\sum_{j=1}^n w_{i,j} A_1(Y_{jt}, \cdot) = A_1(\bar Y_{it}, \cdot)$.
While this linear specification provides a reasonable starting point, more complicated forms are discussed in Subsection \ref{subsec:nonlinear}. 

The parameters of primary interest are the interaction effect functions $\alpha_0(s)$ and $\gamma_0(s)$, and the coefficient functions $\beta_0 (s)= (\beta_{01}(s), \ldots, \beta_{0d_x}(s))^\top$.
The functional individual effects $f_{01}(s), \ldots, f_{0n}(s)$ are treated as nuisance parameters.
Restricting the support of $s$ to be a unit interval is a normalization, which is harmless as long as the response functions have the same interval support.
For simplicity, we do not explicitly assume that $X_{it}$ is a function of $s$, which can be relaxed easily at the expense of more complicated notation and proofs.
In the rest of this section, we assume that $Y_{it} \in L^2(0,1)$, $\varepsilon_{it} \in L^2(0,1)$, and that $\alpha_0$, $\gamma_0$, and $\beta_0$ are continuous functions.
To derive the asymptotic properties of our estimator, we impose stronger regularity conditions later.

\subsection{Completeness and stability}

We discuss the \textit{completeness} of our model (in the sense of \citealp{tamer2003incomplete, lewbel2007coherency}) and the stability of the dynamic process.
Since the model is a system of simultaneous functional equations, it may not have a unique interior solution in general, depending on the parameter values.
Moreover, even when the simultaneous equation system has a unique solution at each $t$, the dynamic process may exhibit explosive behavior.
Thus, we need to impose conditions ensuring both that the outcome functions follow a unique data-generating process and that the dynamic process is stable.

Let $Y_t(s) = (Y_{1t}(s), \ldots, Y_{nt}(s))^\top$, $\bm{A}_k(Y_t, s) = (A_k(Y_{1t}, s), \ldots, A_k(Y_{nt}, s))^\top$, $X_t = (X_{1t}, \ldots, X_{nt})^\top$, $F_0(s) = (f_{01}(s), \ldots, f_{0n}(s))^\top$, and $\mathcal{E}_t(s) = (\varepsilon_{1t}(s), \ldots, \varepsilon_{nt}(s))^\top$.
Then, we can re-write \eqref{eq:model} in matrix form as
\begin{align}
    Y_t(s) = \alpha_0(s) W_n \bm{A}_1(Y_t, s) + \gamma_0(s) \bm{A}_2(Y_{t-1}, s) + X_t\beta_0(s) + F_0(s) + \mathcal{E}_t(s), \;\; t \in [T].
\end{align}
Now, denote $\mathcal{H}_{n,p} \coloneqq \{H = (h_1, \ldots, h_n) : h_i \in L^p(0,1), \; i \in [n]\}$, and define linear operators $\mathcal{A}_1$ and $\mathcal{A}_2$ as
\begin{align}
    (\mathcal{A}_1 H)(s) \coloneqq \alpha_0(s) W_n \bm{A}_1(H, s) \;\; \text{and} \;\; (\mathcal{A}_2 H)(s) \coloneqq \gamma_0(s) \bm{A}_2(H, s),
\end{align}
respectively, for $H \in \mathcal{H}_{n,p}$.
Then, we can write our model symbolically as follows: $Y_t = \mathcal{A}_1 Y_t + \mathcal{A}_2 Y_{t-1} + X_t\beta_0 + F_0 + \mathcal{E}_t$.
Further, let $\text{Id}$ denote the identity operator.
If the inverse operator $(\text{Id}-\mathcal{A}_1)^{-1}$ exists, define $\mathcal{A}_3\coloneqq(\text{Id}-\mathcal{A}_1)^{-1}\mathcal{A}_2$.
Then, the solution $Y_t$ of the system can be uniquely determined up to an equivalence class in $\mathcal{H}_{n,2}$ as
\begin{align}
    Y_t
    & = \mathcal{A}_3 Y_{t-1} + (\text{Id} - \mathcal{A}_1)^{-1}[X_t\beta_0 + F_0 + \mathcal{E}_t] \\
    & = \mathcal{A}_3^t Y_0 + \sum_{\ell = 0}^{t-1} \mathcal{A}_3^\ell (\text{Id} - \mathcal{A}_1)^{-1}[X_{t-\ell}\beta_0 + F_0 + \mathcal{E}_{t-\ell}].
\end{align}
From this expression, we can see that the existence of $(\text{Id}-\mathcal{A}_1)^{-1}$ ensures completeness, while stability is ensured if $\mathcal{A}_3$ is a contraction mapping.

\begin{assumption}[Completeness and stability]\label{as:inverse}
    (i) $\bar \alpha_0 \lesssim 1$, $\bar \gamma_0 \lesssim 1$, and $||W_n||_\infty \lesssim 1$ uniformly in $n \ge 1$ such that $\bar\alpha_0 \|W_n\|_\infty + \bar\gamma_0 < 1$, where $\bar \alpha_0 \coloneqq \max_{s \in [0,1]} |\alpha_0(s)|$, and $\bar \gamma_0 \coloneqq \max_{s \in [0,1]} |\gamma_0(s)|$.
    (ii) For any $h \in L^2(0,1)$ and $k \in \{1,2\}$, $||A_k(h,\cdot)||_{L^2} \le ||h||_{L^2}$.
\end{assumption}

Assumption \ref{as:inverse}(i) requires that the magnitudes of network and dynamic interactions are not too strong.
With Assumption \ref{as:inverse}(ii), we have for any $h,h' \in L^2(0,1)$ that $||A_k(h - h',\cdot)||_{L^2} \le ||h - h'||_{L^2}$, which indicates the non-expansive property of the operator $A_k$. 
This assumption still accommodates many empirically interesting interaction patterns.
For example, in the case of the point-evaluation functional $A_k(h,s) = h(s)$, it trivially satisfies $||A_k(h,\cdot)||_{L^2} = ||h||_{L^2}$.
For another example, suppose $A_k(h,s) = \int_0^1 h(u) \nu_k(u,s) \text{d}u$ for some continuous $\nu_k$.
Since
\begin{align}
    ||A_k(h,\cdot)||_{L^2}^2 \le  \int_0^1 \int_0^1 \left| h(u) \nu_k(u,s) \right|^2 \text{d}u \text{d}s \le \bar \nu_k^2 ||h||_{L^2}^2,
\end{align}
where $\bar \nu_k \coloneqq \max_{u,s \in [0,1]^2} |\nu_k(u,s)|$, Assumption \ref{as:inverse}(ii) is implied if $\bar \nu_k \le 1$ holds.

\begin{proposition}\label{prop:stationarity}
    Suppose that Assumption \ref{as:inverse} holds.
    Then, $(\text{Id} - \mathcal A_1)^{-1}$ exists, $\mathcal A_3$ is a contraction mapping on $\mathcal H_{n,2}$, and model \eqref{eq:model} admits a unique solution sequence $\{Y_t: t \in [T]\}$ in the Banach space $(\mathcal H_{n,2}, \|\cdot\|_{\infty,2})$, where $\|H\|_{\infty, p} \coloneqq \max_{1 \le i \le n}\|h_i\|_{L^p}$.
\end{proposition}

The proof of Proposition \ref{prop:stationarity} is given in Appendix \ref{app:proofs}. 
Note that the explicit form of the inverse operator $(\text{Id} - \mathcal{A}_1)^{-1}$ cannot be derived in general, except for some simple cases such as $(\mathcal{A}_1 Y_t)(s) = \alpha_0(s) W_n Y_t(s)$.
In this case, $(\text{Id} - \mathcal{A}_1)^{-1}$ is obtained as $(I_n - \alpha_0(\cdot) W_n)^{-1}$.
However, in practice, we can approximate it with arbitrary precision by truncating the Neumann series expansion $(\text{Id} - \mathcal{A}_1)^{-1} = \sum_{\ell = 0}^\infty \mathcal{A}_1^\ell$ at a sufficiently large order (see, e.g., \citealp{Kress2014linear}).

Finally, we note that Assumption \ref{as:inverse} is only one type of sufficient condition for the existence of $(\text{Id} - \mathcal{A}_1)^{-1}$.
Condition (i) can be relaxed by allowing $\alpha_0(s)$ to exhibit locally explosive network dependence for some $s \in [0,1]$, provided that the operator $A_1$ satisfies a condition more restrictive than condition (ii); see Subsection \ref{subsub:alt} in Appendix \ref{app:proofs} for details.
Accordingly, unlike standard non-functional NAR models, the completeness of our model can still be satisfied even when the network effect $\alpha_0(s)$ takes locally extreme values.
That said, we adopt the current form of Assumption \ref{as:inverse} for technical convenience in the subsequent analysis.

\section{Estimation and Asymptotic Theory}\label{sec:estimation}

\subsection{Integrated-GMM estimation}\label{subsec:gmm}

To estimate the unknown functional parameters, there are broadly two approaches.
The first is a "local" approach that estimates the values of these functions at specific $s$-values, repeating the estimation across different points to recover the full functional forms.
The second is a "global" approach that estimates the entire functional forms in a single step using a series approximation method.
Although both approaches are theoretically valid, the local approach typically requires more computation time and often leads to larger variance (but smaller bias) because it does not exploit information from nearby evaluation points.
This study adopts the global approach.

Let $\{\phi_k: k = 1,2, \ldots\}$ be a series of orthonormal basis functions.
Throughout, we assume that $\phi_k$'s are continuous on $[0,1]$.
Then, if the functions $\alpha_0$, $\gamma_0$ and $\beta_0$ are sufficiently smooth, we can approximate 
\begin{align}
    \alpha_0(s) \approx \sum_{k = 1}^K \phi_k(s) \theta_{0\alpha, k}, \;\; 
    \gamma_0(s) \approx \sum_{k = 1}^K \phi_k(s) \theta_{0\gamma, k}, \;\; \text{and} \;\;
    \beta_{0j}(s) \approx \sum_{k = 1}^K \phi_k(s) \theta_{0j, k}, \;\; j \in [d_x],
\end{align}
uniformly in $s \in [0,1]$, for some coefficient vectors $\theta_{0\alpha} = (\theta_{0\alpha, 1}, \ldots, \theta_{0\alpha, K})^\top$, $\theta_{0\gamma} = (\theta_{0\gamma, 1}, \ldots, \theta_{0\gamma, K})^\top$, and $\theta_{0j} = (\theta_{0j, 1}, \ldots, \theta_{0j, K})^\top$, $j \in [d_x]$.
Here, $K \equiv K_{nT}$ is a sequence of positive integers tending to infinity as $nT$ increases.
For simplicity of presentation, we use the same basis function $\phi_k$ and the same basis order $K$ to approximate $\alpha_0$, $\gamma_0$, and $\beta_0$.
Define $\theta_0 = (\theta_{0\alpha}^\top, \theta_{0\gamma}^\top, \theta_{01}^\top, \ldots, \theta_{0d_x}^\top)^\top$, $\phi^K(s) = (\phi_1(s), \ldots, \phi_K(s))^\top$,
\begin{align}
    R_{it}(s) \coloneqq (A_1(\bar Y_{it}, s), A_2(Y_{i,t-1}, s), X_{it}^\top)^\top,  \;\;
    H_{it}(s) \coloneqq R_{it}(s) \otimes \phi^K(s), \;\; \text{and} \;\; 
    H_t(s) = (H_{1t}(s), \ldots, H_{nt}(s))^\top.
\end{align}
Then, we can further re-write the model in \eqref{eq:model} as
\begin{align}
    Y_t(s) = H_t(s) \theta_0 + F_0(s) + V_t(s) + \mathcal{E}_t(s), \;\; t \in [T].
\end{align}
Here, $V_t(s) = (v_{1t}(s), \ldots, v_{nt}(s))^\top$ is an $n \times 1$ vector of series approximation errors:
\begin{align}\label{eq:approxer}
    v_{it}(s)
    & \coloneqq A_1(\bar Y_{it}, s)\{\alpha_0(s) - \phi^K(s)^\top \theta_{0\alpha} \} + A_2(Y_{i,t-1}, s)\{\gamma_0(s) - \phi^K(s)^\top \theta_{0\gamma} \} \\
    & \quad + \sum_{j = 1}^{d_x} X_{it}^j \{\beta_{0j}(s) - \phi^K(s)^\top \theta_{0j}\}. 
\end{align}
Under the assumptions we will introduce, this approximation error diminishes to zero at a certain rate as $K$ goes to infinity. 
How to choose an appropriate $K$ will be discussed in Remark \ref{rem:K}. 

Further, let
\begin{align}
    N \coloneqq n(T - 1),
\end{align}
\begin{align}
    \bm{Y}(s) = \left(\begin{array}{c}
        Y_1(s) \\
        \vdots \\
        Y_T(s)
    \end{array}\right), \quad
    \bm{H}(s) = \left(\begin{array}{c}
        H_1(s) \\
        \vdots \\
        H_T(s)
    \end{array}\right), \quad
        \bm{V}(s) = \left(\begin{array}{c}
        V_1(s)  \\
        \vdots \\
        V_T(s) 
    \end{array}\right), \quad
        \bm{\mathcal{E}}(s) = \left(\begin{array}{c}
        \mathcal{E}_1(s)  \\
        \vdots \\
        \mathcal{E}_T(s) 
    \end{array}\right),
\end{align}
and $\underset{N \times nT}{\bm{D}} = (d_{ij})$ be the one-period lag operator, whose $(i,j)$-th element is defined as
\begin{align}
    d_{ij} = \begin{cases}
    -1 & \text{if} \;\; i = j \\
    1 & \text{if} \;\;  n + i = j \\
    0 & \text{otherwise}
    \end{cases}
\end{align}
Then, we can remove the unknown fixed effects from the model in the following manner:
\begin{align}
    \bm{D} \bm{Y}(s)
    & = \bm{D} \bm{H}(s) \theta_0 + \bm{D} \bm{V}(s) + \bm{D} \bm{\mathcal{E}}(s).
\end{align}
We estimate $\theta_0$ based on this expression.
In order to consistently estimate $\theta_0$, we need to address the endogeneity issue caused by the network and dynamic interactions of the response functions.
Specifically, $A_1(\bar Y_{it},s)$ is endogenous due to simultaneity, while one-period differencing makes $A_2(Y_{i,t-1},s)$ endogenous through the correlation with $\varepsilon_{i,t-1}(s)$.
Hence, simply regressing $\bm{D} \bm{Y}(s)$ on $\bm{D} \bm{H}(s)$ does not yield a consistent estimate of $\theta_0$.
To tackle this issue, we employ an instrumental variable (IV) approach.

Suppose that we have a $d_q \times 1$ vector of IVs $Q_{it} = (Q_{it}^1, \ldots, Q_{it}^{d_q})^\top$ for $A_1(\bar Y_{it}, s)$ and $A_2(Y_{i,t-1}, s)$.
Define
\begin{align}
    B_{it} \coloneqq (Q_{it}^\top, X_{it}^\top)^\top,  \;\;
    Z_{it}(s) \coloneqq B_{it} \otimes \phi^K(s), \;\; 
    Z_t(s) = (Z_{1t}(s), \ldots, Z_{nt}(s))^\top, \;\; \text{and} \;\; 
    \bm{Z}(s) = \left(\begin{array}{c}
        Z_1(s)  \\
        \vdots \\
        Z_T(s) 
    \end{array}\right).
\end{align}
Then, we have the \textit{linear} moment conditions 
\begin{align}
    \bbE[ \bm{Z}(s)^\top \bm{D}^\top \bm{D}  \bm{\mathcal{E}}(s)] = \bm{0}_{(d_q + d_x)K}.
\end{align}
Noting that $W_n Y_t = W_n \mathcal{A}_1 Y_t + W_n \mathcal{A}_2 Y_{t-1} + W_n X_t \beta_0 + W_n F_0 + W_n \mathcal{E}_t$, we can find that the network lagged covariates $\bar X_{it} \coloneqq \sum_{j = 1}^n w_{i,j} X_{jt}$ (and also their lags) are valid IV candidates for $A_1(\bar Y_{it}, s)$. 
However, when $X_{it}$ is locally homogeneous, $\bar X_{it}$ and $X_{it}$ become strongly correlated, which leads to a weak IV problem.
Moreover, even when they are not strongly correlated, if the impact of $X_{it}$ on $Y_{it}$ is small (i.e., $\beta_0 \approx 0$), a weak IV problem may still occur.
As IV candidates for $A_2(Y_{i,t-1},s)$, we can employ the lagged covariates $X_{i,t-1}$.

Although one can estimate $\theta_0$ based on the linear moment conditions only, which results in a two-stage least squares (2SLS) type estimator, we can utilize additionally the \textit{quadratic} moment conditions to improve the efficiency of estimation (see, e.g., \citealp{lin2010gmm}).
That is, under the independence assumption on the error terms $\{\varepsilon_{it}(s)\}_{i \in [n], t \in [T]}$ (Assumption \ref{as:error}(i) below), for any $N \times N$ matrices $P_m \coloneqq I_{T-1} \otimes P_{m,1}$, where $P_{m,1}$ ($m = 1, \ldots , M$) is an $n \times n$ matrix whose diagonal elements are all zero, we have\footnote{
    Throughout the paper, we assume that $M$ is fixed.
    However, as pointed out by a referee, it would be interesting to consider an increasing number of quadratic moment conditions by setting, for example, $P_{m,1} = W_n^m - \text{diag}(W_n^m)$.
    The quadratic moments constructed in this way are strongly correlated with each other, and the contribution of higher-order quadratic moments may be marginal.
    This setting corresponds to GMM with many potentially weak moments (e.g., \citealp{han2006gmm}), for which a separate and careful discussion is required.
    Thus, we leave this issue for future research.
}
\begin{align}
   \bbE[ \bm{\mathcal{E}}(s)^\top \bm{D}^\top P_m \bm{D} \bm{\mathcal{E}}(s) ] = 0, \;\; m \in [M].
\end{align}
Some examples of $P_{m,1}$ include $P_{m,1} = W_n$ and $P_{m,1} = W_n^\top W_n - \text{diag}(W_n^\top W_n)$.
The use of quadratic moment conditions is particularly important when the weakness of the IV $\bm Z(s)$ is a concern.
As discussed in \cite{yang2025estimation} and also demonstrated in our numerical results, quadratic moment conditions can help mitigate weak identification when $\bm Z(s)$ are not strong IVs.

Combining the linear and quadratic moment conditions, we can construct our estimator based on the following $d_g \coloneqq (d_q + d_x)K + M$ moment conditions: for $s \in [0,1]$,
\begin{align}
    \frac{1}{N} \left(\begin{array}{c}
        \bbE[ \bm{Z}(s)^\top \bm{D}^\top \bm{D}  \bm{\mathcal{E}}(s)] \\
        \bbE[ \bm{\mathcal{E}}(s)^\top \bm{D}^\top P_1 \bm{D} \bm{\mathcal{E}}(s)] \\
        \vdots \\
        \bbE[ \bm{\mathcal{E}}(s)^\top \bm{D}^\top P_M \bm{D} \bm{\mathcal{E}}(s)]
    \end{array}
    \right) = \bm{0}_{d_g}.
\end{align}
As the empirical counterpart of these moment conditions, given a candidate value $\theta$ for $\theta_0$, we define
\begin{align}
    \underset{d_g \times 1}{g_{nT}(s; \theta)} \coloneqq
    \frac{1}{N}\left(\begin{array}{c}
        \bm{Z}(s)^\top \bm{D}^\top \bm{D} \bm{E}(s; \theta) \\
        \bm{E}(s; \theta)^\top \bm{D}^\top P_1 \bm{D} \bm{E}(s; \theta) \\
        \vdots \\
        \bm{E}(s; \theta)^\top \bm{D}^\top P_M \bm{D} \bm{E}(s; \theta)
    \end{array}
    \right),
\end{align}
where $\bm{E}(s; \theta) \coloneqq \bm{Y}(s) - \bm{H}(s)\theta$.
Since we have a set of moment conditions that holds continuously over $s \in [0,1]$, to construct our GMM estimator, we numerically integrate these conditions.
Specifically, for a set of pre-specified $L \equiv L_{nT} \ge K$ grid points in $[0,1]$, denoted by $0 \le s_1 \le \cdots \le s_L \le 1$, we define
\begin{align}
    \bar{g}_{nT}(\theta) \coloneqq \frac{1}{L} \sum_{l=1}^L g_{nT}(s_l; \theta).
\end{align}
Remark \ref{rem:grid} below provides practical guidance on how to choose the grid $\{s_l\}$.
Situations where the response functions are not fully observable are discussed in Subsection \ref{sub:incomplete}.

Now, we are ready to introduce our estimator:
\begin{align}\label{eq:gmm}
    \begin{split}
    \hat \theta_{nT} = (\hat \theta_{nT, \alpha}^\top, \hat \theta_{nT, \gamma}^\top, \hat \theta_{nT, 1}^\top, \ldots, \hat \theta_{nT, d_x}^\top)^\top 
    & \coloneqq \argmin_{\theta \in \Theta_K} \mathcal{Q}_{nT}(\theta),  \\
    \text{where} \;\; \mathcal{Q}_{nT}(\theta)
    & \coloneqq \bar g_{nT}(\theta)^\top \Omega_{nT} \bar g_{nT}(\theta),
    \end{split}
\end{align}
$\Omega_{nT}$ is a $d_g \times d_g$ positive definite symmetric weight matrix.
For each $K \ge  1$, $\Theta_K \subset \mathbb{R}^{(d_x + 2)K}$ is a compact parameter space containing $\theta_0$ in its interior.
Moreover, we assume that $\phi^K(s)^\top \theta_\alpha$, $\phi^K(s)^\top \theta_\gamma$, and $\phi^K(s)^\top \theta_j$ are uniformly bounded over $s \in [0,1]$, $\theta = (\theta_\alpha^\top, \theta_\gamma^\top, \theta_1^\top, \ldots, \theta_{d_x}^\top)^\top \in \Theta_K$, and $K$.
For one example of the weight matrix, we can use $\Omega_{nT} = I_{d_g}$. 
Another example is the 2SLS-type weight matrix:
\begin{align}\label{eq:2slsweight}
    \Omega_{nT} = \left( \begin{array}{cc}
        \left( \sum_{l = 1}^L \bm{Z}(s_l)^\top \bm{D}^\top \bm{D} \bm{Z}(s_l) / (NL)\right)^{-1} & \bm{0}_{(d_q + d_x)K \times M} \\
        \bm{0}_{M \times (d_q + d_x)K} & I_M
    \end{array} \right).
\end{align}
Once $\hat \theta_{nT}$ is obtained, the estimators of $\alpha_0(s)$, $\gamma_0(s)$, and $\beta_0(s)$ are given as
\begin{align}
    \hat \alpha_{nT}(s) \coloneqq \phi^K(s)^\top \hat \theta_{nT, \alpha}, \;\; 
    \hat \gamma_{nT}(s) \coloneqq \phi^K(s)^\top \hat \theta_{nT, \gamma}, \;\; \text{and} \;\;
    \hat \beta_{nT,j}(s)
    \coloneqq \phi^K(s)^\top \hat \theta_{nT, j}, \;\; j \in [d_x]
\end{align}
which we refer to as the integrated-GMM estimators.  

\subsection{Asymptotic theory}

To derive the asymptotic properties of our estimator, we first need to specify the structure of our sampling space.
Following \cite{jenish2012spatial}, let $\mathcal{D} \subset \bbR^d$ be a possibly uneven lattice, and $\mathcal{D}_n \subset \mathcal{D}$ be the set of observation locations. 
Once the observation locations are determined for a given sample of size $n$, we assume that they do not vary over time $t$.
For spatial data, $\mathcal{D}$ would be defined by a geographical space with $d = 2$.\footnote{
    Note that $\mathcal{D}$ does not necessarily have to be exactly observable to us.
    For example, $\mathcal{D}$ is possibly a complex space of general social and economic characteristics. In this case, we can consider it to be an embedding of individuals in a latent space, rather than their physical locations.
    }

We first derive the rates of convergence of our estimator under the following set of assumptions.

\begin{assumption}[Sampling space]\label{as:sample_space}
    (i) The maximum coordinate difference (i.e., the Chebyshev distance) between any two observations $i,j \in \mathcal{D}$, which we denote as $\Delta(i,j)$, is at least (without loss of generality) 1; and
    (ii) a threshold distance $\bar \Delta$ exists such that $w_{i,j} = 0$ if $\Delta(i,j) > \bar \Delta$.
\end{assumption}

\begin{assumption}[Observables]\label{as:observables}
    (i) $\{(X_{i0}, Y_{i0})\}_{i \in [n]}$ and $\{(X_{it}, Q_{it})\}_{i \in [n], t \in [T]}$ are treated as non-stochastic and uniformly bounded; and
    (ii) for all $s \in [0,1]$, $i \in [n]$, and $t \in [T]$, $||Y_{it}(s)||_p \lesssim 1$ for some $p > 4$.
\end{assumption}

\begin{assumption}[Error term]\label{as:error}
    (i) $\{\varepsilon_{it}\}_{i \in [n],  t \in [T]}$ are independent; 
    (ii) for all $s \in [0,1]$, $i \in [n]$, and $t \in [T]$, $\bbE [\varepsilon_{it}(s)] = 0$, $||\varepsilon_{it}(s)||_2 > 0$, and $||\varepsilon_{it}(s)||_4 \lesssim 1$; and
    (iii) for all $i \in [n]$ and $t \in [T]$, $\sum_{k=1}^K \left( L^{-2}\sum_{l = 1}^L \sum_{l' = 1}^L \Gamma_{it}(s_l, s_{l'}) \phi_k(s_l) \phi_k(s_{l'})\right) \lesssim 1$ uniformly in $K$, where $\Gamma_{it}(s_l, s_{l'}) \coloneqq \text{Cov}(\varepsilon_{it}(s_l), \varepsilon_{it}(s_{l'}))$.
\end{assumption}

\begin{assumption}[Interaction operators]\label{as:A}
    For each $k \in \{1,2\}$, there exists a function $\omega_p$ satisfying $|A_k(h, s)|^p \le \int_0^1 |h(u)|^p \: \omega_p(u, s)\text{d}u$ for any given $1 \le p < \infty$, such that $\int_0^1 \omega_p(u, s)\text{d}u \le 1$ for all $s \in [0,1]$. 
\end{assumption}

\begin{assumption}[Weight matrices]\label{as:weights}
    (i) For all $m \in [M]$, $P_{m,1}$ is symmetric, $\text{diag}(P_{m,1}) = \bm{0}_n$, and $||P_{m,1}||_1 \: (= ||P_{m,1}||_\infty) \lesssim 1$.
    In addition, writing $P_{m,1} = (p_{m,i,j})$, a threshold distance $\bar \Delta_m$ exists such that $p_{m,i,j} = 0$ if $\Delta(i,j) > \bar \Delta_m$; and
    (ii) the GMM weight matrix, including the special case of the 2SLS-type weight \eqref{eq:2slsweight}, satisfies $0 < \liminf_{nT \to \infty} \lambda_{\min}(\Omega_{nT}) \le \limsup_{nT \to \infty} \lambda_{\max}(\Omega_{nT}) < \infty$.
\end{assumption}

\begin{assumption}[Identification]\label{as:matrix1}
    $0 < \liminf_{nT \to \infty} \lambda_{\min}\left(\Pi_{nT}^\top \Pi_{nT}\right) \le \limsup_{nT \to \infty} \lambda_{\max}\left(\Pi_{nT}^\top \Pi_{nT}\right) < \infty$ uniformly in $K$, where $\Pi_{nT} \coloneqq (NL)^{-1} \sum_{l = 1}^L \bm{Z}(s_l)^\top \bm{D}^\top \bm{D} \bbE[\bm{H}(s_l)]$.
\end{assumption}

\begin{assumption}[Series approximation]\label{as:basis}
    $\{\phi_k: k = 1,2, \ldots\}$ is a series of continuous orthonormal basis functions satisfying $\sup_{s \in [0,1]} \|\phi^K(s)\| \lesssim \sqrt{K}$, $\sup_{s \in [0,1]} |\alpha_0(s) - \phi^K(s)^\top \theta_{0\alpha}| \lesssim K^{-\varsigma}$, $\sup_{s \in [0,1]} |\gamma_0(s) - \phi^K(s)^\top \theta_{0\gamma}| \lesssim K^{-\varsigma}$, and $\max_{j \in [d_x]} \sup_{s \in [0,1]}|\beta_{0j}(s) - \phi^K(s)^\top \theta_{0j}| \lesssim K^{-\varsigma}$.
\end{assumption}

Assumptions \ref{as:sample_space}(i) and (ii) together imply that the number of interacting partners for each unit is bounded (i.e., the network must be sparse).
These assumptions play a crucial role in characterizing the stochastic process of the outcome functions.
In Assumption \ref{as:observables}, part (i) assumes that the initial outcomes and covariates are non-stochastic and bounded.
This type of assumption is frequently utilized in the spatial and network literature and can be interpreted as viewing the analysis conditional on the realized values of these variables.
Meanwhile, part (ii) is introduced to ensure some convergence results for the quadratic moments.

Assumption \ref{as:error}(i) allows the error terms to be fully heteroskedastic.
Part (ii) should be standard.
Part (iii) is a high-level condition, which plays an important role in obtaining the parametric convergence rate for the GMM estimator. 
This assumption essentially requires the covariance operator $\Gamma_{it}$ to be \textit{trace class} (see, e.g., \citealp{hsing2015theoretical}).
More specifically, if $\Gamma_{it}$ belongs to $L^2([0,1]^2)$, it admits the following series expansion:
\begin{align}
    \Gamma_{it}(s, s') = \sum_{k_1, k_2 = 1}^\infty \kappa_{it,k_1,k_2} \phi_{k_1}(s) \phi_{k_2}(s')
\end{align}
for some sequence of constants $\{ \kappa_{it,k_1,k_2}\}$.
By the orthonormality of $\phi_k$,
\begin{align}
    \int_0^1 \int_0^1 \Gamma_{it}(s, s') \phi_k(s) \phi_k(s') \text{d}s \text{d}s'
    &= \sum_{k_1, k_2 = 1}^\infty \kappa_{it,k_1,k_2} \left( \int_0^1 \phi_k(s) \phi_{k_1}(s) \text{d}s \right) \left( \int_0^1 \phi_k(s') \phi_{k_2}(s') \text{d}s' \right) \\
    &= \kappa_{it,k,k}.
\end{align}
Since $L^{-2} \sum_{l=1}^L \sum_{l'=1}^L \Gamma_{it}(s_l, s_{l'}) \phi_k(s_l) \phi_k(s_{l'})$ can be seen as a numerical approximation of the left-hand side of the above expression, Assumption \ref{as:error}(iii) essentially requires that $\sum_{k=1}^K \kappa_{it,k,k} \lesssim 1$ uniformly in $K$.
This type of assumption is commonly made in the FDA literature.

Assumption \ref{as:A} is not restrictive in most empirically relevant situations. 
For example, in the case where $A_k(h,s) = h(s)$, we can set $\omega_p(u,s) = \delta(u - s)$ for any $p$.
For another example, when $A_k(h,s) = \int_0^1 h(u) \nu_k(u, s) \text{d}u$ for some kernel $\nu_k(u, s)$, since $|A_k(h,s)|^p \le \int_0^1 |h(u)|^p |\nu_k(u, s)|^p \text{d}u$, we can set $|\nu_k(u, s)|^p = \omega_p(u,s)$ in this case.

In Assumption \ref{as:weights}, we assume that the weight matrices in the quadratic moments are symmetric. 
Note that this assumption does not lose any generality because $A^\top P_{m,1} A = A^\top P_{m,1}^\top A$ for any $n \times 1$ vector $A$.
If $P_{m,1}$ is not symmetric in practice, we can always symmetrize it as $(P_{m,1} + P_{m,1}^\top)/2$.
The assumption for the existence of a threshold distance $\bar \Delta_m$ may be non-standard, but it simplifies the proof.
Since $P_{m,1}$'s are usually created from the interaction matrix $W_n$ and its powers, this assumption is consistent with Assumption \ref{as:sample_space}(ii).

Assumption \ref{as:matrix1} is a regularity condition to ensure the identifiability of $\theta_0$.
Assumption \ref{as:basis} is standard in the series estimation literature.
As shown in Theorem \ref{thm:normality}, we require $\varsigma > 3/2$ in order to establish the asymptotic normality of the estimator.
This condition is satisfied, for example, when spline basis functions are used and the functions $\alpha_0$, $\gamma_0$, and $\beta_{0j}$ are twice continuously differentiable (see, e.g., \citealp{CHEN20075549, belloni2015some}).

\begin{remark}[Choice of the locations and the number of grid points]\label{rem:grid}
    The choice of the grid $\{s_l\}$ determines the locations of $s$ at which the continuum of moment conditions is evaluated, which affects both the identifiability of parameters and the performance of our estimator.
    As long as all of the above conditions are satisfied, any set of grid points can in principle be employed.
    For example, when the covariance operator $\Gamma_{it}$ is indeed a trace-class operator, Assumption \ref{as:error}(iii) is more likely to be satisfied if the grid points are placed quasi-uniformly over $[0,1]$ with a sufficiently large $L$.
    Further, to satisfy Assumption \ref{as:matrix1}, it is generally necessary that the collection of basis vectors $\{\phi^K(s_1), \ldots, \phi^K(s_L)\}$ contains at least $K$ linearly independent columns.
    In particular, when the 2SLS-type weight matrix $\Omega_{nT}$ in \eqref{eq:2slsweight} is used, the nonsingularity of $L^{-1} \sum_{l=1}^L \phi^K(s_l)\phi^K(s_l)^\top$ is required to satisfy Assumption \ref{as:weights}(ii).
    Thus, in applications, it is advisable to choose grid points with equal spacing and to take $L$ sufficiently large.
    Our numerical simulation results suggest that setting $L = 2K$ is sufficient to produce accurate estimates (see Appendix \ref{app:MC}).
\end{remark}

\begin{theorem}[Rates of convergence]\label{thm:roc}
Suppose that Assumptions \ref{as:inverse}, \ref{as:sample_space} -- \ref{as:basis} hold.
In addition, assume that $K/\sqrt{nT} \to 0$ and $K^{1/2 - \varsigma} \to 0$ as $nT \to \infty$.
Then,
\begin{itemize}
    \item[(i)] $|| \hat \theta_{nT} - \theta_0|| \lesssim_p 1/\sqrt{nT} + K^{-\varsigma}$
    \item[(ii)] $||\hat \alpha_{nT} - \alpha_0 ||_{L^2} \lesssim_p 1/\sqrt{nT} + K^{-\varsigma}$, and $\sup_{s \in [0,1]}|\hat \alpha_{nT}(s) - \alpha_0(s) | \lesssim_p \sqrt{K}/\sqrt{nT} + K^{1/2 - \varsigma}$ 
    \item[(iii)] $||\hat \gamma_{nT} - \gamma_0 ||_{L^2} \lesssim_p 1/\sqrt{nT} + K^{-\varsigma}$, and $\sup_{s \in [0,1]}|\hat \gamma_{nT}(s) - \gamma_0(s) | \lesssim_p \sqrt{K}/\sqrt{nT} + K^{1/2 - \varsigma}$ 
    \item[(iv)] $||\hat \beta_{nT, j} - \beta_{0j} ||_{L^2} \lesssim_p 1/\sqrt{nT} + K^{-\varsigma}$, and $\sup_{s \in [0,1]}|\hat \beta_{nT, j}(s) - \beta_{0j}(s) | \lesssim_p \sqrt{K}/\sqrt{nT} + K^{1/2 - \varsigma}$ for all $j \in [d_x]$.
\end{itemize}
\end{theorem}

The proof of Theorem \ref{thm:roc} is given in Appendix \ref{app:proofs}. 
Result (i) shows that, if the functional parameters are sufficiently smooth so that $K^{-\varsigma} \lesssim (nT)^{-1/2}$, the series coefficient estimator is consistent and converges at the parametric rate.
The same convergence rate applies to the $L^2$-convergence rate of the functional estimators, as shown in (ii) -- (iv).
Although these results may appear somewhat unexpected since the dimension of $\theta_0$ diverges as $K \to \infty$, parametric convergence rates are commonly observed in dense functional data settings (e.g., \citealp{huang2004polynomial, li2010uniform, hu2021robust}), our setup corresponding to an extreme case in which the entire functions are observed.
The key condition underlying this rate is Assumption \ref{as:error}(iii), which ensures that, once $K$ is sufficiently large, adding further basis terms does not inflate the variance of the series coefficient estimator.
The uniform convergence rate for these estimators is $K^{1/2}$ slower than the $L^2$-convergence rate.\footnote{
    Note that the uniform convergence results obtained here are not necessarily the sharpest, and the theoretically optimal convergence rates under our setup are also unknown. These points are left for future research. 
    }

\bigskip

We next present the limiting distribution of our estimators.
To this end, we introduce the following notations and additional assumption:
\begin{align}
    & \underset{d_g \times (d_x + 2)K}{J_{nT}(s; \theta)}
    \coloneqq \frac{\partial g_{nT}(s; \theta)}{\partial \theta^\top}, \;\;
    \bar J_{nT}(\theta)
    \coloneqq \frac{1}{L} \sum_{l = 1}^L J_{nT}(s_l; \theta), \;\;
    \bar J_{nT}
    \coloneqq \bbE\left[\bar J_{nT}(\theta_0)\right] \\
    & \underset{d_g \times 1}{g_{1,nT}(s)}
    \coloneqq \frac{1}{N} \left(\begin{array}{c}
     \bm{Z}(s)^\top \bm{D}^\top \bm{D} \bm{\mathcal{E}}(s) \\
     \bm{\mathcal{E}}(s)^\top \bm{D}^\top P_1 \bm{D} \bm{\mathcal{E}}(s) \\
     \vdots \\
     \bm{\mathcal{E}}(s)^\top \bm{D}^\top P_M \bm{D} \bm{\mathcal{E}}(s)
    \end{array}\right), \;\; \bar g_{1,nT}
    \coloneqq \frac{1}{L} \sum_{l = 1}^L g_{1,nT}(s_l)\\
    & \mathcal{V}_{nT} \coloneqq N \bbE\left[\bar g_{1,nT} \bar g_{1,nT}^\top \right] \\
    & \underset{(d_x + 2)K \times (d_x + 2)K}{\Sigma_{nT}} \coloneqq \left( \bar J^\top_{nT} \Omega_{nT} \bar J_{nT} \right)^{-1} \bar J^\top_{nT} \Omega_{nT} \mathcal{V}_{nT} \Omega_{nT} \bar J_{nT} \left( \bar J^\top_{nT} \Omega_{nT} \bar J_{nT} \right)^{-1}.
\end{align} 
More explicit forms of the matrices $J_{nT}(s; \theta)$ and $\mathcal{V}_{nT}$ can be found in \eqref{eq:Jmat} and in \eqref{eq:Vmat} in Appendix \ref{app:prep}, respectively.
Further, let $\mathbb{S}_\alpha$, $\mathbb{S}_\gamma$, and $\mathbb{S}_j$ be the $K \times (d_x + 2)K$ selection matrices such that $\theta_{0\alpha} = \mathbb{S}_\alpha \theta_0$, $\theta_{0\gamma} = \mathbb{S}_\gamma \theta_0$, and $\theta_{0j} = \mathbb{S}_j \theta_0$ hold.

\begin{assumption}[Misc.]\label{as:matrix2}
    (i) $\limsup_{nT \to \infty} \lambda_{\max} \left( (NL)^{-1} \allowbreak \sum_{l = 1}^L \bbE[\bm{H}(s_l)^\top \bm{H}(s_l)] \right) < \infty$;
    (ii) $0 < \liminf_{nT \to \infty} \lambda_{\min}\left( \bar J^\top_{nT} \bar J_{nT} \right) \le \limsup_{nT \to \infty} \lambda_{\max}\left( \bar J^\top_{nT} \bar J_{nT} \right) < \infty$; and 
    (iii) $0 < \liminf_{nT \to \infty} \lambda_{\min}\left( \mathcal{V}_{nT} \right) \le \limsup_{nT \to \infty} \lambda_{\max}\left( \mathcal{V}_{nT} \right) < \infty$.
\end{assumption}

\begin{theorem}[Asymptotic normality]\label{thm:normality}
Suppose that Assumptions \ref{as:inverse}, \ref{as:sample_space} -- \ref{as:matrix2} hold.
In addition, assume that $K/\sqrt{nT} \to 0$, $K^2/(\sqrt{nT} \left\| \phi^K(s) \right\|^2) \to 0$, and $\sqrt{nT} K^{-\varsigma} \to 0$ as $nT \to \infty$.
Then, for any given $s \in [0,1]$,
\begin{align}
    \text{(i)} \;\; & \frac{\sqrt{N}\left( \hat \alpha_{nT}(s) - \alpha_0(s) \right)}{\sigma_{nT,\alpha}(s)} \overset{d}{\to} \mathcal{N}(0,1), \;\; 
    \text{(ii)} \;\; \frac{\sqrt{N}\left( \hat \gamma_{nT}(s) - \gamma_0(s) \right)}{\sigma_{nT,\gamma}(s)} \overset{d}{\to} \mathcal{N}(0,1) \\
    \text{(iii)} \;\; & \frac{\sqrt{N}\left( \hat \beta_{nT, j}(s) - \beta_{0j}(s) \right)}{\sigma_{nT,j}(s)}  \overset{d}{\to} \mathcal{N}(0,1),
\end{align}
where $[\sigma_{nT,\alpha}(s)]^2 \coloneqq \phi^K(s)^\top \mathbb{S}_\alpha \Sigma_{nT} \mathbb{S}_\alpha^\top \phi^K(s)$, $[\sigma_{nT,\gamma}(s)]^2 \coloneqq \phi^K(s)^\top \mathbb{S}_\gamma \Sigma_{nT} \mathbb{S}_\gamma^\top \phi^K(s)$, $[\sigma_{nT,j}(s)]^2 \coloneqq \phi^K(s)^\top \mathbb{S}_j \Sigma_{nT} \mathbb{S}_j^\top \phi^K(s)$, and recall that $N = n(T-1)$.
\end{theorem}

\begin{corollary}\label{cor:normality}
Suppose that the assumptions in Theorem \ref{thm:normality} hold.
If $K^3/(nT) \to 0$ and $K^{3/2 - \varsigma} \to 0$ as $nT \to \infty$ additionally hold, we have
\begin{align}
    \text{(i)} \;\; & \frac{\sqrt{N}\left( \hat \alpha_{nT}(s) - \alpha_0(s) \right)}{\hat \sigma_{nT,\alpha}(s)} \overset{d}{\to} \mathcal{N}(0,1), \;\; 
    \text{(ii)} \;\; \frac{\sqrt{N}\left( \hat \gamma_{nT}(s) - \gamma_0(s) \right)}{\hat \sigma_{nT,\gamma}(s)} \overset{d}{\to} \mathcal{N}(0,1) \\
    \text{(iii)} \;\; & \frac{\sqrt{N}\left( \hat \beta_{nT, j}(s) - \beta_{0j}(s) \right)}{\hat \sigma_{nT,j}(s)}  \overset{d}{\to} \mathcal{N}(0,1),
\end{align}
where $\hat \sigma_{nT,\alpha}(s)$, $\hat \sigma_{nT,\gamma}(s)$, and $\hat \sigma_{nT,j}(s)$ are standard error estimators, whose definitions are provided in \eqref{eq:se} in Appendix \ref{app:covmat}.
\end{corollary}

Theorem \ref{thm:normality} and Corollary \ref{cor:normality} establish the pointwise asymptotic normality of the integrated-GMM estimators.
We impose additional undersmoothing conditions to ensure that the bias terms vanish sufficiently quickly. 
See also Proposition \ref{prop:var} in Appendix \ref{app:covmat}, where additional asymptotic results for consistent variance estimation are established.

For the interpretation of these results, there are several points that should be noted.
First, the standard deviations $\sigma_{nT,\alpha}(s)$, $\sigma_{nT,\gamma}(s)$, and $\sigma_{nT,j}(s)$ depend on the choice of the grid $\{s_l\}$ and its size $L$, as shown in their definitions.
In this sense, one could consider a data-driven procedure for selecting the grid so as to improve the efficiency.
We leave the investigation of such procedures for future research.
Second, a straightforward calculation shows that the order of the standard deviations and their estimators is governed by $\|\phi^{K}(s)\|$, which is typically of order $\sqrt{K}$ for $s \in (0,1)$.
Therefore, the theorem does not necessarily assert the $\sqrt{N}$-asymptotic normality for our functional estimator (similar to Theorem 2 of \cite{newey1997convergence}, for example).
Third, the asymptotic normality results hold for any evaluation point $s \in [0,1]$, including the boundary points $s = 0$ and $s = 1$.
Nevertheless, both the required conditions on $K$ and the structure of the standard deviations depend on $s$ through the behavior of $\phi^{K}(s)$.
For example, consider $\phi^{K}(0) = (1,0,\ldots,0)^\top$, as is typical in practice.
In this case, $\|\phi^{K}(0)\|^{2} = 1$.
Since Theorem \ref{thm:normality} requires $K^{2}/(\sqrt{nT}\|\phi^{K}(s)\|^{2}) \to 0$, this implies that achieving the normality at the boundary needs a more restrictive condition on $K$ than at interior points.

\begin{remark}[Choice of $K$]\label{rem:K}
Suppose that $K$ is proportional to $(nT)^c$ for some $0 < c < \infty$, and that $c_1 K \le ||\phi^K(s)||^2 \le c_2 K$ for some $0 < c_1 \le c_2 < \infty$.
Then, to obtain the asymptotic normality results in Corollary \ref{cor:normality}, we require $K^3/(nT) \to 0$ and $\sqrt{nT} K^{-\varsigma} \to 0$ simultaneously.
These conditions can be reduced to the following restriction on $c$: $1/(2\varsigma) < c < 1/3$.
Thus, when the functional coefficients are believed to be sufficiently smooth, setting, for example, $K = \lfloor c_K (nT)^{1/5} \rfloor$ for some constant $c_K > 0$ would be a reasonable choice.

As a more data-driven approach, we also consider a cross-validation procedure in our numerical studies.
Specifically, we split the data into a training sample consisting of periods $t = 1,\ldots,T_{\text{train}}$ and a validation sample consisting of $t = T_{\text{train}}+1,\ldots,T$.
We then estimate $\alpha_0$, $\gamma_0$, and $\beta_0$ using the training sample to predict the outcome difference $Y_{i,t+1}-Y_{it}$ for $t = T_{\text{train}}+1,\ldots,T-1$, and the value of $K$ is chosen to minimize the average mean squared prediction error (AMSPE):
\begin{align}
    \text{AMSPE}
    \coloneqq
    \frac{1}{|\mathcal S|} \sum_{s \in \mathcal S} \left( \frac{1}{n (T - T_{\text{train}} - 1)} \sum_{i = 1}^n \sum_{t = T_{\text{train}} + 1}^{T - 1} \left\{ [\hat Y_{i,t+1}(s) - \hat Y_{it}(s)] - [Y_{i,t+1}(s) - Y_{it}(s)] \right\}^2 \right),
\end{align}
where $\mathcal S$ is a given set of evaluation points, and $|\mathcal S|$ denotes its cardinality.
In our simulation analysis, this procedure performs reasonably well.
For further details, see Appendix \ref{app:MC}.
\end{remark}

\section{Network Multiplier Effects: Marginal Effects and Impulse Responses}\label{sec:impulse}

Once the model is estimated, as a next step, one might be interested in computing the marginal effects of covariates on the outcome.
In a standard linear regression model without network interaction, the estimated coefficients directly represent the marginal effects of their corresponding covariates.
However, in the presence of intricate functional interaction, this is no longer the case.

As shown in Section \ref{sec:model}, under Assumption \ref{as:inverse}, we have the following moving-average type representation:
\begin{align}
    Y_t = \mathcal{A}_3^t Y_0 + \sum_{\ell = 0}^{t-1} \mathcal{A}_3^\ell (\text{Id} - \mathcal{A}_1)^{-1}[X_{t-\ell}\beta_0 + F_0 + \mathcal{E}_{t-\ell}], \;\; t \in [T].
\end{align}
This expression indicates that the marginal effect of increasing $X_{it}^j$ by one unit on $Y_t$ is given by $\partial Y_t/(\partial X_{it}^j) = (\text{Id} - \mathcal{A}_1)^{-1}\bm{e}_i\beta_{0j} = \sum_{\ell = 0}^\infty \mathcal{A}_1^\ell \bm{e}_i\beta_{0j}$, where $X_t^j$ is the $j$-th column of $X_t$, and $\bm{e}_i$ denotes the $i$-th column of $I_n$.
Alternatively, a little more informative expression can be obtained as follows: letting $\tau(h,s) \coloneqq \alpha_0(s) A_1(h, s)$,
\begin{align}
    M(i,j,s)
    & \coloneqq \partial Y_t(s)/(\partial X_{it}^j) = \bm{e}_i \beta_{0j}(s) + W_n \bm{e}_i \tau(\beta_{0j}, s) + W_n^2 \bm{e}_i \tau^2(\beta_{0j}, s) + \cdots \\
    & = \sum_{\ell = 0}^\infty W_n^\ell \bm{e}_i \tau^\ell(\beta_{0j}, s),
\end{align}
where $\tau^0(\beta_{0j}, s) = \beta_{0j}(s)$, and $\tau^\ell(\beta_{0j}, s) = \tau(\tau^{\ell - 1}(\beta_{0j}, \cdot), s)$ for $\ell \geq 1$.
From this, we can clearly see that the marginal effects $M(i,j,s)$ of increasing $X_{it}^j$ consist of the direct effect on unit $i$, the indirect effect on $i$'s immediate neighbors, the second-order indirect effect on $i$'s neighbors' neighbors, and so forth, highlighting the presence of the \textit{network multiplier} effect.
More specifically, recall that when $W_n$ represents a (weighted) adjacency matrix, the $(i,j)$-th element of $W_n^\ell$ corresponds to the number of (weighted) \textit{walks} between $i$ and $j$ of length $\ell$.
Thus, the $k$-th element of $M(i,j,s)$ is interpreted as the weighted sum of the number of walks from $i$ to $k$, where the contribution of each length-$\ell$ walk is weighted by $\tau^\ell(\beta_{0j}, s)$.

To estimate the marginal effects, in addition to replacing the unknown parameters with their estimators, the infinite sum generally needs to be approximated by a truncated sum: for some positive integer $S$,
\begin{align}
    \hat M^S_{nT}(i,j,s) \coloneqq \sum_{\ell = 0}^S W_n^\ell \bm{e}_i \hat \tau_{nT}^\ell(\hat \beta_{nT,j}, s),
\end{align}
where $\hat \tau_{nT}(h,s) \coloneqq \hat \alpha_{nT}(s) A_1(h,s)$.
Meanwhile, in the special case of concurrent interaction such that $A_1(h,s) = h(s)$, we have $\tau(h,s) = \alpha_0(s)h(s)$, $\tau^2(h,s) = (\alpha_0(s))^2 h(s)$, and so forth.
Thus, $M(i,j,s) = \sum_{\ell = 0}^\infty (\alpha_0(s) W_n)^\ell \bm{e}_i \beta_{0j}(s) = (I_n - \alpha_0(s) W_n)^{-1}\bm{e}_i\beta_{0j}(s)$.
In this case, we can estimate $M(i,j,s)$ directly as $(I_n - \hat \alpha_{nT}(s) W_n)^{-1}\bm{e}_i\hat \beta_{nT,j}(s)$, without computing the infinite sum.

\bigskip

In the above discussion, we have demonstrated how the impacts of shifting one's covariate propagate to others.
The dynamic interaction structure of our model also allows us to consider conventional impulse responses over future periods, as in time series vector autoregression.
However, because our framework accommodates panels with a short or fixed $T$, we focus instead on functional network impulse responses, which characterize how a functional shock propagates across units and evaluation points.
In particular, in a similar spirit to \cite{koop1996impulse}, we define
\begin{align}
    I(i, \eta, s) \coloneqq \mathbb{E}[Y_t(s) \mid \varepsilon_{it} = \eta] - \mathbb{E}[Y_t(s)],
\end{align}
where $\eta$ is a given bounded external shock "function".
By a similar calculation as above, we obtain
\begin{align}
    I(i, \eta, s) = \sum_{\ell = 0}^\infty W_n^\ell \bm{e}_i \tau^\ell(\eta, s).
\end{align}
Plotting each element of $W_n^\ell \bm{e}_i \tau^\ell(\eta, s)$ against $\ell = 0, 1, 2, \ldots$ shows how the shock is transmitted through successive orders of network connections and can be interpreted as a functional network version of the impulse response function, similarly to \cite{denbee2021network}.

Under a concurrent interaction model, the impulse responses at $s$ take the following form: $I(i,\eta,s) = (I_n - \alpha_0(s) W_n)^{-1}\bm{e}_i\eta(s)$.
Thus, if there is no exogenous shock at $s$, i.e., if $\eta(s) = 0$, the expected outcome at $s$ remains unaffected.
This implies, for instance, that a travel demand shock that occurred five minutes ago has no impact on current mobility availability, which is unrealistic.
On the other hand, if the interaction structure is given by $A_1(h,s) = \int_0^1 h(u) \nu_1(u,s) \text{d}u$ with $\nu_1(s',s) \neq 0$ for $s' < s$, then a shock occurring at $s'$ can transmit to the outcome at $s$, leading to nonzero impulse responses at $s$ even when $\eta(s) = 0$.

The estimation of $I(i,\eta,s)$ can be performed in the same manner as above.
For some positive integer $S$, we estimate $I(i,\eta,s)$ by $\hat I^S_{nT}(i,\eta,s) \coloneqq \sum_{\ell = 0}^S W_n^\ell \bm{e}_i \hat \tau_{nT}^\ell(\eta, s)$.
The next proposition provides the convergence rate of $\hat M^S_{nT}(i,j,s)$ and that of $\hat I^S_{nT}(i,\eta,s)$.

\begin{proposition}\label{prop:impulse}
    Suppose that the assumptions in Theorem \ref{thm:roc} hold.
    In addition, assume that $\bar \alpha_0 < 1$ and $\|W_n\|_\infty \le 1$ uniformly in $n \ge 1$.
    Then, uniformly in $s \in [0,1]$,
    \begin{itemize}
        \item[(i)] $\max_{i \in [n]} \left\| \hat M^S_{nT}(i,j,s) - M(i,j,s) \right\|_\infty \lesssim_p \sqrt{K}/\sqrt{nT} + K^{1/2 - \varsigma} + \bar \alpha_0^{S+1}$,
        \item[(ii)] $\max_{i \in [n]} \left\| \hat I^S_{nT}(i,\eta,s) - I(i,\eta,s) \right\|_\infty \lesssim_p \sqrt{K}/\sqrt{nT} + K^{1/2 - \varsigma} + \bar \alpha_0^{S+1}$.
    \end{itemize}
\end{proposition}

This proposition indicates that the uniform convergence rates for the marginal effect and the impulse response estimators depend on the uniform convergence rate of the integrated-GMM estimator and the summation order $S$.
Since the approximation error from truncating the infinite sum decreases geometrically as $S$ increases, if $\bar{\alpha}_0$ is close to zero, setting $S = 4$ or $5$ would be sufficient.
Although $\bar{\alpha}_0$ is unknown in practice, it can be naturally estimated by $\bar{\hat{\alpha}}_{nT} \coloneqq \max_{s \in [0,1]} |\hat{\alpha}_{nT}(s)|$.
Then, if the value of $\bar{\hat{\alpha}}_{nT}$ is relatively large, it would be preferable to use a larger truncation order, say $S = 9$ or $10$.

\section{Monte Carlo Simulation}\label{sec:MC}

In this section, we conduct a series of Monte Carlo experiments to evaluate the finite-sample performance of the integrated-GMM estimator.
To save space, we summarize only the main findings from the Monte Carlo simulations here; a complete description of the simulation design and the full numerical results are provided in Appendix \ref{app:MC}.
We compare three estimators: \textbf{GMM 1}, the integrated-GMM estimator with the weight matrix given in \eqref{eq:2slsweight}; \textbf{GMM 2}, the integrated-GMM estimator with the identity weight matrix; and (integrated) \textbf{2SLS}, the GMM 1 estimator without the quadratic moment conditions.

Overall, all estimators perform well in terms of bias; however, GMM 1 consistently achieves the lowest RMSE (root-mean-squared error) across almost all settings, highlighting the efficiency gains from incorporating quadratic moment conditions.
Interestingly, GMM 2 underperforms compared with 2SLS in most cases, indicating that the choice of the GMM weight matrix is at least as crucial as the inclusion of additional moment conditions.
Meanwhile, in terms of RMSE for estimating $\alpha_0$, 2SLS is more sensitive to the strength of IV than the others, with its RMSE decreasing more drastically as the IV strength increases.
For the estimation of $\gamma_0$, the IV strength matters significantly for all three estimators in the same way.
As expected, the strength of IV has only a minor impact on the estimation of $\beta_0$.
Overall, the choice of $L$ appears to have little effect on performance.
Lastly, when we increase the sample size from $nT = 250$ to $nT = 1000$, the RMSE values are roughly halved, which is consistent with the $\sqrt{nT}$-consistency of the estimators established in Theorem \ref{thm:roc}.

\section{Extensions}\label{sec:extension}

In this section, we present several extensions of the current framework that are useful in empirical applications and can be implemented with relatively minor modifications to our method.

\subsection{Functional DNAR models with two-way fixed effects}\label{sub:timeeffect}

So far, we have assumed that the error terms are independent over time, which may be restrictive in practice.
A practical approach to accommodating time effects is to extend the present model \eqref{eq:model} to the following two-way fixed effects model: for $s \in [0,1]$,
\begin{align}\label{eq:model_te}
    Y_{it}(s) = \alpha_0(s) \sum_{j = 1}^n w_{i,j} A_1(Y_{jt}, s) + \gamma_0(s) A_2(Y_{i,t-1}, s) + X_{it}^\top \beta_0(s) + f_{0i}(s) + c_{0t}(s) + \varepsilon_{it}(s).
\end{align}
To facilitate the analysis of this model, we introduce the following additional conditions.
\begin{assumption}\label{as:twoway}
(i) For all $i \in [n]$, $\sum_{j = 1}^n w_{i,j} = 1$; and (ii) $\text{diag}( \calR_n^\top P_{m,1} \calR_n) = \bm 0_n$.
\end{assumption}
Assumption \ref{as:twoway}(i) rules out isolated units.
In condition (ii), we assume that $P_{m,1}$ satisfies Assumption \ref{as:weights}(i), while the restriction $\text{diag}(P_{m,1}) = \bm 0_n$ is now no longer necessary.

Under Assumption \ref{as:twoway}(i), to eliminate the time effects $c_{0t}(s)$ from the model, we can consider de-meaning by the neighborhood average rather than the conventional de-meaning by the global average.
Specifically, letting $\calR_n \coloneqq I_n - W_n$ and $\underset{N \times N}{\bm{\calR}} \coloneqq I_{T-1} \otimes \calR_n$, we obtain
\begin{align}
    \bm{\calR} \bm D \bm Y(s) = \bm{\calR} \bm D \bm H(s) \theta_0 + \bm{\calR} \bm D \bm V(s) + \bm{\calR} \bm D \bm{\mathcal{E}}(s).
\end{align}

It is straightforward to verify that the linear moment conditions
\begin{align}
    \mathbb{E}\left[ \bm{Z}(s)^\top \bm{D}^\top \bm{\calR}^\top \bm{\calR} \bm{D} \bm{\mathcal{E}}(s) \right] 
    = \bm{0}_{(d_q + d_x)K}
\end{align}
hold.
In addition, by Assumption \ref{as:twoway}(ii), the quadratic moment condition
\begin{align}
    \mathbb{E}\left[ \bm{\mathcal{E}}(s)^\top \bm{D}^\top \bm{\calR}^\top P_m \bm{\calR} \bm{D} \bm{\mathcal{E}}(s) \right] 
    = 0
\end{align}
also holds.
Based on these moment conditions, we can construct a GMM estimator as in Subsection \ref{subsec:gmm}.
With an abuse of notation, let $\hat \alpha_{nT}(s)$, $\hat \gamma_{nT}(s)$, and $\hat \beta_{nT, j}(s)$ be the resulting estimators of $\alpha_0(s)$, $\gamma_0(s)$, and $\beta_{0j}(s)$, respectively.

\begin{theorem}[Asymptotic normality of the two-way fixed effects estimator]\label{thm:twoway}
Suppose that the assumptions in Corollary \ref{cor:normality} and Assumption \ref{as:twoway} are satisfied.
Then, we have
\begin{align}
    \text{(i)} \;\; & \frac{\sqrt{N}\left( \hat \alpha_{nT}(s) - \alpha_0(s) \right)}{\hat \sigma_{nT,\alpha}(s)} \overset{d}{\to} \mathcal{N}(0,1), \;\; 
    \text{(ii)} \;\; \frac{\sqrt{N}\left( \hat \gamma_{nT}(s) - \gamma_0(s) \right)}{\hat \sigma_{nT,\gamma}(s)} \overset{d}{\to} \mathcal{N}(0,1) \\
    \text{(iii)} \;\; & \frac{\sqrt{N}\left( \hat \beta_{nT, j}(s) - \beta_{0j}(s) \right)}{\hat \sigma_{nT,j}(s)}  \overset{d}{\to} \mathcal{N}(0,1),
\end{align}
where $\hat \sigma_{nT,\alpha}(s)$, $\hat \sigma_{nT,\gamma}(s)$, and $\hat \sigma_{nT,j}(s)$ are standard error estimators in this context, whose definitions are given in Appendix \ref{app:twoway}.
\end{theorem}

The proof of Theorem \ref{thm:twoway} is outlined in Appendix \ref{app:twoway}.
It is worth noting that constructing a weight matrix $P_{m,1}$ satisfying $\text{diag}(\calR_n^\top P_{m,1} \calR_n)=\bm 0_n$ is relatively straightforward.
For an arbitrarily chosen $\tilde P_{m,1} = (\tilde p_{m,i,j})$, let $\calL=\text{diag}(\ell_{1,1},\ldots,\ell_{n,n})$ be a diagonal matrix and set $P_{m,1} = \tilde P_{m,1} + \calL$.
Then, for each $i \in [n]$,
\begin{align}
    (\calR_n^\top P_{m,1} \calR_n)_{ii}
    &= \sum_{j=1}^n \sum_{k=1}^n r_{j,i} \tilde p_{m,j,k} r_{k,i} + \sum_{j=1}^n \ell_{j,j} r_{j,i}^2,
\end{align}
where $r_{j,i}$ denotes the $(j,i)$-th element of $\calR_n$.
Hence, the constraint $\text{diag}(\calR_n^\top P_{m,1} \calR_n)=\bm 0_n$ is equivalent to the linear system
\begin{align}
    \bm y = \bm C \bm \ell,
\end{align}
where $\bm y = - \text{diag}(\calR_n^\top \tilde P_{m,1} \calR_n)$, $\bm C = (c_{i,j})$, $c_{i,j}= r_{j,i}^2$, and $\bm \ell = (\ell_{1,1},\ldots,\ell_{n,n})^\top$.
Therefore, one may choose $\bm \ell = \bm C^{-} \bm y$, where $\bm C^{-}$ is a generalized inverse of $\bm C$.


\subsection{More general interaction structures}\label{subsec:nonlinear}

As a straightforward extension of our DNAR model, one may consider a nonlinear DNAR model:
\begin{align}\label{eq:nonlinearmodel}
    Y_{it}(s)
    = \alpha_0(s) \sum_{j = 1}^n w_{i,j} A_1(m_{1s}(Y_{jt}), s) + \gamma_0(s) A_2(m_{2s}(Y_{i,t-1}), s) + \mu_{it}(s),
\end{align}
where $\mu_{it}(s) \coloneqq X_{it}^\top \beta_0(s) + f_{0i}(s) + \varepsilon_{it}(s)$, and $m_{ks}$ is a known transformation (up to parameters).
For example, consider $m_{ks}(Y_{it}) = b_{k1}(s)Y_{it} + b_{k2}(s) Y_{it}^2 + \cdots + b_{kP}(s) Y_{it}^P$.
Then, setting $\alpha_{0p}(s) \coloneqq \alpha_0(s) b_{1p}(s)$ and $\gamma_{0p}(s) \coloneqq \gamma_0(s) b_{2p}(s)$, we can write the model as
\begin{align}\label{eq:exmodel}
    Y_{it}(s) = \sum_{p = 1}^P \alpha_{0p}(s) A_1(\bar Y_{it}^{(p)}, s) + \sum_{p = 1}^P \gamma_{0p}(s) A_2( Y_{i,t-1}^p, s) + \mu_{it}(s),
\end{align}
where $\bar Y_{it}^{(p)} \coloneqq \sum_{j = 1}^n w_{i,j} Y_{jt}^p$.

Another potentially useful extension is a multiple-network model:
\begin{align}
    Y_{it}(s)
    = \sum_{p = 1}^P \alpha_{0p}(s) \sum_{j = 1}^n w^{(p)}_{i,j} A_1(Y_{jt}, s) + \gamma_0(s) A_2(Y_{i,t-1}, s) + \mu_{it}(s),
\end{align}
where $W_n^{(p)} = (w^{(p)}_{i,j})$, $p \in [P]$, are different interaction matrices.
In this case, if we define $\bar Y_{it}^{(p)} \coloneqq \sum_{j = 1}^n w^{(p)}_{i,j} Y_{jt}$, the model can  be expressed in a similar form to \eqref{eq:exmodel}.

For both cases, as long as valid IVs $Q_{it}$ for the multiple endogenous components are available, essentially the same estimation and inference procedure as in Section \ref{sec:estimation} can be applied, while we need to impose more restrictive conditions to ensure completeness and stability.

\subsection{Incompletely observed response function}\label{sub:incomplete}
The integrated-GMM estimator is often infeasible because the response functions are typically observed only at a finite set of points in $[0,1]$.
Even in such cases, we can approximate the entire functional form of $Y_{it}$ using a linear interpolation method.

Suppose that for each $(i,t)$, $Y_{it}$ is observed at $L_{it}$ distinct points $0 \le s_{it,1} < s_{it,2} < \dots < s_{it,L_{it}} \le 1$.
Then, for each given $s \in [s_{it,l}, s_{it,l + 1}]$, define
\begin{align}\label{eq:interp}
        Y^{\text{int}}_{it}(s) \coloneqq Y_{it}(s_{it,l}) + \frac{Y_{it}(s_{it,l + 1}) - Y_{it}(s_{it,l})}{s_{it,l + 1} - s_{it,l}} (s - s_{it,l}).
\end{align}
When $s < s_{it,1}$ (resp. $s > s_{it,L_{it}}$), we set $Y^{\text{int}}_{it}(s) \coloneqq Y_{it}(s_{it,1})$ (resp. $Y^{\text{int}}_{it}(s) \coloneqq Y_{it}(s_{it,L_{it}})$).

Other than linear interpolation, one may also use a kernel method, as in \cite{zhu2022network}, to obtain $Y^{\text{int}}_{it}(s)$.
Then, using $Y^{\text{int}}_{it}(s)$ in place of $Y_{it}(s)$, we can write
\begin{align}
    Y^{\text{int}}_{it}(s) = \alpha_0(s) A_1(\bar Y^{\text{int}}_{it}, s) + \gamma_0(s) A_2(Y^{\text{int}}_{i,t-1}, s) + X_{it}^\top \beta_0(s) + f_{0i}(s) + \varepsilon_{it}(s) + u_{it}(s),
\end{align}
where $u_{it}(s)$ is the interpolation error: $u_{it}(s) \coloneqq Y^{\text{int}}_{it}(s) - Y_{it}(s) + \alpha_0(s) A_1(\bar Y_{it} - \bar Y^{\text{int}}_{it}, s) + \gamma_0(s) A_2(Y_{i,t-1} - Y^{\text{int}}_{i,t-1}, s)$. 
Thus, if $u_{it}(s)$ converges to zero sufficiently quickly for all $s \in [0,1]$, $i \in [n]$, and $t \in [T]$, we can apply the same estimation and inference strategy as above.

\section{Analyzing the Demand of Bike-Sharing System}\label{sec:empiric}

As an empirical application, we analyze spatial interactions in the demand for a bike-sharing system in the U.S.
Demand analysis of shared mobility has been a highly active research topic in recent years across various areas, including transportation research, marketing, economics, and environmental studies.
In particular, bike-sharing systems have attracted increasing attention.
For a comprehensive review of this literature, see \cite{eren2020review}.

\subsection{Data}
The dataset used in this analysis comes from the \textit{Bay Area Bike Share} in San Francisco, which was established in August 2013 and is now known as \textit{Bay Wheels}.
The dataset is publicly available on the \textit{Kaggle} website.\footnote{\url{https://www.kaggle.com/datasets/benhamner/sf-bay-area-bike-share}}
It contains detailed information about the system from August 2013 to August 2015, including station locations, the number of available bicycles at each station over time, and all trip-level data during this time period.
The trip data include details such as start and end times and stations, as well as the user type (subscriber or casual user).
In this dataset, there are 70 bike stations in total; for a map of all 70 station locations, see Figure \ref{fig:map} in Appendix \ref{app:emp}.

Since the initial installation of stations in August 2013, the 70th station (\textit{Ryland Park} station) was added in April 2014.
Accordingly, we use data from May 2014 to August 2015 for this analysis, which represents the largest balanced panel dataset that can be extracted from the raw data.

One concern in the analysis is that shared mobility services often relocate bikes across stations in order to rebalance bike availability.
To mitigate the effects of rebalancing operations in demand analysis, several approaches have been considered in the literature.
For example, using the number of trips, pickups, or returns at each station as the outcome variable is a common strategy (e.g., \citealp{el2017effects, sun2018promoting}).
However, this approach is not straightforwardly applicable in the present functional data framework.
Another approach is to detect potential rebalancing operations using heuristic thresholding rules (e.g., \citealp{faghih2014land, faghih2017empirical}).
Although it is generally impossible to access formal records of relocation operations, as noted by \cite{gammelli2022predictive}, the majority of bike-sharing systems are rebalanced during nighttime (this is called \textit{static rebalancing}).
Then, if we observe sudden jumps or drops in bike counts during nighttime hours, they are likely due to relocation operations conducted by the service provider.
We adopt this approach in our analysis.
Specifically, we first identify instances where the number of available bicycles jumps up/down by more than or equal to 10 all at once.
We then examine the distribution of such events across hours and days, as shown in Figure \ref{fig:reloc} in Appendix \ref{app:emp}.
This figure indicates that sudden drops or increases in bike availability tend to occur between midnight and early morning, particularly on Sundays.
Based on this, we exclude these time periods from our analysis.

Another concern is the enormous size of the dataset.
Because the original data are recorded in minutes every day, using the raw data directly can lead to a memory problem.
Moreover, daily data tend to fluctuate and to be noisy due to random events.

To address the aforementioned issues, we first rounded the trip data to 15-minute intervals and then averaged over Monday through Friday at each interval, discarding data from Saturdays and Sundays.
Furthermore, to avoid potential bike relocation events in weekdays, we restrict the analysis to the time period from 6 AM to 9 PM. 
Consequently, our final dataset is a weekly panel consisting of $n = 70$ stations and $T = 69$ weeks; the first week is not used for estimation but serves as the observation at $t = 0$.
The outcome of interest is the number of bicycles at each station for $s \in [0,1]$, where $s=0$ corresponds to 6 AM and $s=1$ corresponds to 9 PM. 
Figure \ref{fig:avgbike} presents the trajectories of average bike availability for all 70 stations during the first week in our panel.
It clearly shows that most of the variation in bike availability occurs between 6 AM and 9 PM.

\begin{figure}[!ht]
    \centering
    \includegraphics[width = 13cm]{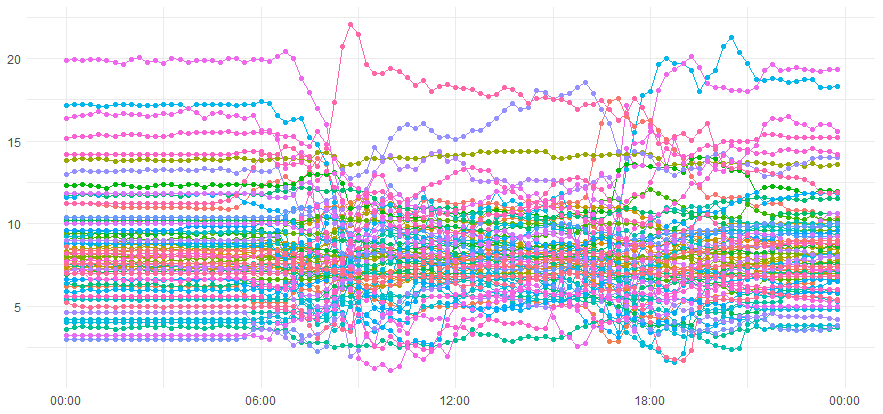}
    \caption{Availability of bikes at each station (averaged over 5-9 May, 2014)}
    \label{fig:avgbike}
\end{figure}

\subsection{Empirical results}

Based on the dataset constructed as described above, we estimate model \eqref{eq:model}, where
\begin{align}
    Y_{it}(s) & = \text{number of available bikes at each station} \\
    A_1(Y_{jt},s) & = \text{average of $Y_{jt}$ in the past one hour} \\
    A_2(Y_{i,t-1},s) & = \text{daily average of $Y_{i,t-1}$} \\
    X_{it} & = \text{[\;ratio of round trips, ratio of subscribers (departing from station $i$), ratio of} \\
    & \phantom{=} \quad \text{subscribers (arriving at station $i$), rainy day, holiday, month dummies\;]$^\top$} \\
    w_{ij} & = \frac{\tilde w_{ij}}{\sum_{j \neq i} \tilde w_{ij}}, \;\; \text{where} \;\; \tilde w_{ij} = \frac{\bm{1}\{\text{dist}(i,j) \le 1\text{km}\}}{\text{dist}(i,j)} 
\end{align}
Here, $\text{dist}(i,j)$ denotes the Euclidean distance between stations $i$ and $j$.
The coefficients on the month dummies are assumed to be constant over $s$.
The rainy day, holiday, and month dummy variables are not used as IVs.
The number of basis terms is selected via the cross-validation procedure described in Appendix \ref{app:MC}, using the last 11 weeks as the validation sample.
As a result, we set $K = 6$ and $L = 3K$.
The estimation procedure is the same as the GMM 1 estimator in Section \ref{sec:MC}.

The estimation result for the interaction effect function $\alpha_0$ is presented in Figure \ref{fig:alpha}.
In the figure, the shaded area depicts the (pointwise) 95\% confidence interval.
From the figure, we observe that positive spatial interaction in bicycle availability exists during the morning hours.
It is plausible that as bike-sharing becomes more popular particularly among commuters, it encourages further use of the service, thereby reinforcing demand during the morning.
Meanwhile, interestingly, negative interaction appears around 4--8 PM.
In the evening, main users may include not only returning commuters but also individuals going out for dining, shopping, concerts, etc.
As a result, bicycles might accumulate at certain popular stations while nearby less-popular stations experience lower availability, leading to the negative interaction.
Alternatively, as kindly suggested by a referee, the negative correlation may also be driven by rebalancing operations conducted by the service provider, which artificially accumulate bikes at popular origin stations.
However, it should be noted that autoregressive models of the current type are generally able to identify the presence of autocorrelation but are not designed to distinguish its underlying mechanisms.
Investigating the detailed mechanisms generating such autocorrelation would require a different empirical framework.
    
\begin{figure}[ht!]
    \centering
    \includegraphics[width = 10cm]{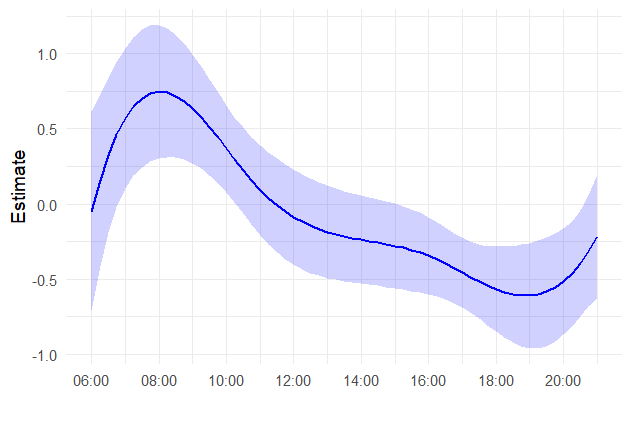}
    \caption{Estimated $\alpha_0(s)$}
    \label{fig:alpha}
\end{figure}

To save space, the estimation results for $\gamma_0(s)$ and $\beta_0(s)$ are presented in Figure \ref{fig:beta} in Appendix \ref{app:emp}, excluding the coefficients on the month dummies.
We find that the dynamic interaction effect across weeks is overall statistically insignificant.
Considering the nature of demand for bike-sharing, the dynamics of within-day bike availability should be much more important than what happened in the previous week.
Therefore, this may not be surprising.
For the other covariates, we observe that only the ratio of arriving subscribers has a statistically significant positive impact on bike availability.
This result is intuitive, as stations with a higher ratio of arriving subscribers are expected to have a larger stock of bikes.

For comparison, we also estimate a non-functional dynamic spatial autoregressive model with individual fixed effects using a 2SLS estimator.
The dependent variable is the average weekly bike availability at each station, obtained by averaging $Y_{it}(s)$ over $s \in [\text{6 AM}, \text{9 PM}]$.
To save space, the results are reported in Table \ref{tab:panelNAR} in Appendix \ref{app:emp}.
Notably, the estimated spatial autoregressive parameter is statistically insignificant, possibly because the positive interaction in the morning and the negative interaction at night offset each other, as observed in Figure \ref{fig:alpha}.
This illustrates that ignoring the functional nature of time-varying bike availability may lead to a misinterpretation of the magnitude and direction of spatial interactions that vary over time.

Lastly, we conduct an impulse response analysis. 
The figures summarizing the results are presented in Figure \ref{fig:imp}. 
For illustration, we arbitrarily select the \textit{Embarcadero at Folsom} station as the target station receiving an external shock.
Specifically, we consider a hypothetical scenario in which the bike stock at this station is reduced by 2 at the peak of 8 AM (panel (a)). 
Panels (b) and (c) illustrate how the shock propagates to its two nearest stations, \textit{Spear at Folsom} and \textit{Temporary Transbay Terminal}.
These figures indicate that the external shock spills over to these stations with a slight time delay.
Since the magnitude of both the external shock and spatial interaction is moderate in this analysis, the impulse responses for both stations are relatively mild.

\begin{figure}[ht!]
    \centering
    \begin{subfigure}{0.5\textwidth}
        \includegraphics[width=\textwidth]{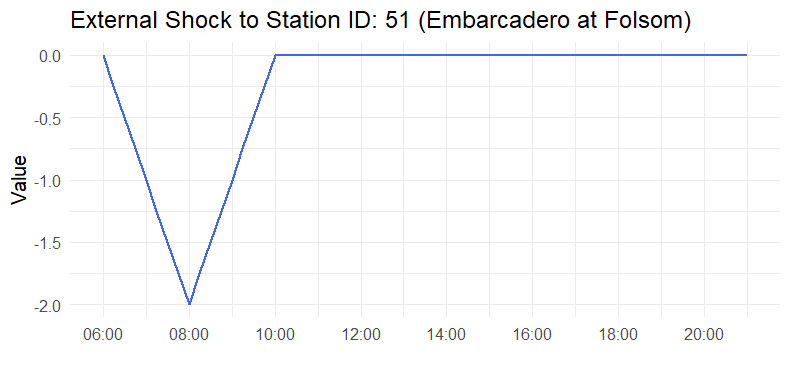}
        \caption{External shock function $\eta$}
    \end{subfigure}
    
    \begin{subfigure}{0.45\textwidth}
        \includegraphics[width=\textwidth]{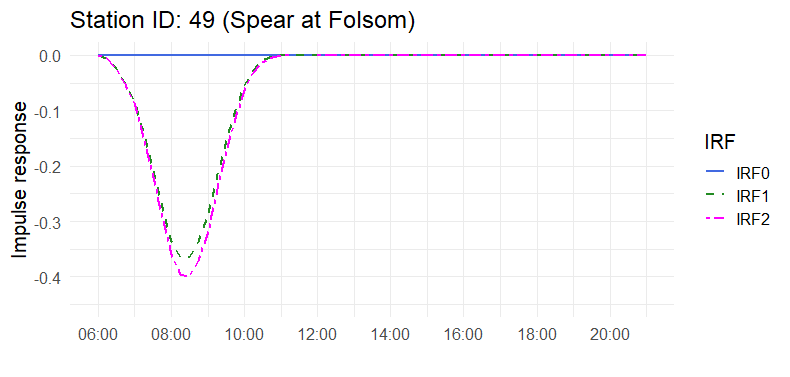}
        \caption{\textit{Spear at Folsom} station}
    \end{subfigure}
    \hfill
    \begin{subfigure}{0.45\textwidth}
        \includegraphics[width=\textwidth]{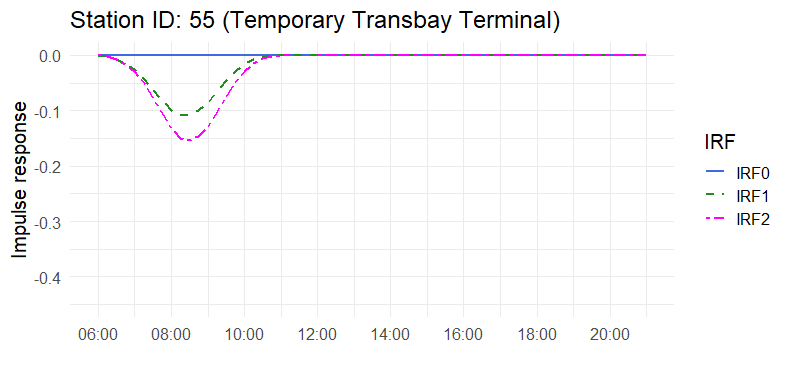}
        \caption{\textit{Temporary Transbay Terminal} station}
    \end{subfigure}

    \caption{Impulse responses}
    \label{fig:imp}

    \bigskip

    \footnotesize

    IRF: $\text{IRF0} = W_n^0 \bm{e}_i \hat \tau_{nT}^0(\eta, s)$, $\text{IRF1} = \text{IRF0} + W_n^1 \bm{e}_i \hat \tau_{nT}^1(\eta, s)$, and $\text{IRF2} = \text{IRF1} + W_n^2 \bm{e}_i \hat \tau_{nT}^2(\eta, s)$. 
\end{figure}

\section{Conclusion} \label{sec:conclusion}

In this paper, we proposed a novel functional DNAR framework to analyze network and dynamic interactions in functional panel data settings.
By extending the standard DNAR model to accommodate functional outcomes and individual fixed effects, we developed an integrated-GMM estimator that can estimate the functional parameters potentially more efficiently than 2SLS-based estimators.
Under certain conditions, we established the theoretical properties of our estimator, including consistency, convergence rates, and pointwise asymptotic normality, and confirmed its finite-sample performance through Monte Carlo simulations.
We also discussed the estimation of marginal effects and functional network impulse responses.
As an empirical application, we analyzed bike availability in a bike-sharing system in the San Francisco Bay Area and discovered a significant spatial interaction that varies over the time of day.
This complex interaction pattern cannot be observed if the functional nature of the response variable is ignored, and thus our finding highlights the importance of accounting for functional dependencies in modeling the demand for shared mobility systems.

\section*{Acknowledgments}

The authors thank the participants at the seminars of National Cheng Kung University and National Taiwan University for their valuable feedback.
Hoshino's work was supported by JSPS KAKENHI Grant Number 23KK0226.
Most parts of this paper were written during Hoshino's research visit at the Melbourne Business School (MBS), University of Melbourne.
He is deeply grateful to MBS for their hospitality.

\clearpage
\begin{center}
\Large Supplementary Appendix for 

"Dynamic Network Autoregressive Models for Functional Panel Data"

\end{center}

\appendix

Appendix \ref{app:prep} collects preparatory materials used to establish our theoretical results.
Appendix \ref{app:proofs} contains the proofs of Theorems \ref{thm:roc} and \ref{thm:normality}, Corollary \ref{cor:normality}, and Proposition \ref{prop:impulse}.
Appendix \ref{app:covmat} develops additional asymptotic results for consistent variance estimation.
Appendices \ref{app:MC} and \ref{app:emp} contain additional materials for the Monte Carlo simulation results and the empirical analysis, respectively.
Appendix \ref{app:twoway} provides the proof of Theorem \ref{thm:twoway} on the asymptotic normality of the two-way fixed effects estimator.

\section{Preparation}\label{app:prep}
The following definition is from \cite{jenish2012spatial}.
\begin{definition}[Near-epoch dependence]\label{def:NED}
    Let $\bm{x} =\{ x_{n,i}: i \in \mathcal{D}_n; \; n \geq 1\}$ and $\bm{e} = \{e_{n,i} : i \in \mathcal{D}_n; \; n \geq 1\}$ be triangular arrays of random fields, where $x$ and $e$ are real-valued and general (possibly infinite-dimensional) random variables, respectively. 
    Then, the random field $\bm{x}$ is said to be $L^p$-near-epoch dependent (NED) on $\bm{e}$ if
    \begin{equation*}
    \left\| x_{n,i} - \bbE \left[ x_{n,i} \mid \mathcal{F}_{n,i}(\delta)\right] \right\|_{p} \leq c_{n,i} \rho(\delta)
    \end{equation*}
    for an array of finite positive constants $\{c_{n,i} : i \in \mathcal{D}_n; \; n \geq 1\}$ and some function $\rho(\delta) \geq 0$ with $\rho(\delta) \to 0$ as $\delta\to \infty $, where $\mathcal{F}_{n,i}(\delta)$ is the $\sigma$-field generated by $\{ e_{n,j} : \Delta(i,j) \leq \delta \}$. 
    The $c_{n,i}$'s and $\rho(\delta)$ are called the NED scaling factors and NED coefficient, respectively.
    $\bm{x}$ is said to be uniformly $L^p$-NED on $\bm{e}$ if $c_{n,i}$ is uniformly bounded.
    If $\rho(\delta) \lesssim \varrho^{\delta}$ for some $ 0 < \varrho <1$, then it is called geometrically $L^p$-NED. 
\end{definition}

Similar to \cite{zhang2025nonlinear}, we apply Definition \ref{def:NED} to spatial-time random fields by treating each pair $(i, t)$ as a single location.
Let $\mathcal{D}_{nT}^{\mathrm{ST}} \coloneqq \mathcal{D}_n \times [T] \subset \mathbb{R}^{d+1}$, and define the spatial-time distance between $(i, t), (j, u) \in \mathcal{D}_{nT}^{\mathrm{ST}}$ as
\begin{align}
    \Delta_{\mathrm{ST}} \left( (i, t), (j, u) \right) \coloneqq \max \left\{ \Delta(i, j), |t - u| \right\}.
\end{align}
For a spatial-time random field $\bm{e} = \{ e_{nT,it} : (i, t) \in \mathcal{D}_{nT}^{\mathrm{ST}} \}$, define $\mathcal{F}_{nT,it}^{\mathrm{ST}}(\delta)$ as the $\sigma$-field generated by $\{ e_{nT,ju} : \Delta_{\mathrm{ST}} \left( (i, t), (j, u) \right) \le \delta \}$.
Accordingly, Definition \ref{def:NED} applies to spatial-time random fields by replacing $\mathcal{D}_n$, $\Delta(i, j)$, and $\mathcal{F}_{n,i}(\delta)$ with $\mathcal{D}_{nT}^{\mathrm{ST}}$, $\Delta_{\mathrm{ST}} \left( (i, t), (j, u) \right)$, and $\mathcal{F}_{nT,it}^{\mathrm{ST}}(\delta)$, respectively.

\bigskip

In the following, for a general $\theta = (\theta_\alpha^\top, \theta_\gamma^\top, \theta_1^\top, \ldots, \theta_{d_x}^\top)^\top \in \Theta_K$, we denote
\begin{align}
    \alpha(s; \theta)
    & \coloneqq \phi^K(s)^\top \theta_\alpha, \;\;
    \gamma(s; \theta)
    \coloneqq \phi^K(s)^\top \theta_\gamma \\
    \beta_j(s; \theta)
    & \coloneqq \phi^K(s)^\top \theta_j, \;\; j \in [d_x]
\end{align}
Recall that we have assumed that $\alpha(s; \theta)$, $\gamma(s; \theta)$, and $\beta_j(s; \theta)$ are uniformly bounded on $[0,1]$ for all $\theta \in \Theta_K$ and $K$.
For a given $\theta$, the residual vector can be written as $\bm{E}(s; \theta) = (E_1(s; \theta)^\top, \ldots, E_T(s; \theta)^\top)^\top$, where
\begin{align}
    E_t(s; \theta)
    & = (e_{1t}(s; \theta), \ldots , e_{nt}(s; \theta))^\top \\
    e_{it}(s; \theta)
    & = Y_{it}(s) - \alpha(s; \theta) A_1(\bar Y_{it}, s) - \gamma(s; \theta) A_2( Y_{i,t-1}, s) - \sum_{j = 1}^{d_x} X_{it}^j \beta_j(s; \theta).
\end{align}
Under Assumptions \ref{as:inverse}(i), \ref{as:observables}, and \ref{as:A}, we have
\begin{align}\label{eq:error_bound}
    \left\| e_{it}(s; \theta) \right\|_p
    & \le \left\|Y_{it}(s)\right\|_p + |\alpha(s; \theta)| \sum_{j = 1}^n |w_{i,j}| \left\| A_1(Y_{jt}, s)\right\|_p + |\gamma(s; \theta)| \left\| A_2(Y_{i,t-1}, s)\right\|_p + \sum_{j = 1}^{d_x} |X_{it}^j| \cdot |\beta_j(s; \theta)| \\
    & \lesssim 1
\end{align}
for $p > 4$, uniformly in $s \in [0,1]$, $\theta \in \Theta_K$, and $(i,t)$.

Finally, for ease of reference, we provide a list of some basic facts below:
\small
\begin{align}
    \bm{D} \bm{E}(s; \theta) 
    & = \bm{D}\bm{H}(s) (\theta_0 - \theta) + \bm{D} \bm{V}(s) + \bm{D} \bm{\mathcal{E}}(s) \\
    \bm{Z}(s)^\top \bm{D}^\top \bm{D} \bm{E}(s; \theta) 
    & = \bm{Z}(s)^\top \bm{D}^\top \bm{D} \bm{H}(s) (\theta_0 - \theta) + \bm{Z}(s)^\top \bm{D}^\top \bm{D} \bm{V}(s) + \bm{Z}(s)^\top \bm{D}^\top \bm{D} \bm{\mathcal{E}}(s) \\
    \bm{Z}(s)^\top \bm{D}^\top \bm{D} \bm{E}(s; \theta_0) 
    & = \bm{Z}(s)^\top \bm{D}^\top \bm{D} \bm{V}(s) + \bm{Z}(s)^\top \bm{D}^\top \bm{D} \bm{\mathcal{E}}(s) \\
    \bm{E}(s; \theta)^\top  \bm{D}^\top P_m \bm{D} \bm{E}(s; \theta) 
    & =  (\theta_0 - \theta)^\top \bm{H}(s)^\top \bm{D}^\top P_m \bm{D} \bm{H}(s) (\theta_0 - \theta) + \bm{V}(s)^\top \bm{D}^\top P_m \bm{D} \bm{V}(s) \\
    & \quad + \bm{\mathcal{E}}(s)^\top \bm{D}^\top P_m \bm{D} \bm{\mathcal{E}}(s) + 2\bm{V}(s)^\top \bm{D}^\top P_m \bm{D} \bm{H}(s) (\theta_0 - \theta) \\
    & \quad + 2\bm{\mathcal{E}}(s)^\top \bm{D}^\top P_m \bm{D} \bm{H}(s) (\theta_0 - \theta) + 2\bm{V}(s)^\top \bm{D}^\top P_m \bm{D} \bm{\mathcal{E}}(s) \\
    \bm{E}(s; \theta_0)^\top  \bm{D}^\top P_m \bm{D} \bm{E}(s; \theta_0) 
    & =  \bm{V}(s)^\top \bm{D}^\top P_m \bm{D} \bm{V}(s) + \bm{\mathcal{E}}(s)^\top \bm{D}^\top P_m \bm{D} \bm{\mathcal{E}}(s) + 2\bm{V}(s)^\top \bm{D}^\top P_m \bm{D} \bm{\mathcal{E}}(s)
\end{align}
\normalsize
Empirical moment function:
\begin{align}
    g_{nT}(s; \theta) \coloneqq
    \frac{1}{N}\left(\begin{array}{c}
        \bm{Z}(s)^\top \bm{D}^\top \bm{D} \bm{E}(s; \theta) \\
        \bm{E}(s; \theta)^\top \bm{D}^\top P_1 \bm{D} \bm{E}(s; \theta) \\
        \vdots \\
        \bm{E}(s; \theta)^\top \bm{D}^\top P_M \bm{D} \bm{E}(s; \theta)
    \end{array}
    \right)
\end{align}
Jacobian of $g_{nT}(s; \theta)$:
\begin{align}\label{eq:Jmat}
    J_{nT}(s; \theta)
    & \coloneqq \frac{\partial g_{nT}(s; \theta)}{\partial \theta^\top} = -\frac{1}{N}\left(\begin{array}{c}
     \bm{Z}(s)^\top \bm{D}^\top \bm{D} \bm{H}(s) \\
     2\bm{E}(s; \theta)^\top \bm{D}^\top P_1 \bm{D} \bm{H}(s) \\
     \vdots \\
     2\bm{E}(s; \theta)^\top \bm{D}^\top P_M \bm{D} \bm{H}(s)
 \end{array}
 \right)
\end{align}
Decompose $\bar g_{nT}(\theta_0) = \bar g_{1,nT} + \bar g_{2,nT}$ with
\begin{align}
    \bar g_{1,nT}
    & \coloneqq \frac{1}{NL} \sum_{l=1}^L \left(\begin{array}{c}
     \bm{Z}(s_l)^\top \bm{D}^\top \bm{D} \bm{\mathcal{E}}(s_l) \\
     \bm{\mathcal{E}}(s_l)^\top \bm{D}^\top P_1 \bm{D} \bm{\mathcal{E}}(s_l) \\
     \vdots \\
     \bm{\mathcal{E}}(s_l)^\top \bm{D}^\top P_M \bm{D} \bm{\mathcal{E}}(s_l)
 \end{array}\right) \\
  \bar g_{2,nT} 
  & \coloneqq \frac{1}{NL} \sum_{l=1}^L \left(\begin{array}{c}
     \bm{Z}(s_l)^\top \bm{D}^\top \bm{D} \bm{V}(s_l) \\
     \bm{V}(s_l)^\top \bm{D}^\top P_1 \bm{D} \bm{V}(s_l) + 2 \bm{V}(s_l)^\top \bm{D}^\top P_1 \bm{D} \bm{\mathcal{E}}(s_l) \\
     \vdots \\
     \bm{V}(s_l)^\top \bm{D}^\top P_M \bm{D} \bm{V}(s_l) + 2 \bm{V}(s_l)^\top \bm{D}^\top P_M \bm{D} \bm{\mathcal{E}}(s_l)
 \end{array}\right) 
\end{align}
The variance-covariance matrix of $\sqrt{N}\bar g_{1,nT}$:
\begin{align}\label{eq:Vmat}
    \mathcal{V}_{nT} 
    \coloneqq N \bbE\left[\bar g_{1,nT} \bar g_{1,nT}^\top \right]
     = \left( \begin{array}{cccc}
    \mathcal{V}_{z, nT} & \bm{0}_{(d_q + d_x)K \times 1} & \cdots & \bm{0}_{(d_q + d_x)K \times 1} \\
     \bm{0}_{1 \times (d_q + d_x)K} &     \mathcal{V}_{11, nT} & \cdots & \mathcal{V}_{1M, nT} \\
     \vdots &  \vdots & \ddots & \vdots \\
    \bm{0}_{1 \times (d_q + d_x)K} &     \mathcal{V}_{M1, nT} & \cdots & \mathcal{V}_{MM, nT}.
    \end{array} \right)
\end{align}
Here, for convenience, let $\bm{n} \coloneqq nT$ and relabel each pair $(i,t)$ by $\mathrm{i} \coloneqq i + n(t - 1) \in [\bm{n}]$; for example, $(i,t) = (1,1) \iff \text{i} = 1$,  $(i,t) = (2,1) \iff \text{i} = 2$, \ldots, $(i,t) = (n,T) \iff \text{i} = \bm{n}$.
Then, we can write
\begin{align}\begin{split}\label{eq:Vz}
    \mathcal{V}_{z, nT}
    & \coloneqq \frac{N}{(NL)^2} \sum_{l = 1}^L \sum_{l' = 1}^L  \bm{Z}(s_l)^\top \bm{D}^\top \bm{D} \bbE[\bm{\mathcal{E}}(s_l) \bm{\mathcal{E}}(s_{l'})^\top]  \bm{D}^\top \bm{D} \bm{Z}(s_{l'}) \\
    & = \frac{1}{L^2 N} \sum_{l = 1}^L \sum_{l' = 1}^L  \sum_{\text{i} = 1}^{\bm n} z_{\text{i}}^\dagger(s_l) z_{\text{i}}^\dagger(s_{l'})^\top \Gamma_{\text{i}}(s_l, s_{l'}),
\end{split}\end{align}
$z_\text{i}^\dagger(s)$ denotes the $\text{i}$-th column of $\underbracket{\bm{Z}(s)^\top \bm{D}^\top \bm{D}}_{(d_q + d_x)K \times nT}$.
Furthermore, let $\underbracket{\bm{D}^\top P_m \bm{D}}_{nT \times nT} \eqqcolon \tilde P_m = (\tilde p_{m,\text{i},\text{j}})$, and 
\begin{align}
    \underbracket{D_1}_{(T-1) \times T} \coloneqq \left(\begin{array}{ccccc}
        -1 & 1 & 0  & \cdots & 0 \\
        0  & -1 & 1 & \cdots & 0 \\
        \vdots & \ddots & \ddots & \ddots & \vdots \\
        0 & \cdots & 0 & -1 & 1 
    \end{array}\right)
\end{align}
such that $\bm D = D_1 \otimes I_n$.
Then, we have 
\begin{align}
    \text{diag}(\tilde P_m) 
    & = \text{diag}\left( (D_1^\top \otimes I_n) (I_{T-1} \otimes P_{m,1}) (D_1 \otimes I_n) \right)\\
    & = \text{diag}\left( (D_1^\top  D_1) \otimes P_{m,1} \right) = \bm 0_{nT}.
\end{align}
Since $\tilde P_m$ is symmetric,
\begin{align}\begin{split}\label{eq:Vab}
    \mathcal{V}_{ab, nT}
    & \coloneqq \frac{N}{(NL)^2} \sum_{l=1}^L \sum_{l'=1}^L \bbE[ \bm{\mathcal{E}}(s_l)^\top \bm{D}^\top P_a \bm{D} \bm{\mathcal{E}}(s_l) \bm{\mathcal{E}}(s_{l'})^\top \bm{D}^\top P_b \bm{D} \bm{\mathcal{E}}(s_{l'}) ] \\
    & = \frac{1}{N} \sum_{1 \le \text{i}, \text{j}, \text{k}, \text{l} \le \bm n} \tilde p_{a, \text{i}, \text{j}} \tilde p_{b, \text{k}, \text{l}} \frac{1}{L^2} \sum_{l = 1}^L \sum_{l' = 1}^L \bbE[ \varepsilon_{\text{i}}(s_l) \varepsilon_{\text{j}}(s_l) \varepsilon_{\text{k}}(s_{l'}) \varepsilon_{\text{l}}(s_{l'}) ] \\
    & = \frac{1}{N} \sum_{1 \le \text{i}, \text{k} \le \bm n} \tilde p_{a, \text{i}, \text{k}} \tilde p_{b, \text{i}, \text{k}} \frac{1}{L^2} \sum_{l = 1}^L \sum_{l' = 1}^L \Gamma_{\text{i}}(s_l, s_{l'}) \Gamma_{\text{k}}(s_l, s_{l'}) \\
    & \quad + \frac{1}{N} \sum_{1 \le \text{i}, \text{k} \le \bm n} \tilde p_{a, \text{i}, \text{k}} \tilde p_{b, \text{k}, \text{i}} \frac{1}{L^2} \sum_{l = 1}^L \sum_{l' = 1}^L \Gamma_{\text{i}}(s_l, s_{l'}) \Gamma_{\text{k}}(s_l, s_{l'}) \\
    & = \frac{2}{N} \sum_{1 \le \text{i}, \text{k} \le \bm n} \tilde p_{a, \text{i}, \text{k}} \tilde p_{b, \text{i}, \text{k}} \frac{1}{L^2} \sum_{l = 1}^L \sum_{l' = 1}^L \Gamma_{\text{i}}(s_l, s_{l'}) \Gamma_{\text{k}}(s_l, s_{l'}) .
\end{split}\end{align}
Note that the cross terms between the linear and quadratic moments are zero:
\begin{align}
    & \frac{N}{(NL)^2} \sum_{l=1}^L \sum_{l'=1}^L \bbE[ \bm{Z}(s_l)^\top \bm{D}^\top \bm{D} \bm{\mathcal{E}}(s_l) \bm{\mathcal{E}}(s_{l'})^\top \bm{D}^\top P_m \bm{D} \bm{\mathcal{E}}(s_{l'})] \\
    & =  \frac{N}{(NL)^2} \sum_{l = 1}^L \sum_{l' = 1}^L \sum_{1 \le \text{i}, \text{j}, \text{k} \le \bm n} z_\text{i}^\dagger(s_l) \tilde p_{m, \text{j}, \text{k}} \bbE[ \varepsilon_\text{i}(s_l) \varepsilon_\text{j}(s_{l'}) \varepsilon_\text{k}(s_{l'}) ] = \bm{0}_{(d_q + d_x)K}.
\end{align}


\section{Proofs} \label{app:proofs}

\subsection{Results for Proposition \ref{prop:stationarity}}

\subsubsection{Proof of Proposition \ref{prop:stationarity}}

We first derive the condition for completeness.
Observe that
\begin{align}\label{eq:Lpinequality}
    \begin{split}
    \left\| \{\mathcal{A}_1 H\}_i \right\|_{L^2}
     = \left\| \alpha_0(\cdot) \sum_{j = 1}^n w_{i,j} A_1(h_j, \cdot) \right\|_{L^2}
    & \le \sum_{j = 1}^n |w_{i,j}|\left\| \alpha_0(\cdot)  A_1(h_j, \cdot) \right\|_{L^2} \\
    & = \sum_{j = 1}^n |w_{i,j}| \left(\int_0^1 \left|\alpha_0(s)  A_1(h_j, s) \right|^2 \text{d}s \right)^{1/2} \\
    & \le \bar \alpha_0 \sum_{j = 1}^n |w_{i,j}| \left\| A_1(h_j, \cdot) \right\|_{L^2} \\
    & \le \bar \alpha_0 ||W_n||_\infty \max_{1 \le j \le n}||h_j||_{L^2}
    \end{split}
\end{align}
for any $H \in \mathcal{H}_{n,2}$.
Taking the maximum over $i\in[n]$ in \eqref{eq:Lpinequality} yields
\begin{align}
    \left\|\mathcal A_1H\right\|_{\infty,2} \le \bar\alpha_0 \|W_n\|_\infty \| H \|_{\infty,2}.
\end{align}
Thus, if $\bar\alpha_0 \|W_n\|_\infty < 1$, $\mathcal A_1$ is a bounded linear operator on $\mathcal H_{n,2}$, and its operator norm satisfies $\|\mathcal A_1\|_{\mathrm{op}} < 1$.
It follows from the Neumann series theorem that $(\mathrm{Id} - \mathcal A_1)^{-1}$ exists and
\begin{align}
    \left\|(\mathrm{Id} - \mathcal A_1)^{-1}\right\|_{\mathrm{op}} \le \sum_{\ell = 0}^\infty \|\mathcal A_1\|_{\mathrm{op}}^\ell \le \frac{1}{1 - \bar\alpha_0\|W_n\|_\infty};
\end{align}
see, for example, Theorem 2.14 of \citet{Kress2014linear}.

Next, we establish stability.
By Assumption \ref{as:inverse}(ii),
\begin{align}
    \left\|\mathcal A_2H\right\|_{\infty,2}=\max_{i\in[n]}\left\|\gamma_0(\cdot)A_2(h_i,\cdot)\right\|_{L^2}\le\bar\gamma_0\max_{i\in[n]}\left\|A_2(h_i,\cdot)\right\|_{L^2}\le\bar\gamma_0\|H\|_{\infty,2}.
\end{align}
Therefore,
\begin{align}
    \left\|\mathcal A_3H\right\|_{\infty,2}\le\left\|(\mathrm{Id}-\mathcal A_1)^{-1}\right\|_{\mathrm{op}}\left\|\mathcal A_2H\right\|_{\infty,2}\le\frac{\bar\gamma_0}{1-\bar\alpha_0\|W_n\|_\infty}\|H\|_{\infty,2}.
\end{align}
Then, under $\bar\alpha_0 \|W_n\|_\infty + \bar\gamma_0 < 1$, we have $\bar\gamma_0/(1-\bar\alpha_0\|W_n\|_\infty) < 1$.
Thus, $\mathcal A_3$ is a contraction mapping on $\mathcal H_{n,2}$.

$\blacksquare$

\subsubsection{Alternative sufficient condition for the completeness}\label{subsub:alt}

Suppose that
\begin{align}\label{eq:alt_inverse}
\begin{array}{ll}
    \text{(i)} & ||\alpha_0||_{L^2}  || W_n ||_\infty < 1, \\
    \text{(ii)} & A_1(h,s) = \int_0^1 h(u) \nu_1(u,s) \text{d}u, \; \text{and} \;  \sup_{(u,s) \in [0,1]^2} |\nu_1(u,s)| \le 1 .
\end{array}
\end{align}
In contrast to the condition bounding $\bar{\alpha}_0 ||W_n ||_\infty$ in the main text, the network effect function $\alpha_0(s)$ does not need to be uniformly bounded in a small range to satisfy condition (i) in \eqref{eq:alt_inverse}.

Under these conditions, we obtain
\begin{align}
    \int_0^1 \left|\alpha_0(s) A_1(h,s)\right|^2 \text{d}s
    & = \int_0^1 |\alpha_0(s)|^2
        \left| \int_0^1 h(u) \nu_1(u,s) \text{d}u \right|^2
        \text{d}s \\
    & \le \int_0^1 |\alpha_0(s)|^2
        \int_0^1 |h(u)|^2 |\nu_1(u,s)|^2 \text{d}u
        \text{d}s \\
    & \le ||\alpha_0||_{L^2}^2  ||h||_{L^2}^2 .
\end{align}
Hence, similarly to \eqref{eq:Lpinequality}, we have
\begin{align}
    \left\| \{\mathcal{A}_1 H\}_i \right\|_{L^2}
    & \le \sum_{j = 1}^n |w_{i,j}| \left(\int_0^1 \left|\alpha_0(s) A_1(h_j, s) \right|^2 \text{d}s \right)^{1/2} \\ 
    & \le \left\|\alpha_0\right\|_{L^2} \left\| W_n \right\|_\infty \max_{1 \le j \le n}\left\| h_j \right\|_{L^2} < ||H||_{\infty, 2}
\end{align}
for any non-zero $H \in \mathcal{H}_{n,2}$.

\subsection{Proof of Theorem \ref{thm:roc}}


\subsubsection{Lemmas}

\begin{lemma}\label{lem:g0}
    Suppose that Assumptions \ref{as:inverse}(i), \ref{as:observables}, \ref{as:error}(i)--(ii), \ref{as:A}, \ref{as:weights}(i), and \ref{as:basis} hold.
    Then, $\left|\bbE[g_{nT}(s; \theta_0)]\right| \lesssim \bm{1}_{d_g} K^{- \varsigma}$ elementwise.
\end{lemma}

\begin{proof}
Observe that
\begin{align}
    \bbE[ g_{nT}(s; \theta_0)]
    & = \frac{1}{N}\left(\begin{array}{c}
     \bm{Z}(s)^\top \bm{D}^\top \bm{D} \bbE[\bm{V}(s)] \\
     \bbE[ \bm{V}(s)^\top \bm{D}^\top P_1 \bm{D} \bm{V}(s) ] +  2 \bbE[ \bm{V}(s)^\top \bm{D}^\top P_1 \bm{D} \bm{\mathcal{E}}(s)]\\
     \vdots \\
     \bbE[ \bm{V}(s)^\top \bm{D}^\top P_M \bm{D} \bm{V}(s) ] +  2 \bbE[ \bm{V}(s)^\top \bm{D}^\top P_M \bm{D} \bm{\mathcal{E}}(s)]
 \end{array}
 \right) 
\end{align}
By Assumptions \ref{as:observables}(ii) and \ref{as:A}, 
\begin{align}
    |A_k(\bbE[Y_{it}], s)|
    & \le \int_0^1 |\bbE[Y_{it}(u)]| \omega_1(u,s) \text{d}u \\
    & \le \int_0^1 \bbE|Y_{it}(u)| \omega_1(u,s) \text{d}u \lesssim 1
\end{align} 
uniformly in $s \in [0,1]$, implying that $\sup_{s \in [0,1]}|A_k(\bbE[Y_{it}], s)| \lesssim 1$.
Then, we have
\begin{align}
    |\bbE[v_{it}(s)]| 
    & \le \sum_{j = 1}^n |w_{i,j}| \cdot | A_1(\bbE[Y_{jt}], s) | \cdot |\alpha_0(s) - \alpha(s; \theta_0) | +  | A_2(\bbE[Y_{i,t-1}], s) | \cdot |\gamma_0(s) - \gamma(s; \theta_0) | \\
    & \quad + \sum_{j = 1}^{d_x} |X_{it}^j| \cdot |\beta_{0j}(s) - \beta_j(s; \theta_0)| \\
    & \lesssim K^{- \varsigma}
\end{align}
uniformly in $s \in [0,1]$ and $(i,t)$ under Assumption \ref{as:basis}.
This implies that the first $(d_q + d_x)K$ elements of $\bbE[g_{nT}(s; \theta_0)]$ are of order $K^{- \varsigma}$.

Next, by Cauchy-Schwarz inequality and the fact that $\left\|\bm{D}^\top P_m \bm{D}\right\|_{\mathrm{op}} \lesssim 1$ under Assumption \ref{as:weights}(i), we obtain
\begin{align}
    \bbE \left| \bm{V}(s)^\top \bm{D}^\top P_m \bm{D} \bm{V}(s) \right|
    & \le \sqrt{\bbE \left\| \bm{V}(s)^\top \bm{D}^\top P_m \bm{D}\right\|^2} \sqrt{\bbE \left\|\bm{V}(s)\right\|^2} \\
    & = \sqrt{\text{trace}\{ \bm{D}^\top P_m \bm{D} \bm{D}^\top P_m \bm{D} \bbE [\bm{V}(s) \bm{V}(s)^\top]\}} \sqrt{\bbE \left\|\bm{V}(s)\right\|^2} \\
    & \lesssim \sum_{t = 1}^T \sum_{i = 1}^n \bbE|v_{it}(s)|^2.
\end{align}
Similarly as above, by the $c_r$ inequality,
\begin{align}
    \bbE|v_{it}(s)|^2 
    & \le 4 \bbE \left|\sum_{j = 1}^n w_{i,j} A_1(Y_{jt}, s) [ \alpha_0(s) - \alpha(s; \theta_0)] \right|^2 + 4 \bbE \left| A_2(Y_{i,t-1}, s) [ \gamma_0(s) - \gamma(s; \theta_0)] \right|^2 \\
    & \quad + 2 \left|\sum_{j = 1}^{d_x} X_{it}^j (\beta_{0j}(s) - \beta_j(s; \theta_0)) \right|^2 \\
    & \le 4 \sum_{j = 1}^n \sum_{j' = 1}^n w_{i,j} w_{i,j'} \bbE[ A_1(Y_{jt}, s) A_1(Y_{j't}, s) ] [ \alpha_0(s) - \alpha(s; \theta_0)]^2  \\
    & \quad + 4 \bbE \left| A_2(Y_{i,t-1}, s) \right|^2 [ \gamma_0(s) - \gamma(s; \theta_0)]^2 + c K^{- 2\varsigma}.
\end{align}
By Cauchy-Schwarz inequality, 
\begin{align}
    | \bbE[ A_1(Y_{jt}, s) A_1(Y_{j't}, s) ] |
    & \le \left\|A_1(Y_{jt}, s) \right\|_2  \left\|A_1(Y_{j't}, s) \right\|_2.
\end{align}
Further, Assumptions \ref{as:observables}(ii) and \ref{as:A} imply that $\bbE |A_1(Y_{jt}, s) |^2 \le \int_0^1 \bbE|Y_{jt}(u)|^2 \omega_2(u,s) \text{d}u \lesssim 1$ uniformly in $s \in [0,1]$ and $(j,t)$. 
Thus, $| \bbE[ A_1(Y_{jt}, s) A_1(Y_{j't}, s) ] |$ is uniformly bounded.
We can see similarly that $\bbE |A_2(Y_{i,t-1}, s) |^2$ is uniformly bounded, and we have 
\begin{align}\label{eq:bias_order}
    \left\|v_{it}(s)\right\|_2 \lesssim K^{- \varsigma}
\end{align}
uniformly in $s \in [0, 1]$ and $(i, t)$.

Lastly, by Cauchy-Schwarz and Minkowski's inequalities,
\begin{align}
    \left|\bbE[ \bm{V}(s)^\top \bm{D}^\top P_m \bm{D} \bm{\mathcal{E}}(s)]\right|
    & = \left|\sum_{t \in [T - 1]}\sum_{1 \le i,j \le n} p_{m,i,j} \bbE[(v_{i,t+1}(s) - v_{it}(s))(\varepsilon_{j,t+1}(s) - \varepsilon_{jt}(s))]\right| \\
    & \le \sum_{t \in [T - 1]}\sum_{1 \le i,j \le n} |p_{m,i,j}| \bbE |(v_{i,t+1}(s) - v_{it}(s))(\varepsilon_{j,t+1}(s) - \varepsilon_{jt}(s))| \\
    & \le \sum_{t \in [T - 1]}\sum_{1 \le i,j \le n} |p_{m,i,j}| \cdot \left\|v_{i,t+1}(s) - v_{it}(s)\right\|_2 \left\|\varepsilon_{j,t+1}(s) - \varepsilon_{jt}(s)\right\|_2 \\
    & \le \sum_{t \in [T - 1]}\sum_{1 \le i,j \le n} |p_{m,i,j}| \cdot \left\{\left\|v_{i,t+1}(s)\right\|_2 + \left\|v_{it}(s)\right\|_2\right\} \left\{\left\|\varepsilon_{j,t+1}(s)\right\|_2 + \left\|\varepsilon_{jt}(s)\right\|_2\right\} \\
    & \lesssim N K^{- \varsigma},
\end{align}
where the last inequality follows from \eqref{eq:bias_order} and Assumptions \ref{as:error}(ii) and \ref{as:weights}(i).
Combining these results gives the desired result.
\end{proof}
\bigskip


Denote the population GMM objective function as follows:
\begin{align}
    \mathcal{Q}^*_{nT}(\theta) \coloneqq \bbE[\bar g_{nT}(\theta)]^\top \Omega_{nT} \bbE[\bar g_{nT}(\theta)]
\end{align}

\begin{lemma}\label{lem:identification}
    Suppose that Assumptions \ref{as:inverse}(i), \ref{as:observables}, \ref{as:error}(i) -- (ii), \ref{as:A}, \ref{as:weights}, \ref{as:matrix1}, and \ref{as:basis} hold.
    In addition, assume that $K^{(1 - 2\varsigma)/2} \to 0$ as $nT \to \infty$.
    Then, for any $\theta \in \Theta_K$ and $e >0$ such that $\left\| \theta - \theta_0 \right\| \ge e$, there exists a constant $c_e > 0$ such that $\mathcal{Q}^*_{nT}(\theta) - \mathcal{Q}^*_{nT}(\theta_0) > c_e$ for all sufficiently large $nT$.
\end{lemma}

\begin{proof}
Decompose
\begin{align}
    \mathcal{Q}^*_{nT}(\theta) - \mathcal{Q}^*_{nT}(\theta_0)
    & = \underbracket{\left( \bbE[\bar g_{nT}(\theta)] - \bbE[\bar g_{nT}(\theta_0)]\right)^\top \Omega_{nT} \left( \bbE[\bar g_{nT}(\theta)] - \bbE[\bar g_{nT}(\theta_0)] \right)}_{\eqqcolon A_{nT}(\theta)} \\
    & \quad + 2 \left( \bbE[\bar g_{nT}(\theta)] - \bbE[\bar g_{nT}(\theta_0)]\right)^\top \Omega_{nT} \bbE[\bar g_{nT}(\theta_0)]
\end{align}
In view of
\begin{align}
    &\bbE[\bar g_{nT}(\theta)] - \bbE[\bar g_{nT}(\theta_0)] \\
    &  = \frac{1}{NL}\sum_{l=1}^L \left(\begin{array}{c}
       \bm{Z}(s_l)^\top \bm{D}^\top \bm{D} \bbE[\bm{H}(s_l)] \\
       \bbE \left[ (\bm{H}(s_l)(\theta_0 - \theta) + 2\bm{V}(s_l) + 2\bm{\mathcal{E}}(s_l))^\top \bm{D}^\top P_1 \bm{D} \bm{H}(s_l) \right] \\
       \vdots \\
       \bbE\left[ (\bm{H}(s_l)(\theta_0 - \theta) + 2\bm{V}(s_l) + 2\bm{\mathcal{E}}(s_l))^\top \bm{D}^\top P_M \bm{D} \bm{H}(s_l) \right]
   \end{array}
   \right) (\theta_0 - \theta),
\end{align}
we can find that $A_{nT}(\theta)$ is bounded below from $\lambda_{\min}(\Omega_{nT}) \lambda_{\min}(\Pi_{nT}^\top \Pi_{nT}) ||\theta_0 - \theta||^2 \ge c_1 e^2$ for some $c_1 > 0$ for all sufficiently large $nT$, under Assumptions \ref{as:weights}(ii) and \ref{as:matrix1}.
Further, Cauchy-Schwarz inequality and Lemma \ref{lem:g0} give that
\begin{align}
    \left| \left( \bbE[\bar g_{nT}(\theta)] - \bbE[\bar g_{nT}(\theta_0)]\right)^\top \Omega_{nT} \bbE[\bar g_{nT}(\theta_0)] \right|
    & \le (A_{nT}(\theta))^{1/2} \left( \bbE[\bar g_{nT}(\theta_0)]^\top \Omega_{nT} \bbE[\bar g_{nT}(\theta_0)] \right)^{1/2} \\
    & \le c_2 (A_{nT}(\theta))^{1/2} K^{(1 - 2\varsigma)/2}
\end{align}
Hence, since $(A_{nT}(\theta))^{1/2}$ is bounded below from zero and $K^{(1 - 2\varsigma)/2} \to 0$, we have
\begin{align}
    \mathcal{Q}^*_{nT}(\theta) - \mathcal{Q}^*_{nT}(\theta_0)
    & \ge A_{nT}(\theta) - 2 c_2 (A_{nT}(\theta))^{1/2} K^{(1 - 2\varsigma)/2} \\
    & = (A_{nT}(\theta))^{1/2} ((A_{nT}(\theta))^{1/2} - 2 c_2 K^{(1 - 2\varsigma)/2}) > c_e
\end{align}
for all sufficiently large $nT$.
This completes the proof.
\end{proof}


\begin{lemma}\label{lem:NED1}
    Suppose that Assumptions \ref{as:inverse}, \ref{as:sample_space}, \ref{as:observables}, and \ref{as:A} hold.
    Then, for any given $s \in [0,1]$, $\{ Y_{it}(s) : (i, t) \in \mathcal{D}_{nT}^{\mathrm{ST}}\}$ is uniformly and geometrically $L^2$-NED on $\{ \varepsilon_{it} : (i, t) \in \mathcal{D}_{nT}^{\mathrm{ST}}\}$ with respect to the spatial-time distance $\Delta_{\mathrm{ST}}$.
\end{lemma}

\begin{proof}
    We prove the lemma in a similar manner to \cite{jenish2012nonparametric} and \cite{hoshino2022sieve}.
    Under Assumption \ref{as:inverse}, the reduced-form representation of the model implies that, conditional on the non-stochastic initial outcomes and covariates, each $Y_{it}$ is a measurable function of the entire error array.
    Therefore, we can write $Y_{it} = \xi_{it} ( \{ \varepsilon_{ju} : (j, u) \in \mathcal{D}_{nT}^{\mathrm{ST}} \} )$.

    For some $\delta > 0$, define
    \begin{align}
        \mathcal{E}_{1,it}^{(\delta)}
        \coloneqq \left\{ \varepsilon_{ju} : \Delta_{\mathrm{ST}} \left( (i, t), (j, u) \right) \le \delta \right\}, \;\;
        \mathcal{E}_{2,it}^{(\delta)}
        \coloneqq \left\{ \varepsilon_{ju} : \Delta_{\mathrm{ST}} \left( (i, t), (j, u) \right) > \delta \right\}.
    \end{align}
    Since $L^2(0,1)$ is separable, both $\mathcal{E}_{1,it}^{(\delta)}$ and $\mathcal{E}_{2,it}^{(\delta)}$ are Polish space-valued random elements.
    Thus, by Lemma 2.11 of \cite{dudley1983invariance} (see also Lemma A.1 of \citealp{jenish2012nonparametric}), there exists a function $\chi$ such that
    \begin{align}
        \left( \mathcal{E}_{1,it}^{(\delta)}, \chi \left( U, \mathcal{E}_{1,it}^{(\delta)} \right) \right)
        \overset{d}{=}
        \left( \mathcal{E}_{1,it}^{(\delta)}, \mathcal{E}_{2,it}^{(\delta)} \right),
    \end{align}
    where $U$ is uniformly distributed on $[0,1]$ and independent of $\mathcal{E}_{1,it}^{(\delta)}$.
    With a slight abuse of notation, write
    \begin{align}
        Y_{it}
        & = \xi_{it} \left( \mathcal{E}_{1,it}^{(\delta)}, \mathcal{E}_{2,it}^{(\delta)} \right), \\
        Y_{it}^{(\delta)}
        & \coloneqq \xi_{it} \left( \mathcal{E}_{1,it}^{(\delta)}, \chi \left( U, \mathcal{E}_{1,it}^{(\delta)} \right) \right).
    \end{align}
    Let $\{ \varepsilon_{ju}^{(\delta)} \}$ denote the coupled error array corresponding to $\left( \mathcal{E}_{1,it}^{(\delta)}, \chi \left( U, \mathcal{E}_{1,it}^{(\delta)} \right) \right)$ and $\mathcal{E}_t^{(\delta)} \coloneqq \left( \varepsilon_{1t}^{(\delta)}, \ldots, \varepsilon_{nt}^{(\delta)} \right)^\top$.
    By construction, $\varepsilon_{ju}^{(\delta)} = \varepsilon_{ju}$ whenever $\Delta_{\mathrm{ST}} \left( (i, t), (j, u) \right) \le \delta$.

    The coupled outcome process is given by
    \begin{align}
        Y_u^{(\delta)}
        = \mathcal{A}_3^u Y_0 + \sum_{\ell = 0}^{u-1} \mathcal{A}_3^\ell (\text{Id} - \mathcal{A}_1)^{-1} \left[ X_{u-\ell}\beta_0 + F_0 + \mathcal{E}_{u-\ell}^{(\delta)} \right], \qquad u \in [T].
    \end{align}
    In particular, its $(j,u)$-th component satisfies
    \begin{align}
        Y_{ju}^{(\delta)}(s)
        & = \alpha_0(s) \sum_{k = 1}^n w_{j,k} A_1 ( Y_{ku}^{(\delta)}, s) + \gamma_0(s) A_2 ( Y_{j,u-1}^{(\delta)}, s ) + X_{ju}^\top \beta_0(s) + f_{0j}(s) + \varepsilon_{ju}^{(\delta)}(s),
    \end{align}
    where $Y_{j0}^{(\delta)} = Y_{j0}$.
    By construction,
    \begin{align}
        \bbE [ Y_{it}(s) \mid \mathcal{F}_{nT,it}^{\mathrm{ST}}(\delta) ]
        & = \bbE [ \xi_{it} ( \mathcal{E}_{1,it}^{(\delta)}, \mathcal{E}_{2,it}^{(\delta)} )(s) \mid \mathcal{E}_{1,it}^{(\delta)} ] \\
        & = \bbE [ \xi_{it} ( \mathcal{E}_{1,it}^{(\delta)}, \chi ( U, \mathcal{E}_{1,it}^{(\delta)} ) )(s) \mid \mathcal{E}_{1,it}^{(\delta)} ] \\
        & = \bbE [ Y_{it}^{(\delta)}(s) \mid \mathcal{F}_{nT,it}^{\mathrm{ST}}(\delta) ],
    \end{align}
    where $\mathcal{F}_{nT,it}^{\mathrm{ST}}(\delta)$ is the $\sigma$-field generated by $\mathcal{E}_{1,it}^{(\delta)}$.
    Moreover, since $Y_{ju}^{(\delta)}$ and $Y_{ju}$ have the same marginal distribution,
    \begin{align}
        \left\| Y_{ju}(s) - Y_{ju}^{(\delta)}(s) \right\|_2^2
        & \le 2 \left\| Y_{ju}(s) \right\|_2^2
        + 2 \left\| Y_{ju}^{(\delta)}(s) \right\|_2^2 = 4 \left\| Y_{ju}(s) \right\|_2^2.
    \end{align}
    Hence, by Assumption \ref{as:observables}(ii), there exists a finite constant $C_Y$ such that
    \begin{align}
        \sup_{n \ge 1} \max_{j \in [n], \; u \in \{0\} \cup [T]} \left\| Y_{ju}(s) - Y_{ju}^{(\delta)}(s) \right\|_2
        \le C_Y
    \end{align}
    uniformly in $s \in [0,1]$.

    We next evaluate how the replacement of errors outside the spatial-time $\delta$-neighborhood affects $Y_{it}(s)$.
    Subtracting the equations for $Y_{ju}$ and $Y_{ju}^{(\delta)}$ gives
    \begin{align}
        Y_{ju}(s) - Y_{ju}^{(\delta)}(s)
        & = \alpha_0(s) \sum_{k = 1}^n w_{j,k} A_1 ( Y_{ku} - Y_{ku}^{(\delta)}, s ) + \gamma_0(s) A_2 ( Y_{j,u-1} - Y_{j,u-1}^{(\delta)}, s ) \\
        & \quad + \varepsilon_{ju}(s) - \varepsilon_{ju}^{(\delta)}(s)
    \end{align}
    for general $\delta > 0$.
    Let 
    \begin{align}
    \bar\Delta_{\mathrm{ST}} \coloneqq \max \{ \bar\Delta, 1 \}.
    \end{align}
    First, suppose that $0 < \delta < \bar\Delta_{\mathrm{ST}}$.
    Since $\varepsilon_{it}^{(\delta)} = \varepsilon_{it}$ in this case, we have
    \begin{align}
        Y_{it}(s) - Y_{it}^{(\delta)}(s) = \alpha_0(s) \sum_{j = 1}^n w_{i,j} A_1 ( Y_{jt} - Y_{jt}^{(\delta)}, s ) + \gamma_0(s) A_2 ( Y_{i,t-1} - Y_{i,t-1}^{(\delta)}, s ).
    \end{align}
    Hence, by Minkowski's inequality and Assumption \ref{as:A},
    \begin{align}
        \left\| Y_{it}(s) - Y_{it}^{(\delta)}(s) \right\|_2 
        & \le \bar\alpha_0 \sum_{j = 1}^n |w_{i,j}| \left( \int_0^1 \left\| Y_{jt}(v) - Y_{jt}^{(\delta)}(v) \right\|_2^2 \omega_2(v,s) \text{d}v \right)^{1/2} \\
        & \quad + \bar\gamma_0 \left( \int_0^1 \left\| Y_{i,t-1}(v) - Y_{i,t-1}^{(\delta)}(v) \right\|_2^2 \omega_2(v,s) \text{d}v \right)^{1/2} \\
        & \le C_Y \left( \bar\alpha_0 \left\| W_n \right\|_\infty + \bar\gamma_0 \right).
    \end{align}
    Next, suppose that $\bar\Delta_{\mathrm{ST}} \le \delta < 2\bar\Delta_{\mathrm{ST}}$.
    In this case, $\varepsilon_{jt}^{(\delta)} = \varepsilon_{jt}$ for every $j$ such that $w_{i,j} \neq 0$, and $\varepsilon_{i,t-1}^{(\delta)} = \varepsilon_{i,t-1}$.
    Therefore, this gives
    \begin{align}
        Y_{it}(s) - Y_{it}^{(\delta)}(s)
        & = \alpha_0(s) \sum_{j = 1}^n w_{i,j} A_1 \left( \alpha_0(\cdot) \sum_{k = 1}^n w_{j,k} A_1 ( Y_{kt} - Y_{kt}^{(\delta)}, \cdot ) + \gamma_0(\cdot) A_2 ( Y_{j,t-1} - Y_{j,t-1}^{(\delta)}, \cdot), s \right) \\
        & \quad + \gamma_0(s) A_2 \left( \alpha_0(\cdot) \sum_{k = 1}^n w_{i,k} A_1 ( Y_{k,t-1} - Y_{k,t-1}^{(\delta)}, \cdot ) + \gamma_0(\cdot) A_2 ( Y_{i,t-2} - Y_{i,t-2}^{(\delta)}, \cdot), s \right).
    \end{align}
    Again, applying Minkowski's inequality and Assumption \ref{as:A} to the preceding expansion yields
    \begin{align}
        \left\| Y_{it}(s) - Y_{it}^{(\delta)}(s) \right\|_2 & \le C_Y \left\{ \left( \bar\alpha_0 \left\| W_n \right\|_\infty \right)^2 + 2\bar\alpha_0 \left\| W_n \right\|_\infty \bar\gamma_0 + \bar\gamma_0^2 \right\} \\
        & = C_Y \left( \bar\alpha_0 \left\| W_n \right\|_\infty + \bar\gamma_0 \right)^2.
    \end{align}

    Applying the same argument recursively, if $m\bar\Delta_{\mathrm{ST}} \le \delta < (m+1)\bar\Delta_{\mathrm{ST}}$ for some nonnegative integer $m$, all innovation differences encountered during the first $m$ recursive expansions are zero.
    We therefore obtain
    \begin{align}
        \left\| Y_{it}(s) - Y_{it}^{(\delta)}(s) \right\|_2 \le C_Y \left( \bar\alpha_0 \left\| W_n \right\|_\infty + \bar\gamma_0 \right)^{m+1}.
    \end{align}
    Let $\varrho_0 \coloneqq \sup_{n \ge 1} \left\{ \bar\alpha_0 \left\| W_n \right\|_\infty + \bar\gamma_0 \right\} < 1$, where the inequality follows from Assumption \ref{as:inverse}(i).
    Since $m = \left\lfloor \delta / \bar\Delta_{\mathrm{ST}} \right\rfloor$, it follows that
    \begin{align}\label{eq:qdiff}
        \left\| Y_{it}(s) - Y_{it}^{(\delta)}(s) \right\|_2 \le C_Y \varrho_0^{\left\lfloor \delta / \bar\Delta_{\mathrm{ST}} \right\rfloor + 1}.
    \end{align}

    Finally, by Jensen's inequality and \eqref{eq:qdiff},
    \begin{align}
        \left\| Y_{it}(s) - \bbE [ Y_{it}(s) \mid \mathcal{F}_{nT,it}^{\mathrm{ST}}(\delta) ] \right\|_2
        & = \left\| \int_0^1 \left[ \xi_{it} ( \mathcal{E}_{1,it}^{(\delta)}, \mathcal{E}_{2,it}^{(\delta)} )(s) - \xi_{it} ( \mathcal{E}_{1,it}^{(\delta)}, \chi (v, \mathcal{E}_{1,it}^{(\delta)}) )(s) \right] \text{d}v \right\|_2 \\
        & \le \left\{ \bbE \int_0^1 \left| \xi_{it} ( \mathcal{E}_{1,it}^{(\delta)}, \mathcal{E}_{2,it}^{(\delta)} )(s) - \xi_{it} ( \mathcal{E}_{1,it}^{(\delta)}, \chi (v, \mathcal{E}_{1,it}^{(\delta)}) )(s) \right|^2 \text{d}v \right\}^{1/2} \\
        & = \left\| \xi_{it} ( \mathcal{E}_{1,it}^{(\delta)}, \mathcal{E}_{2,it}^{(\delta)} )(s) - \xi_{it} ( \mathcal{E}_{1,it}^{(\delta)}, \chi (U, \mathcal{E}_{1,it}^{(\delta)}) )(s) \right\|_2 \\
        & = \left\| Y_{it}(s) - Y_{it}^{(\delta)}(s) \right\|_2 \\
        & \le C_Y \varrho_0^{\left\lfloor \delta / \bar\Delta_{\mathrm{ST}} \right\rfloor + 1} \to 0
    \end{align}
    as $\delta \to \infty$.
    This proves the desired result.
\end{proof}


\begin{lemma}\label{lem:NED2}
    Suppose that Assumptions \ref{as:inverse}, \ref{as:sample_space}, \ref{as:observables}, and \ref{as:A} hold.
    Then, for any given $s \in [0, 1]$, $\{A_1(\bar Y_{it}, s) : (i, t) \in \mathcal{D}_{nT}^\text{ST}\}$ and $\{A_2(Y_{i,t-1}, s) : (i, t) \in \mathcal{D}_{nT}^\text{ST}\}$ are uniformly and geometrically $L^2$-NED on $\{\varepsilon_{it}: (i,t) \in \mathcal{D}_{nT}^\text{ST}\}$.
\end{lemma}

\begin{proof}
Note that $\mathcal{F}_{nT,jt}^{\text{ST}}((\delta - 1) \bar\Delta_{\text{ST}}) \subseteq \mathcal{F}_{nT,it}^{\text{ST}}(\delta \bar\Delta_{\text{ST}})$ for $(i,j)$ with $\Delta(i,j) \le \bar\Delta$ and $\delta > 1$.
Thus, by Lemma \ref{lem:NED1},
\begin{align}
    \left\| \bar Y_{it}(u) - \bbE \left[ \bar Y_{it}(u) \mid \mathcal{F}_{nT,it}^{\text{ST}}(\delta \bar\Delta_{\text{ST}}) \right] \right\|_2
    & \le \sum_{j = 1}^n |w_{i,j}| \left\| Y_{jt}(u) - \bbE \left[ Y_{jt}(u) \mid \mathcal{F}_{nT,it}^{\text{ST}}(\delta \bar\Delta_{\text{ST}}) \right] \right\|_2 \\
    & \le \sum_{j : \Delta(i,j) \le \bar\Delta} |w_{i,j}| \left\| Y_{jt}(u) - \bbE \left[ Y_{jt}(u) \mid \mathcal{F}_{nT,jt}^{\text{ST}}((\delta - 1) \bar\Delta_{\text{ST}}) \right] \right\|_2 \\
    & \lesssim \varrho^{\lfloor \delta \rfloor}
\end{align}
uniformly in $u \in [0,1]$, which implies that $\{\bar Y_{it}(u)\}$ is uniformly and geometrically $L^2$-NED.

By the linearity of $A_1$ and Assumption \ref{as:A},
\begin{align}
    \left\| A_1(\bar Y_{it}, s) - \bbE \left[ A_1(\bar Y_{it}, s) \mid \mathcal{F}_{nT,it}^{\text{ST}}(\delta) \right] \right\|_2
    & = \left\| A_1 \left( \bar Y_{it} - \bbE \left[ \bar Y_{it} \mid \mathcal{F}_{nT,it}^{\text{ST}}(\delta) \right], s \right) \right\|_2 \\
    & \le \left( \int_0^1 \left\| \bar Y_{it}(u) - \bbE \left[ \bar Y_{it}(u) \mid \mathcal{F}_{nT,it}^{\text{ST}}(\delta) \right] \right\|_2^2 \omega_2(u,s) \text{d}u \right)^{1/2} \\
    & \lesssim \varrho^{\left\lfloor \delta / \bar\Delta_{\text{ST}} \right\rfloor}.
\end{align}
This proves that $\{A_1(\bar Y_{it}, s) : (i,t) \in \mathcal{D}_{nT}^{\text{ST}}\}$ is uniformly and geometrically $L^2$-NED.

Similarly, since $\mathcal{F}_{nT,i,t-1}^{\text{ST}}((\delta - 1) \bar\Delta_{\text{ST}}) \subseteq \mathcal{F}_{nT,it}^{\text{ST}}(\delta \bar\Delta_{\text{ST}})$, the same argument applied to $Y_{i,t-1}$ and $A_2$ shows that $\{A_2(Y_{i,t-1}, s) : (i,t) \in \mathcal{D}_{nT}^{\text{ST}}\}$ is uniformly and geometrically $L^2$-NED.
This completes the proof.
\end{proof}

As a useful consequence from Lemma \ref{lem:NED1} and \ref{lem:NED2}, we have
\begin{align}
    \left\| e_{it}(s; \theta) - \bbE [ e_{it}(s; \theta) \mid \mathcal{F}_{nT,it}^{\text{ST}}(\delta)] \right\|_2
    & \le \left\| Y_{it}(s) - \bbE [ Y_{it}(s) \mid \mathcal{F}_{nT,it}^{\text{ST}}(\delta)] \right\|_2 \\
    & \quad + |\alpha(s; \theta)| \cdot \left\| A_1(\bar Y_{it}, s) - \bbE [ A_1(\bar Y_{it},s) \mid \mathcal{F}_{nT,it}^{\text{ST}}(\delta)] \right\|_2 \\
    & \quad + |\gamma(s; \theta)| \cdot \left\| A_2(Y_{i,t-1}, s) - \bbE [ A_2(Y_{i,t-1},s) \mid \mathcal{F}_{nT,it}^{\text{ST}}(\delta)] \right\|_2 \\
    & \lesssim \varrho^{\lfloor \delta / \bar \Delta_{\text{ST}} \rfloor}
\end{align}
uniformly in $s \in [0,1]$, $\theta \in \Theta_K$, and $(i, t)$; that is, $\{e_{it}(s; \theta)\}$ is uniformly and geometrically $L^2$-NED.


\begin{lemma}\label{lem:cov}
    Suppose that Assumption \ref{as:error}(i) holds.
    Let $\{\xi_{it}: (i,t) \in \mathcal{D}_{nT}^{\text{ST}}\}$ be a uniformly and geometrically $L^2$-NED random field on $\{\varepsilon_{it}: (i,t) \in \mathcal{D}_{nT}^{\text{ST}}\}$.
    Denote $\vec \xi_{it} \coloneqq \xi_{i,t+1} - \xi_{it}$ and $C_\xi \coloneqq \sup_{i,t} \left\| \xi_{it} \right\|_2$.
    Then, for all $i,j \in \mathcal{D}_n$ and $t,u \in [T-1]$,
    \begin{align}
        \left| \text{Cov}\left(\vec \xi_{it},\vec \xi_{ju}\right) \right| \lesssim C_\xi^2 \rho\left(\Delta_{\text{ST}}\left((i,t),(j,u)\right)/3\right)
    \end{align}
    with some geometric NED coefficient $\rho$.
\end{lemma}

\begin{proof}
    Decompose $\vec \xi_{it} = \vec \xi_{1,it}^{(\delta)} + \vec \xi_{2,it}^{(\delta)}$, where
    \begin{align}
        \vec \xi_{1,it}^{(\delta)} \coloneqq \bbE \left[ \vec \xi_{it} \mid \mathcal{F}_{nT,it}^{+}(\delta) \right], \;\; \text{and} \;\; \vec \xi_{2,it}^{(\delta)} \coloneqq \vec \xi_{it} - \bbE \left[ \vec \xi_{it} \mid \mathcal{F}_{nT,it}^{+}(\delta) \right],
    \end{align}
    where $\mathcal{F}_{nT,it}^{+}(\delta) \coloneqq \mathcal{F}_{nT,it}^{\text{ST}}(\delta) \lor \mathcal{F}_{nT,i,t+1}^{\text{ST}}(\delta)$.

    For each pair $\vec \xi_{it}$ and $\vec \xi_{ju}$, denote $\delta_{it,ju} \coloneqq \Delta_{\text{ST}}((i,t),(j,u))/3$.
    If $\Delta_{\text{ST}}((i,t),(j,u)) \le 3$, the desired result follows immediately from the Cauchy--Schwarz inequality and $\left\| \vec \xi_{it} \right\|_2 \le 2C_\xi$.
    Hence, suppose that $\Delta_{\text{ST}}((i,t),(j,u)) > 3$.

    The minimum distance between a point in $\{(i,t),(i,t+1)\}$ and a point in $\{(j,u),(j,u+1)\}$ is at least $\Delta_{\text{ST}}((i,t),(j,u))-1$.
    Since
    \begin{align}
        2\delta_{it,ju} = \frac{2}{3}\Delta_{\text{ST}}((i,t),(j,u)) < \Delta_{\text{ST}}((i,t),(j,u))-1,
    \end{align}
    $\mathcal{F}_{nT,it}^{+}(\delta_{it,ju})$ and $\mathcal{F}_{nT,ju}^{+}(\delta_{it,ju})$ are generated by disjoint collections of errors and are therefore independent under Assumption \ref{as:error}(i).

    It follows that
    \begin{align}
        \left| \text{Cov}\left(\vec \xi_{it},\vec \xi_{ju}\right) \right|
        & = \left| \text{Cov}\left(\vec \xi_{1,it}^{(\delta_{it,ju})}+\vec \xi_{2,it}^{(\delta_{it,ju})},\vec \xi_{1,ju}^{(\delta_{it,ju})}+\vec \xi_{2,ju}^{(\delta_{it,ju})}\right) \right| \\
        & \le \left| \text{Cov}\left(\vec \xi_{1,it}^{(\delta_{it,ju})},\vec \xi_{1,ju}^{(\delta_{it,ju})}\right) \right|+\left| \text{Cov}\left(\vec \xi_{1,it}^{(\delta_{it,ju})},\vec \xi_{2,ju}^{(\delta_{it,ju})}\right) \right| \\
        & \quad+\left| \text{Cov}\left(\vec \xi_{2,it}^{(\delta_{it,ju})},\vec \xi_{1,ju}^{(\delta_{it,ju})}\right) \right|+\left| \text{Cov}\left(\vec \xi_{2,it}^{(\delta_{it,ju})},\vec \xi_{2,ju}^{(\delta_{it,ju})}\right) \right|.
    \end{align}
    The first term on the right-hand side is zero by Assumption \ref{as:error}(i).
    By Jensen's and triangle inequalities,
    \begin{align}
        \left\| \vec \xi_{1,it}^{(\delta_{it,ju})} \right\|_2 \le \left\| \vec \xi_{it} \right\|_2 \le \left\| \xi_{i,t+1} \right\|_2+\left\| \xi_{it} \right\|_2 \le 2C_\xi.
    \end{align}
    In addition, since $\{\xi_{it}\}$ is uniformly and geometrically $L^2$-NED,
    \begin{align}
        \left\| \vec \xi_{2,it}^{(\delta_{it,ju})} \right\|_2
        & = \left\| \vec \xi_{it}-\bbE \left[ \vec \xi_{it} \mid \mathcal{F}_{nT,it}^{+}(\delta_{it,ju}) \right] \right\|_2 \\
        & \le \left\| \xi_{i,t+1}-\bbE \left[ \xi_{i,t+1} \mid \mathcal{F}_{nT,i,t+1}^{\text{ST}}(\delta_{it,ju}) \right] \right\|_2+\left\| \xi_{it}-\bbE \left[ \xi_{it} \mid \mathcal{F}_{nT,it}^{\text{ST}}(\delta_{it,ju}) \right] \right\|_2 \\
        & \lesssim C_\xi \rho(\delta_{it,ju}).
    \end{align}
    Hence, the Cauchy--Schwarz inequality gives
    \begin{align}
        \left| \text{Cov}\left(\vec \xi_{1,it}^{(\delta_{it,ju})},\vec \xi_{2,ju}^{(\delta_{it,ju})}\right) \right| \lesssim C_\xi^2 \rho(\delta_{it,ju}).
    \end{align}
    The same inequality applies to $\left| \text{Cov}\left(\vec \xi_{2,it}^{(\delta_{it,ju})},\vec \xi_{1,ju}^{(\delta_{it,ju})}\right) \right|$.
    Furthermore,
    \begin{align}
        \left| \text{Cov}\left(\vec \xi_{2,it}^{(\delta_{it,ju})},\vec \xi_{2,ju}^{(\delta_{it,ju})}\right) \right| \lesssim C_\xi^2 \rho(\delta_{it,ju})^2 \lesssim C_\xi^2 \rho(\delta_{it,ju}).
    \end{align}
    This completes the proof.
\end{proof}


\begin{lemma}\label{lem:matLLN}
    Suppose that Assumptions \ref{as:inverse}, \ref{as:sample_space} -- \ref{as:weights}, and \ref{as:basis} hold.
    Uniformly in $m \in [M]$ and $\theta \in \Theta_K$,
    \begin{itemize}
        \item[(i)] $\left\| \sum_{l = 1}^L \bm{Z}(s_l)^\top \bm{D}^\top \bm{D} (\bm{H}(s_l) - \bbE[\bm{H}(s_l)]) (\theta_0 - \theta) / (NL) \right\| \lesssim_p \sqrt{K}/\sqrt{nT}$
        \item[(ii)] $\left\| \sum_{l = 1}^L \bm{Z}(s_l)^\top \bm{D}^\top \bm{D} (\bm{H}(s_l) - \bbE[\bm{H}(s_l)]) / (NL) \right\| \lesssim_p K /\sqrt{nT}$
        \item[(iii)] $\left\| \sum_{l = 1}^L \bm{Z}(s_l)^\top \bm{D}^\top \bm{D} \bm{V}(s_l) / (NL) \right\| \lesssim_p K^{-\varsigma}$, $\left\| \sum_{l = 1}^L \bm{Z}(s_l)^\top \bm{D}^\top \bm{D} \bbE[\bm{V}(s_l)] / (NL) \right\|\lesssim K^{-\varsigma}$
        \item[(iv)] $\left\| \sum_{l = 1}^L \bm{Z}(s_l)^\top \bm{D}^\top \bm{D} \bm{\mathcal{E}}(s_l) / (NL) \right\| \lesssim_p 1/\sqrt{nT}$
        \item[(v)] $\left| \sum_{l = 1}^L \left\{ \bm{E}(s_l; \theta)^\top \bm{D}^\top P_m \bm{D} \bm{E}(s_l; \theta) - \bbE[\bm{E}(s_l; \theta)^\top \bm{D}^\top P_m \bm{D} \bm{E}(s_l; \theta)] \right\} / (NL) \right| \lesssim_p 1/\sqrt{nT}$
        \item[(vi)] $\left| \sum_{l = 1}^L \bm{V}(s_l)^\top \bm{D}^\top P_m \bm{D} \bm{V}(s_l) / (NL) \right|  \lesssim_p K^{- 2\varsigma}$
        \item[(vii)] $\left| \sum_{l=1}^L \bm{V}(s_l)^\top \bm{D}^\top P_m \bm{D} \bm{\mathcal{E}}(s_l) / (NL) \right| \lesssim_p K^{- \varsigma}$
        \item[(viii)] $\left| \sum_{l = 1}^L \bm{\mathcal{E}}(s_l)^\top \bm{D}^\top P_m \bm{D} \bm{\mathcal{E}}(s_l) / (NL)  \right| \lesssim_p 1/\sqrt{nT}$.
    \end{itemize}
\end{lemma}

\begin{proof}

Below, for a generic variable $\text{x}$ indexed by $i$ and $t$, we denote $\vec{\text{x}}_{it} = \text{x}_{i,t+1} - \text{x}_{it}$.
In addition, we write $a_{1, it}(s) \coloneqq A_1(\bar Y_{it}, s)$ and $a_{2, it}(s) \coloneqq A_2(Y_{i,t-1}, s)$.

\bigskip

(i) Observe that for each $s_l$,
\begin{align}
    \left\| \bm{Z}(s_l)^\top \bm{D}^\top \bm{D} (\bm{H}(s_l) - \bbE[\bm{H}(s_l)]) (\theta_0 - \theta) \right\|
    & \le \left\| \sum_{i = 1}^n \sum_{t = 1}^{T - 1} \vec B_{it} \otimes \phi^K(s_l) [\vec a_{1,it}(s_l) - \bbE(\vec a_{1,it}(s_l))] \alpha(s_l; \theta_0 - \theta) \right\| \\
    & \quad + \left\| \sum_{i = 1}^n \sum_{t = 1}^{T - 1} \vec B_{it} \otimes \phi^K(s_l) [\vec a_{2,it}(s_l) - \bbE(\vec a_{2,it}(s_l))] \gamma(s_l; \theta_0 - \theta)  \right\|.
\end{align}
As a typical element, the variance of the first term involving $a_{1,it}$ is given as
\small\begin{align}
    & \bbE\left( \sum_{i = 1}^n \sum_{t = 1}^{T - 1} \vec Q^1_{it} \phi_1(s_l) (\vec a_{1,it} - \bbE[\vec a_{1,it}]) \alpha(s_l; \theta_0 - \theta) \right)^2 \\
    & \lesssim \sum_{t = 1}^{T - 1} \sum_{t' = 1}^{T - 1} \sum_{i = 1}^n \sum_{i' = 1}^n \left| \text{Cov}(\vec a_{1,it}, \vec a_{1,i't'}) \right| \\
    & = \sum_{t = 1}^{T - 1} \sum_{i = 1}^n \text{Var}(\vec a_{1,it}) + \sum_{t = 1}^{T - 1} \sum_{i = 1}^n \sum_{i' \neq i} \left| \text{Cov}(\vec a_{1,it}, \vec a_{1,i't}) \right| + \sum_{t = 1}^{T - 1} \sum_{t' \neq t} \sum_{i = 1}^n \left| \text{Cov}(\vec a_{1,it}, \vec a_{1,it'}) \right| + \sum_{t = 1}^{T - 1} \sum_{t' \neq t} \sum_{i = 1}^n \sum_{i' \neq i} \left| \text{Cov}(\vec a_{1,it}, \vec a_{1,i't'}) \right|.
\end{align}\normalsize
Here, the dependence on $s_l$, ``$(s_l)$'', is occasionally omitted for notational simplicity.
First, from Assumptions \ref{as:inverse}(i), \ref{as:observables}(ii), and \ref{as:A}, we can easily see that $\sum_{t = 1}^{T - 1} \sum_{i = 1}^n \text{Var}(\vec a_{1,it}) \lesssim nT$.
Second, by Lemmas \ref{lem:NED2} and \ref{lem:cov}, there exists a geometric NED coefficient $\rho$ that satisfies
\begin{align}
    \sum_{t = 1}^{T - 1} \sum_{i = 1}^n \sum_{i' \neq i} \left| \text{Cov}(\vec a_{1,it}, \vec a_{1,i't}) \right|
    & \lesssim \sum_{t = 1}^{T - 1} \sum_{i = 1}^n \sum_{i' \neq i} \rho\left(\Delta_{\text{ST}}((i,t),(i',t))/3\right) \\
    & = \sum_{t = 1}^{T - 1} \sum_{i = 1}^n \sum_{m = 1}^\infty \sum_{i' : \Delta(i,i') \in [m,m + 1)} \rho(\Delta(i,i')/3) \\
    & \lesssim \sum_{t = 1}^{T - 1} \sum_{i = 1}^n \sum_{m = 1}^\infty m^{d - 1} \rho(m) \lesssim nT,
\end{align}
where the second inequality is from Lemma A.1(iii) of \cite{jenish2009central}, and the final claim follows from the geometric NED property.
Third, Lemma \ref{lem:cov} implies that
\begin{align}
    \sum_{t = 1}^{T - 1} \sum_{t' \neq t} \sum_{i = 1}^n \left| \text{Cov}(\vec a_{1,it}, \vec a_{1,it'}) \right|
    & \lesssim \sum_{t = 1}^{T - 1} \sum_{t' \neq t} \sum_{i = 1}^n \rho\left(\Delta_{\text{ST}}((i,t),(i,t'))/3\right) \\
    & = \sum_{t = 1}^{T - 1} \sum_{t' \neq t} \sum_{i = 1}^n \rho(|t - t'|/3) \lesssim nT.
\end{align}
Finally, Lemma \ref{lem:cov} and Lemma A.1(iii) of \cite{jenish2009central} imply that
\begin{align}
    \sum_{t = 1}^{T - 1} \sum_{t' \neq t} \sum_{i = 1}^n \sum_{i' \neq i} \left| \text{Cov}(\vec a_{1,it}, \vec a_{1,i't'}) \right|
    & \lesssim \sum_{t = 1}^{T - 1} \sum_{i = 1}^n \sum_{m = 1}^\infty \sum_{\substack{i',t' : \Delta_{\text{ST}}((i,t),(i',t')) \in [m,m + 1)}} \rho\left(\Delta_{\text{ST}}((i,t),(i',t'))/3\right) \\
    & \lesssim \sum_{t = 1}^{T - 1} \sum_{i = 1}^n \sum_{m = 1}^\infty m^d \rho(m) \lesssim nT.
\end{align}
The same argument applies to the terms involving $a_{2,it}$.

Combining the above results suggests that
\begin{align}
    \bbE\left\| \bm{Z}(s_l)^\top \bm{D}^\top \bm{D} (\bm{H}(s_l) - \bbE[\bm{H}(s_l)]) (\theta_0 - \theta)/N \right\|^2 \lesssim K/(nT)
\end{align}
for each $s_l$, which completes the proof by applying Markov's and triangle inequalities.

\bigskip

(ii) Analogous to the proof of (i).

\bigskip

(iii) We prove the first part.
Let
\begin{align}
    & \underset{nTL \times (d_q + d_x)K}{\bm{Z}_L} =
    \left( \begin{array}{c}
    \bm{Z}(s_1) \\
    \vdots \\
    \bm{Z}(s_L)
    \end{array} \right), \quad
    \underset{nTL \times 1}{\bm V_L} = \left( \begin{array}{c}
    \bm V(s_1) \\
    \vdots \\
    \bm V (s_L)
    \end{array} \right), \quad
    & \underset{NL \times nTL}{\bm{D}_L} = I_L \otimes \bm D.
\end{align}
Then, noting the equivalence $\lambda_{\max}\left(\sum_{l = 1}^L \bm{Z}(s_l)^\top \bm{D}^\top \bm{D} \bm{Z}(s_l) / (NL)\right) = \lambda_{\max}\left(\bm D_L \bm Z_L \bm Z_L^\top \bm D_L^\top / (NL)\right)$, by Assumption \ref{as:weights}(ii),
\begin{align}
    \left\| \sum_{l=1}^L \bm{Z}(s_l)^\top \bm{D}^\top \bm{D} \bm{V}(s_l) / (NL) \right\|^2
    & = \left\| \bm{Z}_L^\top \bm{D}_L^\top \bm{D}_L \bm{V}_L / (NL) \right\|^2 \\
    & = \bm{V}_L^\top \bm{D}_L^\top \bm{D}_L \bm{Z}_L \bm{Z}_L^\top \bm{D}_L^\top \bm{D}_L \bm{V}_L / (N^2 L^2) \\
    & \lesssim \bm{V}_L^\top \bm{D}_L^\top \bm{D}_L \bm{V}_L / (N L)\\
    &\lesssim \frac{1}{NL} \sum_{l = 1}^L \sum_{t = 1}^T \sum_{i = 1}^n |v_{it}(s_l)|^2
\end{align}
Then, the result follows from Markov's inequality and \eqref{eq:bias_order}.

The second part can be proved analogously.

\bigskip

(iv) By the triangle inequality,
\begin{align}
    \left\| \sum_{l = 1}^L \bm{Z}(s_l)^\top \bm{D}^\top \bm{D} \bm{\mathcal{E}}(s_l) / (NL) \right\| 
    & = \left\| \sum_{l = 1}^L \sum_{t = 1}^{T-1} \sum_{i = 1}^n \vec B_{it} \otimes  \phi^K(s_l) \vec \varepsilon_{it}(s_l) / (NL) \right\| \\
    &  \le \left\| \sum_{l = 1}^L \sum_{t = 1}^{T - 1} \sum_{i = 1}^n \vec B_{it} \otimes \varepsilon_{i,t+1}(s_l) \phi^K(s_l) / (NL) \right\| \\
    & \quad + \left\| \sum_{l = 1}^L \sum_{t = 1}^{T - 1} \sum_{i = 1}^n \vec B_{it} \otimes \varepsilon_{it}(s_l) \phi^K(s_l) / (NL) \right\|.
\end{align}
Further, by Assumptions \ref{as:error}(i) and (iii),
\begin{align}
    & \bbE  \left\| \sum_{l = 1}^L \sum_{t = 1}^{T - 1} \sum_{i = 1}^n \vec B_{it} \otimes \varepsilon_{it}(s_l) \phi^K(s_l) / (NL) \right\|^2 \\
    &  = \frac{1}{N^2} \sum_{t = 1}^{T-1} \sum_{i = 1}^n  \text{trace}\left\{ \vec B_{it} \vec B_{it}^\top \otimes \frac{1}{L^2} \sum_{l = 1}^L \sum_{l' = 1}^L \Gamma_{it}(s_l, s_{l'}) \phi^K(s_l)\phi^K(s_{l'})^\top \right\} \\
    & =  \frac{1}{N^2} \sum_{t = 1}^{T -1} \sum_{i = 1}^n \text{trace}\left\{ \vec B_{it} \vec B_{it}^\top \right\}\text{trace}\left\{ \frac{1}{L^2} \sum_{l = 1}^L \sum_{l' = 1}^L \Gamma_{it}(s_l, s_{l'}) \phi^K(s_l)\phi^K(s_{l'})^\top\right\} \lesssim 1/(nT).
\end{align}
Repeating the same calculation for the other term, the result follows from Markov's inequality.

\bigskip

(v) Observe that
\begin{align}
    & \bm{E}(s_l; \theta)^\top \bm{D}^\top P_m \bm{D} \bm{E}(s_l; \theta) - \bbE[\bm{E}(s_l; \theta)^\top \bm{D}^\top P_m \bm{D} \bm{E}(s_l; \theta)] \\
    & = \sum_{t = 1}^{T - 1} \sum_{1 \le i,j \le n} p_{m,i,j} \left( \vec e_{it}(s_l; \theta) \vec e_{jt}(s_l; \theta) - \bbE[\vec e_{it}(s_l; \theta) \vec e_{jt}(s_l; \theta)] \right).
\end{align}
\normalsize
Here, let $\bm{e}_{m,jt}(s_l; \theta) \coloneqq \sum_{i = 1}^n p_{m,i,j} e_{it}(s_l; \theta)$ and recall that there is a constant $\bar\Delta_m$ such that $p_{m,i,j} = 0$ if $\Delta(i,j) > \bar\Delta_m$.
Then, noting that $\mathcal{F}_{nT,it}^{\text{ST}}((\delta - 1) \bar\Delta_m) \subseteq \mathcal{F}_{nT,jt}^{\text{ST}}(\delta \bar\Delta_m)$ for $(i,j)$ with $\Delta(i,j) \le \bar\Delta_m$ and $\delta > 1$,
\small\begin{align}
    \left\| \bm{e}_{m,jt}(s_l; \theta) - \bbE \left[ \bm{e}_{m,jt}(s_l; \theta) \mid \mathcal{F}_{nT,jt}^{\text{ST}}(\delta \bar\Delta_m) \right] \right\|_2
    & \le \sum_{i = 1}^n |p_{m,i,j}| \left\| e_{it}(s_l; \theta) - \bbE \left[ e_{it}(s_l; \theta) \mid \mathcal{F}_{nT,jt}^{\text{ST}}(\delta \bar\Delta_m) \right] \right\|_2 \\
    & \lesssim \sum_{i : \Delta(i,j) \le \bar\Delta_m} \left\| e_{it}(s_l; \theta) - \bbE \left[ e_{it}(s_l; \theta) \mid \mathcal{F}_{nT,it}^{\text{ST}}((\delta - 1) \bar\Delta_m) \right] \right\|_2 \\
    & \lesssim \varrho^{\left\lfloor (\delta - 1) \bar\Delta_m / \bar\Delta_{\text{ST}} \right\rfloor},
\end{align}\normalsize
which implies that $\{\bm{e}_{m,jt}(s_l; \theta)\}$ is uniformly and geometrically $L^2$-NED, for all $l \in [L]$, $m \in [M]$, and $\theta \in \Theta_K$.

Now, suppressing the dependence on both $s_l$ and $\theta$,
\begin{align}
    \text{Var}[\bm{E}^\top \bm{D}^\top P_m \bm{D} \bm{E}]
    & = \bbE\left( \sum_{t = 1}^{T - 1} \sum_{1 \le i,j \le n} p_{m,i,j} \{ \vec e_{it} \vec e_{jt} - \bbE[\vec e_{it} \vec e_{jt}] \} \right)^2 \\
    & = \bbE\left( \sum_{t = 1}^{T - 1} \sum_{j = 1}^n \{ \vec{\bm{e}}_{m,jt} \vec e_{jt} - \bbE[\vec{\bm{e}}_{m,jt} \vec e_{jt}] \} \right)^2 \\
    & \le \sum_{t = 1}^{T - 1} \sum_{t' = 1}^{T - 1} \sum_{j = 1}^n \sum_{j' = 1}^n \left| \text{Cov}(\vec{\bm{e}}_{m,jt} \vec e_{jt}, \vec{\bm{e}}_{m,j't'} \vec e_{j't'}) \right| \\
    & = \sum_{t = 1}^{T - 1} \sum_{j = 1}^n \text{Var}(\vec{\bm{e}}_{m,jt} \vec e_{jt}) + \sum_{t = 1}^{T - 1} \sum_{j = 1}^n \sum_{j' \neq j} \left| \text{Cov}(\vec{\bm{e}}_{m,jt} \vec e_{jt}, \vec{\bm{e}}_{m,j't} \vec e_{j't}) \right| \\
    & \quad + \sum_{t = 1}^{T - 1} \sum_{t' \neq t} \sum_{j = 1}^n \left| \text{Cov}(\vec{\bm{e}}_{m,jt} \vec e_{jt}, \vec{\bm{e}}_{m,jt'} \vec e_{jt'}) \right| + \sum_{t = 1}^{T - 1} \sum_{t' \neq t} \sum_{j = 1}^n \sum_{j' \neq j} \left| \text{Cov}(\vec{\bm{e}}_{m,jt} \vec e_{jt}, \vec{\bm{e}}_{m,j't'} \vec e_{j't'}) \right|.
\end{align}
\normalsize
As we have seen in \eqref{eq:error_bound}, we have $\left\| e_{jt} \right\|_p$, $\left\| \bm{e}_{m,jt} \right\|_p < \infty$ for $p > 4$.
This allows us to use Lemma A.1 of \cite{xu2015spatial} (see also Corollary 4.3(b) of \citealp{gallant1988unified}) to show that $\{\vec{\bm{e}}_{m,jt} \vec e_{jt}\}$ is uniformly and geometrically $L^2$-NED.
Then, following the same argument as in the proof of (i), we can show that $\text{Var}[\bm{E}(s_l; \theta)^\top \bm{D}^\top P_m \bm{D} \bm{E}(s_l; \theta)] \lesssim nT$ for each $s_l$, which gives the desired result by the triangle inequality.

\bigskip

(vi), (vii) These can be proved in a similar manner to the proof of Lemma \ref{lem:g0}.

\bigskip

(viii) For each $s_l$, $|\bm{\mathcal{E}}(s_l)^\top \bm{D}^\top P_m \bm{D} \bm{\mathcal{E}}(s_l) / N|  \lesssim_p 1/\sqrt{nT}$ holds under Assumptions \ref{as:error}(i) and (ii), as in Lemma 9 in \cite{yu2008quasi} and Lemma 1 in \cite{lee2014efficient}.
Then, the result is straightforward.

\end{proof}


\begin{lemma}\label{lem:g_theta}    
    Suppose that Assumptions \ref{as:inverse}(i), \ref{as:observables}, \ref{as:A}, \ref{as:weights}(i), and \ref{as:basis} hold.
    Then, $\sup_{\theta \in \Theta_K} \left\| \bbE[\bar g_{nT}(\theta)] \right\| \lesssim \sqrt{K}$.
\end{lemma}

\begin{proof}
    Observe that
    \begin{align}
        \left\| \bbE[\bar g_{nT}(\theta)] \right\|
    & \le \left\| \sum_{l = 1}^L \bm{Z}(s_l)^\top \bm{D}^\top \bm{D} \bbE[ \bm{E}(s_l; \theta) ] / (NL) \right\| + \sum_{m = 1}^M \left| \sum_{l = 1}^L \bbE[ \bm{E}(s_l; \theta)^\top \bm{D}^\top P_m \bm{D} \bm{E}(s_l; \theta) ] / (NL) \right|.
    \end{align}
    For the first term, 
    \begin{align}
        \left\| \sum_{l = 1}^L \bm{Z}(s_l)^\top \bm{D}^\top \bm{D} \bbE[ \bm{E}(s_l; \theta) ] / (NL) \right\| 
        & = \left\| \sum_{l = 1}^L \sum_{t = 1}^{T-1} \sum_{i = 1}^n \vec B_{it} \otimes \bbE[\vec e_{it}(s_l; \theta)] \phi^K(s_l) / (NL) \right\| \\
        & \lesssim \frac{1}{NL} \sum_{l = 1}^L \sum_{t = 1}^{T-1} \sum_{i = 1}^n \left\| \bbE[\vec e_{it}(s_l; \theta)] \phi^K(s_l) \right\| \\
        & \lesssim \sqrt{K}
    \end{align}
    uniformly in $\theta \in \Theta_K$, since $|\bbE[e_{it}(s_l; \theta)]| \le \bbE|e_{it}(s_l; \theta)| \lesssim 1$ and $\sup_{s \in [0,1]} ||\phi^K(s)|| \lesssim \sqrt{K}$.

    For the second term,
    \begin{align}
        \left| \sum_{l=1}^L \bbE[ \bm{E}(s_l; \theta)^\top \bm{D}^\top P_m \bm{D} \bm{E}(s_l; \theta) ] / (NL) \right|
        & = \left| \sum_{l=1}^L \sum_{t = 1}^{T - 1} \sum_{1 \le i,j \le n} p_{m,i,j} \bbE[\vec e_{it}(s_l; \theta) \vec e_{jt}(s_l; \theta)] / (NL) \right| \\
        & \le \frac{1}{NL} \sum_{l=1}^L \sum_{t = 1}^{T - 1} \sum_{1 \le i,j \le n} |p_{m,i,j}| \cdot |\bbE[\vec e_{it}(s_l; \theta) \vec e_{jt}(s_l; \theta)]| \\
        & \lesssim 1 
    \end{align}
    uniformly in $\theta \in \Theta_K$, since
    \begin{align}
        |\bbE[\vec e_{it}(s_l; \theta) \vec e_{jt}(s_l; \theta)]| 
        & \le \bbE[|\vec e_{it}(s_l; \theta)| \cdot |\vec e_{jt}(s_l; \theta)|] \\
        & \le ||\vec e_{it}(s_l; \theta)||_2 \cdot || \vec e_{jt}(s_l; \theta) ||_2 \\
        & \le \{||e_{i,t+1}(s_l; \theta)||_2 + ||e_{it}(s_l; \theta)||_2\} \cdot \{||e_{j,t+1}(s_l; \theta)||_2 + ||e_{jt}(s_l; \theta)||_2\} \\
        & \lesssim 1.
    \end{align}
    This completes the proof.
\end{proof}


\begin{lemma}\label{lem:consistency}
    Suppose that Assumptions \ref{as:inverse}, \ref{as:sample_space} -- \ref{as:basis} hold.
    In addition, assume that $K/\sqrt{nT} \to 0$ and $K^{1/2 - \varsigma} \to 0$ as $nT \to \infty$.
    Then, $||\hat \theta_{nT} - \theta_0|| = o_P(1)$.
\end{lemma}

\begin{proof}
    Observe that
\begin{align}
    \bar g_{nT}(\theta) - \bbE[\bar g_{nT}(\theta)]
    = \left(\begin{array}{c}
        C_0(\theta) \\
        C_1(\theta) \\
        \vdots \\
        C_M(\theta)
    \end{array}
    \right)
\end{align}
where 
\begin{align}
    C_0(\theta)
    & \coloneqq \underbracket{\sum_{l=1}^L \bm{Z}(s_l)^\top \bm{D}^\top \bm{D} (\bm{H}(s_l) - \bbE[\bm{H}(s_l)]) (\theta_0 - \theta) / (NL)}_{\lesssim_p \sqrt{K}/\sqrt{nT}: \: \text{Lemma \ref{lem:matLLN}(i)}} \\
    & \quad + \underbracket{\sum_{l = 1}^L \bm{Z}(s_l)^\top \bm{D}^\top \bm{D} \bm{V}(s_l)/ (NL) - \sum_{l = 1}^L \bm{Z}(s_l)^\top \bm{D}^\top \bm{D} \bbE[\bm{V}(s_l)]/ (NL)}_{\lesssim_p K^{-\varsigma}: \: \text{Lemma \ref{lem:matLLN}(iii)}} + \underbracket{\sum_{l=1}^L \bm{Z}(s_l)^\top \bm{D}^\top \bm{D} \bm{\mathcal{E}}(s_l) / (NL)}_{\lesssim_p 1/\sqrt{nT}: \: \text{Lemma \ref{lem:matLLN}(iv)}}
\end{align}
and, for $m = 1, \ldots, M$,
\begin{align}
    C_m(\theta)
    & \coloneqq \underbracket{\sum_{l = 1}^L \left\{ \bm{E}(s_l; \theta)^\top \bm{D}^\top P_m \bm{D} \bm{E}(s_l; \theta) - \bbE[\bm{E}(s_l; \theta)^\top \bm{D}^\top P_m \bm{D} \bm{E}(s_l; \theta)] \right\} / (NL)}_{\lesssim_p 1/\sqrt{nT}: \: \text{Lemma \ref{lem:matLLN}(v)}}.
\end{align}
    Hence,
    \begin{align}
    \left\| \bar g_{nT}(\theta) - \bbE[ \bar g_{nT}(\theta)] \right\| 
    & \le \left\| C_0(\theta) \right\| + \sum_{m = 1}^M \left| C_m(\theta) \right| \\
    & \lesssim_p \sqrt{K}/\sqrt{nT} + K^{-\varsigma}
    \end{align}
    uniformly in $\theta \in \Theta_K$.
    Further, by Cauchy-Schwarz inequality and Lemma \ref{lem:g_theta},
    \begin{align}
        \sup_{\theta \in \Theta_K} \left| \mathcal{Q}_{nT}(\theta) - \mathcal{Q}^*_{nT}(\theta) \right|
        & \le \sup_{\theta \in \Theta_K} \left| \left( \bar g_{nT}(\theta) - \bbE[\bar g_{nT}(\theta)]\right)^\top \Omega_{nT} \left( \bar g_{nT}(\theta) - \bbE[\bar g_{nT}(\theta)] \right) \right| \\
        & \quad + 2 \sup_{\theta \in \Theta_K} \left| \left( \bar g_{nT}(\theta) - \bbE[\bar g_{nT}(\theta)]\right)^\top \Omega_{nT} \bbE[\bar g_{nT}(\theta)] \right| \\
        & \lesssim \sup_{\theta \in \Theta_K} \left\| \bar g_{nT}(\theta) - \bbE[\bar g_{nT}(\theta)] \right\|^2 + \sup_{\theta \in \Theta_K} \left\| \bbE[\bar g_{nT}(\theta)] \right\| \sup_{\theta \in \Theta_K} \left\| \bar g_{nT}(\theta) - \bbE[\bar g_{nT}(\theta)] \right\| \\
        & \lesssim_p  K/\sqrt{nT} + K^{1/2 - \varsigma}.
    \end{align}
    Combined with the identifiability of $\theta_0$ (Lemma \ref{lem:identification}), the above result implies the consistency of $\hat \theta_{nT}$ (see, e.g., the proof of Theorem 3.3 in \cite{su2016sieve}).
\end{proof}


\subsubsection{Proof of Theorem \ref{thm:roc}.}

(i) Let $\zeta_{nT} \coloneqq 1/\sqrt{nT} + K^{-\varsigma}$.
Given the consistency result in Lemma \ref{lem:consistency}, if we can show that for an arbitrary $\epsilon > 0$, there exists a constant $C_\epsilon$ such that for all sufficiently large $nT$,
\begin{align}
    \Pr\left(\inf_{||\bm{u}|| = C_\epsilon}\mathcal{Q}_{nT}(\theta_0 + \zeta_{nT} \bm{u}) > \mathcal{Q}_{nT}(\theta_0)\right) \ge 1 - \epsilon,
\end{align}
we can conclude that $|| \hat \theta_{nT} - \theta_0 || \lesssim_p \zeta_{nT}$.

Decompose
\begin{align}
    \mathcal{Q}_{nT}(\theta_0 + \zeta_{nT} \bm{u}) - \mathcal{Q}_{nT}(\theta_0)
    & = \underbracket{\left( \bar g_{nT}(\theta_0 + \zeta_{nT} \bm{u}) - \bar g_{nT}(\theta_0)\right)^\top \Omega_{nT} \left( \bar g_{nT}(\theta_0 + \zeta_{nT} \bm{u}) - \bar g_{nT}(\theta_0) \right)}_{\eqqcolon \tilde A_{nT}(\bm{u})} \\
    & \quad + 2 \left( \bar g_{nT}(\theta_0 + \zeta_{nT} \bm{u}) - \bar g_{nT}(\theta_0)\right)^\top \Omega_{nT} \bar g_{nT}(\theta_0).
\end{align}
Lemma \ref{lem:matLLN}(ii) implies that $||\Pi_{nT} - \hat \Pi_{nT}|| = o_P(1)$, where $\hat \Pi_{nT} \coloneqq \sum_{l = 1}^L \bm{Z}(s_l)^\top \bm{D}^\top \bm{D} \bm{H}(s_l)/(NL)$.
Thus, by Assumption \ref{as:matrix1}, $\lambda_{\min}(\hat \Pi_{nT}^\top \hat \Pi_{nT})$ is bounded away from zero with probability approaching one.
Observing that
\begin{align}
    \bar g_{nT}(\theta) - \bar g_{nT}(\theta_0)
    & = \left(\begin{array}{c}
       \hat \Pi_{nT}  \\
        \sum_{l=1}^L \left[\bm{H}(s_l)(\theta_0 - \theta) + 2\bm{V}(s_l) + 2\bm{\mathcal{E}}(s_l)\right]^\top \bm{D}^\top P_1 \bm{D} \bm{H}(s_l)/(NL) \\
       \vdots \\
        \sum_{l=1}^L \left[\bm{H}(s_l)(\theta_0 - \theta) + 2\bm{V}(s_l) + 2\bm{\mathcal{E}}(s_l)\right]^\top \bm{D}^\top P_M \bm{D} \bm{H}(s_l)/(NL)
   \end{array}
   \right) (\theta_0 - \theta),
\end{align}
we obtain
\begin{align}
    \inf_{||\bm{u}|| = C_\epsilon}\tilde A_{nT}(\bm{u})
    \ge c_1 \zeta_{nT}^2 C_\epsilon^2
\end{align}
for some $c_1 > 0$ with probability approaching one.

For the second term, the Cauchy--Schwarz inequality gives
\begin{align}
    \left| \left( \bar g_{nT}(\theta_0 + \zeta_{nT} \bm{u}) - \bar g_{nT}(\theta_0)\right)^\top \Omega_{nT} \bar g_{nT}(\theta_0) \right|
    & \le \left(\tilde A_{nT}(\bm{u})\right)^{1/2} \left(\bar g_{nT}(\theta_0)^\top \Omega_{nT} \bar g_{nT}(\theta_0)\right)^{1/2} \\
    & \le c_2 \left(\tilde A_{nT}(\bm{u})\right)^{1/2} ||\bar g_{nT}(\theta_0)||
\end{align}
for some $c_2 > 0$.
Hence,
\begin{align}
    \mathcal{Q}_{nT}(\theta_0 + \zeta_{nT}\bm{u}) - \mathcal{Q}_{nT}(\theta_0)
    & \ge \tilde A_{nT}(\bm{u}) - 2 c_2 \left(\tilde A_{nT}(\bm{u})\right)^{1/2} ||\bar g_{nT}(\theta_0)|| \\
    & = \left(\tilde A_{nT}(\bm{u})\right)^{1/2}
    \left(\left(\tilde A_{nT}(\bm{u})\right)^{1/2} - 2 c_2 ||\bar g_{nT}(\theta_0)||\right).
\end{align}
From Lemma \ref{lem:matLLN}(iii), (iv), (vi), (vii), and (viii), we have
\begin{align}\label{eq:g0}
    \left\| \bar g_{nT}(\theta_0) \right\|
    \lesssim_p 1/\sqrt{nT} + K^{-\varsigma}
    = \zeta_{nT}.
\end{align}
Thus, choosing a sufficiently large $C_\epsilon$ completes the proof of (i).

\bigskip

(ii) Note that $\int_0^1 \phi^K(s)\phi^K(s)^\top \text{d}s = I_K$ by orthonormality.
Then, by result (i) and Assumption \ref{as:basis},
\begin{align}
    \left\|\hat \alpha_{nT} - \alpha_0 \right\|_{L^2}
    & \le \left\|\phi^K(\cdot)^\top (\hat \theta_{nT,\alpha} - \theta_{0\alpha}) \right\|_{L^2} + \left\| \phi^K(\cdot)^\top \theta_{0\alpha} - \alpha_0(\cdot) \right\|_{L^2} \\
    & \lesssim \left((\hat \theta_{nT,\alpha} - \theta_{0\alpha})^\top \left[ \int_0^1 \phi^K(s)\phi^K(s)^\top \text{d}s \right] (\hat \theta_{nT,\alpha} - \theta_{0\alpha}) \right)^{1/2} + K^{-\varsigma} \\
    & \lesssim_p 1/\sqrt{nT} + K^{-\varsigma}.
\end{align}
It is also straightforward to see that
\begin{align}
    \sup_{s \in [0,1]}|\hat \alpha_{nT}(s) - \alpha_0(s)|
    & \le \sup_{s \in [0,1]}||\phi^K(s)|| \cdot ||\hat \theta_{nT,\alpha} - \theta_{0\alpha}||
    + \sup_{s \in [0,1]}\left|\phi^K(s)^\top\theta_{0\alpha} - \alpha_0(s)\right| \\
    & \lesssim_p \sqrt{K}/\sqrt{nT} + K^{1/2 - \varsigma}.
\end{align}

\bigskip

(iii) -- (iv) Analogous to the proof of (ii).

\qed

\subsection{Proof of Theorem \ref{thm:normality}}


\subsubsection{Lemmas}

\begin{lemma}\label{lem:matLLN2}
    Suppose that Assumptions \ref{as:inverse}, \ref{as:sample_space} -- \ref{as:basis}, and \ref{as:matrix2}(i), (ii) hold.
    In addition, assume that $K/\sqrt{nT} \to 0$ and $K^{1/2 - \varsigma} \to 0$ as $nT \to \infty$.
    Let $\bar \theta_{nT}$ be any vector in between $\hat \theta_{nT}$ and $\theta_0$.
    Then,
    \begin{itemize}
        \item[(i)] $\left\| \bar J_{nT}(\hat \theta_{nT}) - \bar J_{nT} \right\| \lesssim_p K / \sqrt{nT} + K^{-\varsigma}$
        \item[(ii)] $\left\| \bar J_{nT}(\hat \theta_{nT})\bar J_{nT}(\hat \theta_{nT})^\top - \bar J_{nT}\bar J^\top_{nT} \right\| \lesssim_p K / \sqrt{nT} + K^{-\varsigma}$
        \item[(iii)] $\left\| \bar J_{nT}(\hat \theta_{nT})^\top \Omega_{nT} \bar J_{nT}(\bar \theta_{nT}) - \bar J^\top_{nT} \Omega_{nT} \bar J_{nT} \right\| \lesssim_p K / \sqrt{nT} + K^{-\varsigma}$
        \item[(iv)] $\left\|\left( \bar J_{nT}(\hat \theta_{nT})^\top \Omega_{nT} \bar J_{nT}(\bar \theta_{nT}) \right)^{-} - \left( \bar J^\top_{nT} \Omega_{nT} \bar J_{nT} \right)^{-1} \right\| \lesssim_p K / \sqrt{nT} + K^{-\varsigma}$
    \end{itemize}
\end{lemma}

\begin{proof}
    (i) Observe that
    \begin{align}
        \left\| \bar J_{nT}(\hat \theta_{nT}) - \bar J_{nT} \right\| \le B_{1, nT} + 2\sum_{m = 1}^M \left( B_{2,m,nT} + B_{3,m,nT} \right)
    \end{align}
    where
    \begin{align}
    B_{1, nT}
    & \coloneqq \left\| \sum_{l = 1}^L \bm{Z}(s_l)^\top \bm{D}^\top \bm{D} \left\{ \bm{H}(s_l) - \bbE[\bm{H}(s_l)] \right\} / (NL)   \right\| \lesssim_p \underset{\text{Lemma \ref{lem:matLLN}(ii)}}{K/\sqrt{nT}}\\
    B_{2,m,nT}
    & \coloneqq \left\| \sum_{l = 1}^L \left\{ \bm{E}(s_l; \hat \theta_{nT}) - \bm{E}(s_l; \theta_0)\right\}^\top \bm{D}^\top P_m \bm{D} \bm{H}(s_l) / (NL) \right\|   \\
    B_{3,m,nT}
    & \coloneqq \left\| \sum_{l = 1}^L \left\{ \bm{E}(s_l; \theta_0)^\top \bm{D}^\top P_m \bm{D} \bm{H}(s_l) - \bbE[\bm{E}(s_l; \theta_0)^\top \bm{D}^\top P_m \bm{D} \bm{H}(s_l)]\right\} / (NL) \right\|.
    \end{align}
    In a similar manner to the proof of Lemma \ref{lem:matLLN}(v), we can show that $||(NL)^{-1} \sum_{l = 1}^L \{ \bm{H}(s_l)^\top \bm{H}(s_l) - \bbE[\bm{H}(s_l)^\top \bm{H}(s_l)]\}|| \lesssim_p K/\sqrt{nT}$.
    Then, by Assumption \ref{as:matrix2}(i) and Theorem \ref{thm:roc}(i), we have
    \begin{align}
        B_{2,m,nT}
        & = \left\| \left( \frac{1}{NL} \sum_{l = 1}^L \bm{H}(s_l)^\top \bm{D}^\top P_m \bm{D} \bm{H}(s_l) \right) \left( \hat \theta_{nT} - \theta_0 \right) \right\| \\
        & \le \lambda_{\max} \left( \frac{1}{NL} \sum_{l = 1}^L \bm{H}(s_l)^\top \bm{D}^\top P_m \bm{D} \bm{H}(s_l) \right) \cdot \left\| \hat \theta_{nT} - \theta_0 \right\| \\
        & \le \lambda_{\max} \left( \bm{D}^\top P_m \bm{D}\right) \cdot \lambda_{\max} \left( \frac{1}{NL} \sum_{l = 1}^L \bm{H}(s_l)^\top \bm{H}(s_l) \right) \cdot \left\| \hat \theta_{nT} - \theta_0 \right\| \lesssim_p 1/\sqrt{nT} + K^{-\varsigma}.
    \end{align}
    For $B_{3,m,nT}$, by the same argument as in Lemma \ref{lem:matLLN}(v), we can show that $B_{3,m,nT} \lesssim_p \sqrt{K}/\sqrt{nT}$.
    This completes the proof.
    
    \bigskip

    (ii) By the triangle inequality, 
    \begin{align}
        \left\| \bar J_{nT}(\hat \theta_{nT})\bar J_{nT}(\hat \theta_{nT})^\top - \bar J_{nT}\bar J^\top_{nT} \right\| 
        & \le \left\| (\bar J_{nT}(\hat \theta_{nT}) - \bar J_{nT}) (\bar J_{nT}(\hat \theta_{nT}) - \bar J_{nT})^\top \right\| \\
        & \quad + 2 \left\| \bar J_{nT} (\bar J_{nT}(\hat \theta_{nT}) - \bar J_{nT})^\top \right\| \\
        & \lesssim_p K / \sqrt{nT} + K^{-\varsigma}
    \end{align}
    where the last inequality is from result (i) and Assumption \ref{as:matrix2}(ii). 

    \bigskip

    (iii) By definition of $\bar \theta_{nT}$, we have $\left\|\bar \theta_{nT} - \theta_0\right\| \lesssim_p 1/\sqrt{nT} + K^{-\varsigma}$ and thus $\left\| \bar J_{nT}(\bar \theta_{nT}) - \bar J_{nT} \right\| \lesssim_p K / \sqrt{nT} + K^{-\varsigma}$, as in result (i).
    Then, by the triangle inequality,
    \begin{align}
        \left\| \bar J_{nT}(\hat \theta_{nT})^\top \Omega_{nT} \bar J_{nT}(\bar \theta_{nT}) - \bar J^\top_{nT} \Omega_{nT} \bar J_{nT} \right\| 
        & \le \left\| (\bar J_{nT}(\hat \theta_{nT}) - \bar J_{nT})^\top \Omega_{nT} (\bar J_{nT}(\bar \theta_{nT}) - \bar J_{nT}) \right\| \\
        & \quad + \left\| (\bar J_{nT}(\hat \theta_{nT}) - \bar J_{nT})^\top \Omega_{nT} \bar J_{nT} \right\| \\
        & \quad + \left\| \bar J_{nT}^\top \Omega_{nT} (\bar J_{nT}(\bar \theta_{nT}) - \bar J_{nT}) \right\| \\
        & \lesssim_p K / \sqrt{nT} + K^{-\varsigma}.
    \end{align}

    \bigskip

    (iv) As a result of (iii), $\bar J_{nT}(\hat \theta_{nT})^\top \Omega_{nT} \bar J_{nT}(\bar \theta_{nT})$ is nonsingular with probability approaching one by Assumptions \ref{as:weights}(ii) and \ref{as:matrix2}(ii).
    Then, on this event, noting the equality
    \begin{align}
         & \left( \bar J_{nT}(\hat \theta_{nT})^\top \Omega_{nT} \bar J_{nT}(\bar \theta_{nT}) \right)^{-} - \left( \bar J^\top_{nT} \Omega_{nT} \bar J_{nT} \right)^{-1} \\
         & = \left( \bar J_{nT}(\hat \theta_{nT})^\top \Omega_{nT} \bar J_{nT}(\bar \theta_{nT}) \right)^{-} \left[ \bar J^\top_{nT} \Omega_{nT} \bar J_{nT} - \bar J_{nT}(\hat \theta_{nT})^\top \Omega_{nT} \bar J_{nT}(\bar \theta_{nT}) \right] \left( \bar J^\top_{nT} \Omega_{nT} \bar J_{nT} \right)^{-1},
    \end{align}
    the result is straightforward.
\end{proof}

\subsubsection{Proof of Theorem \ref{thm:normality}.}

Since the proof is similar, we only prove (i).
By the first-order condition of minimization and the mean-value expansion, we have
\begin{align}
    \bm{0}_{(d_x + 2)K} 
    & = \bar J_{nT}(\hat \theta_{nT})^\top \Omega_{nT} \bar g_{nT}(\hat \theta_{nT}) \\
    & = \bar J_{nT}(\hat \theta_{nT})^\top \Omega_{nT} \left[ \bar g_{nT}(\theta_0) + \bar J_{nT}(\bar \theta_{nT}) \left( \hat \theta_{nT} - \theta_0\right) \right],
\end{align}
where $\bar \theta_{nT} \in [\hat \theta_{nT}, \theta_0]$, leading to
\begin{align}
    \hat \theta_{nT} - \theta_0
    & = - \left( \bar J_{nT}(\hat \theta_{nT})^\top \Omega_{nT} \bar J_{nT}(\bar \theta_{nT}) \right)^{-} \bar J_{nT}(\hat \theta_{nT})^\top \Omega_{nT} \bar g_{nT}(\theta_0) \\
    & =  - [G_{1, nT} + G_{2, nT} + G_{3, nT} + G_{4, nT}],
\end{align}
with
\begin{align}
    G_{1, nT}
    & \coloneqq \left( \bar J^\top_{nT} \Omega_{nT} \bar J_{nT} \right)^{-1} \bar J^\top_{nT} \Omega_{nT} \bar g_{1,nT} \\
    G_{2, nT}
    & \coloneqq \left( \bar J^\top_{nT} \Omega_{nT} \bar J_{nT} \right)^{-1} \bar J^\top_{nT} \Omega_{nT} \bar g_{2,nT} \\
    G_{3, nT}
    & \coloneqq \left( \bar J^\top_{nT} \Omega_{nT} \bar J_{nT} \right)^{-1} \left\{ \bar J_{nT}(\hat \theta_{nT}) - \bar J_{nT} \right\}^\top \Omega_{nT}  \bar g_{nT}(\theta_0) \\
    G_{4, nT}
    & \coloneqq \left\{ \left( \bar J_{nT}(\hat \theta_{nT})^\top \Omega_{nT} \bar J_{nT}(\bar \theta_{nT}) \right)^{-} - \left( \bar J^\top_{nT} \Omega_{nT} \bar J_{nT} \right)^{-1} \right\} \bar J_{nT}(\hat \theta_{nT})^\top \Omega_{nT}  \bar g_{nT}(\theta_0).
\end{align}
First, observing that
\begin{align}
    \left\| G_{2,nT} \right\|^2 
    & = \bar g_{2,nT}^\top \Omega_{nT} \bar J_{nT} \left( \bar J^\top_{nT} \Omega_{nT} \bar J_{nT} \right)^{-2} \bar J^\top_{nT} \Omega_{nT} \bar g_{2,nT} \\
    & \lesssim \bar g_{2,nT}^\top \Omega_{nT} \bar J_{nT} \left( \bar J^\top_{nT} \Omega_{nT} \bar J_{nT} \right)^{-1} \bar J^\top_{nT} \Omega_{nT} \bar g_{2,nT} \\
    & \lesssim \bar g_{2,nT}^\top \Omega_{nT} \bar g_{2,nT} \\
    & \lesssim \left\| \bar g_{2,nT} \right\|^2,
\end{align}
we can find that $\left\| G_{2,nT} \right\| \lesssim_p K^{-\varsigma}$ by Lemma \ref{lem:matLLN}(iii), (vi), and (vii).
Next, it is easy to see that
\begin{align}
    \left\| G_{3,nT} \right\| 
    & \lesssim \underbracket{\left\|  \left( \bar J^\top_{nT} \Omega_{nT} \bar J_{nT} \right)^{-1} \left\{ \bar J_{nT}(\hat \theta_{nT}) - \bar J_{nT} \right\}^\top\right\|}_{\lesssim_p K/\sqrt{nT} + K^{-\varsigma}: \: \text{Lemma \ref{lem:matLLN2}(i)}}  \cdot \underbracket{\left\| \bar g_{nT}(\theta_0) \right\|}_{\lesssim_p 1/\sqrt{nT} + K^{-\varsigma}: \: \eqref{eq:g0}} \\
    & \lesssim_p K /(nT)
\end{align}
from $\sqrt{nT} K^{-\varsigma} \to 0$.
Further, for $G_{4,nT}$, by Lemma \ref{lem:matLLN2}(ii) and (iv) and \eqref{eq:g0} with Assumption \ref{as:matrix2}(ii), we can show that $\left\| G_{4,nT} \right\| \lesssim_p K /(nT)$.

Combining all these results and noting that 
\begin{align}
[\sigma_{nT, \alpha}(s)]^2 \ge c \phi^K(s)^\top \underbracket{\mathbb{S}_\alpha \mathbb{S}_\alpha^\top}_{= I_K} \phi^K(s) \ge c ||\phi^K(s) ||^2 > 0
\end{align}
for sufficiently large $nT$, we have
\begin{align}
    \frac{\sqrt{N} \left( \hat \alpha_{nT}(s) - \alpha_0(s) \right)}{\sigma_{nT, \alpha}(s)} 
    & = -\frac{\sqrt{N} \phi^K(s)^\top \mathbb{S}_\alpha[ G_{1, nT} + O_P(K/(nT)) + O_P(K^{-\varsigma})]}{\sigma_{nT, \alpha}(s)} + \frac{\sqrt{N} O(K^{- \varsigma})}{\sigma_{nT, \alpha}(s)} \\
    & = -\frac{\sqrt{N} \phi^K(s)^\top \mathbb{S}_\alpha G_{1, nT}}{\sigma_{nT, \alpha}(s)} + o_P(1)
\end{align}
Here, let $\Lambda_{z, nT}(s)$ and $\Lambda_{m, nT}(s)$ denote the first $(d_q + d_x)K$ elements and $((d_q + d_x)K + m)$-th element of $\phi^K(s)^\top \mathbb{S}_\alpha \left( \bar J^\top_{nT} \Omega_{nT} \bar J_{nT} \right)^{-1} \bar J^\top_{nT} \Omega_{nT}/\sigma_{nT, \alpha}(s)$, respectively.
Then, we can write
\begin{align}
    & \frac{\sqrt{N} \phi^K(s)^\top \mathbb{S}_\alpha G_{1, nT}}{\sigma_{nT, \alpha}(s)} \\
    & = \frac{ \phi^K(s)^\top \mathbb{S}_\alpha \left( \bar J^\top_{nT} \Omega_{nT} \bar J_{nT} \right)^{-1} \bar J^\top_{nT} \Omega_{nT} \left[\sqrt{N} \bar g_{1,nT}\right]}{\sigma_{nT, \alpha}(s)} \\
    & = \frac{1}{L \sqrt{N}} \sum_{l = 1}^L \left[ \Lambda_{z, nT}(s)  \bm{Z}(s_l)^\top \bm{D}^\top \bm{D} \bm{\mathcal{E}}(s_l) +  \bm{\mathcal{E}}(s_l)^\top \bm{D}^\top \left( \sum_{m = 1}^M \Lambda_{m, nT}(s) P_m \right) \bm{D} \bm{\mathcal{E}}(s_l) \right].
\end{align}

\bigskip

Hereinafter in this proof, we use the labelling ($\text{i} = 1, \ldots, \bm n$) introduced in Appendix \ref{app:prep}.
Let us denote $\underset{\bm{n} \times \bm{n}}{\Pi_M(s)} \coloneqq \bm{D}^\top \left( \sum_{m = 1}^M \Lambda_{m, nT}(s) P_m  \right) \bm{D} = (\pi_{M,\text{i},\text{j}}(s))$.
Recalling the block-diagonal structure of $P_m$ and that its diagonals are all zero, we can find that the diagonal elements of $\Pi_M(s)$ are also all zero.
Further, note that $\Pi_M(s)$ is symmetric.
Now, recalling that $z_\text{i}^\dagger(s_l)$ is the $\text{i}$-th column of $\bm{Z}(s_l)^\top \bm{D}^\top \bm{D}$, define
\begin{align}
    a_\text{i}(s)
    & \coloneqq \frac{1}{L\sqrt{N}}\sum_{l=1}^L \Lambda_{z, nT}(s) z_\text{i}^\dagger(s_l) \varepsilon_\text{i}(s_l) \\
    b_{\text{i}, \text{j}}(s)
    & \coloneqq \frac{1}{L\sqrt{N}}\sum_{l=1}^L \pi_{M,\text{i},\text{j}}(s) \varepsilon_\text{i}(s_l) \varepsilon_\text{j}(s_l) \\
    \gamma_{\text{i}}(s)
    & \coloneqq a_\text{i}(s) + 2 \underbracket{\sum_{\text{j} = 1}^{\text{i} - 1} b_{\text{i}, \text{j}}(s)}_{= 0 \text{ if } \text{i} = 1},
\end{align}
and we further re-write
\begin{align}
    \frac{\sqrt{N} \phi^K(s)^\top \mathbb{S}_\alpha G_{1, nT}}{\sigma_{nT, \alpha}(s)} = \sum_{\text{i} = 1}^{\bm{n}} \gamma_{\text{i}}(s).
\end{align}
Here, let $\mathscr{F}_{nT}(\text{i})$ denote the $\sigma$-field generated by $\{\varepsilon_\text{j}: 1 \le \text{j} \le \text{i}\}$.
Under Assumptions \ref{as:error}(i) and (ii), we have $\bbE[\gamma_{\text{i}}(s) \mid \mathscr{F}_{nT}(\text{i} - 1)] = 0$, implying that $\{\gamma_{\text{i}}(s)\}$ forms a martingale difference sequence.

Then, it suffices to check the following two conditions for the central limit theorem of \cite{scott1973central}:
\begin{align}
    (1) \;\; & \sum_{\text{i} = 1}^{\bm{n}} \bbE[(\gamma_{\text{i}}(s))^2 \mid \mathscr{F}_{nT}(\text{i} - 1)] \overset{p}{\to} 1 \\
    (2) \;\; & \sum_{\text{i} = 1}^{\bm{n}} \bbE[(\gamma_{\text{i}}(s))^2 \bm{1}\{|\gamma_{\text{i}}(s)| \ge \eta\}\mid \mathscr{F}_{nT}(\text{i} - 1)] \overset{p}{\to} 0 \;\; \text{for any $\eta > 0$}
\end{align}

\paragraph{Verification of condition (1)}

Observe that
\begin{align}
    \bbE[(\gamma_{\text{i}}(s))^2 \mid \mathscr{F}_{nT}(\text{i} - 1)] 
    & = \bbE[(a_\text{i}(s))^2] + 4 \sum_{\text{j}_1 = 1}^{\text{i} - 1} \sum_{\text{j}_2 = 1}^{\text{i} - 1} \bbE[ b_{\text{i}, \text{j}_1}(s) b_{\text{i}, \text{j}_2}(s) \mid \mathscr{F}_{nT}(\text{i} - 1)] \\
    & \quad + 4 \sum_{\text{j} = 1}^{\text{i} - 1} \bbE[ a_\text{i}(s) b_{\text{i}, \text{j}}(s) \mid \mathscr{F}_{nT}(\text{i} - 1)].
\end{align}
Recalling the definition of $\mathcal{V}_{z,nT}$ in \eqref{eq:Vz}, we can easily see that 
\begin{align}
	\sum_{\text{i} = 1}^{\bm{n}} \bbE[(a_\text{i}(s))^2] = \Lambda_{z, nT}(s) \mathcal{V}_{z,nT} \Lambda_{z, nT}(s)^\top.
\end{align}

For the second term on the right-hand side, noting that $\text{j}_1, \text{j}_2 \le \text{i} - 1$,
\begin{align}
    & 4 \sum_{\text{i} = 1}^{\bm{n}} \sum_{\text{j}_1 = 1}^{\text{i} - 1} \sum_{\text{j}_2 = 1}^{\text{i} - 1} \bbE[ b_{\text{i}, \text{j}_1}(s) b_{\text{i}, \text{j}_2}(s) \mid \mathscr{F}_{nT}(\text{i} - 1)] \\
    & =  \frac{1}{L^2}\sum_{l=1}^L \sum_{l'=1}^L \underbracket{\frac{4}{N} \sum_{\text{i} = 1}^{\bm{n}} \sum_{\text{j}_1 = 1}^{\text{i} - 1} \sum_{\text{j}_2 = 1}^{\text{i} - 1} \pi_{M,\text{i},\text{j}_1}(s) \pi_{M,\text{i},\text{j}_2}(s) \Gamma_\text{i}(s_l, s_{l'}) \varepsilon_{\text{j}_1}(s_l)  \varepsilon_{\text{j}_2}(s_{l'})}_{\eqqcolon D(s,s_l, s_{l'})}.
\end{align}
Since $\Pi_M(s)$ is symmetric and its diagonals are zero, recalling the definition of $\mathcal{V}_{ab,nT}$ in \eqref{eq:Vab}, direct calculation yields
\begin{align}
    & \frac{1}{L^2}\sum_{l=1}^L \sum_{l'=1}^L \bbE[D(s,s_l, s_{l'})] \\
    & = \frac{4}{N} \sum_{\text{i} = 1}^{\bm{n}} \sum_{\text{j} = 1}^{\text{i} - 1} [\pi_{M,\text{i},\text{j}}(s)]^2 \frac{1}{L^2}\sum_{l=1}^L \sum_{l'=1}^L \Gamma_\text{i}(s_l, s_{l'}) \Gamma_\text{j}(s_l, s_{l'}) \\
    & = \frac{2}{N} \sum_{1 \le \text{i}, \text{j} \le \bm{n}} [\pi_{M,\text{i},\text{j}}(s)]^2 \frac{1}{L^2}\sum_{l=1}^L \sum_{l'=1}^L \Gamma_\text{i}(s_l, s_{l'}) \Gamma_\text{j}(s_l, s_{l'}) \\
    & = \frac{2}{N} \sum_{1 \le \text{i}, \text{j} \le \bm{n}} \left[\sum_{m=1}^M \Lambda_{m, nT}(s) \tilde p_{m, \text{i},\text{j}}\right]^2 \frac{1}{L^2}\sum_{l=1}^L \sum_{l'=1}^L \Gamma_\text{i}(s_l, s_{l'}) \Gamma_\text{j}(s_l, s_{l'}) \\
    & =  \sum_{a=1}^M \sum_{b=1}^M \Lambda_{a, nT}(s) \Lambda_{b, nT}(s) \frac{2}{N} \sum_{1 \le \text{i}, \text{j} \le \bm{n}} \tilde p_{a, \text{i},\text{j}} \tilde p_{b, \text{i},\text{j}} \frac{1}{L^2}\sum_{l=1}^L \sum_{l'=1}^L \Gamma_\text{i}(s_l, s_{l'}) \Gamma_\text{j}(s_l, s_{l'}) \\
    & =  \sum_{a=1}^M \sum_{b=1}^M \Lambda_{a, nT}(s) \Lambda_{b, nT}(s) \mathcal{V}_{ab,nT}.
\end{align}
Meanwhile,
\begin{align}
    \text{Var}\left( D(s, s_l, s_{l'}) \right) 
    & \lesssim \frac{1}{N^2} \sum_{\text{i} = 1}^{\bm{n}} \sum_{\text{j}_1 = 1}^{\text{i} - 1} \sum_{\text{j}_2 = 1}^{\text{i} - 1} \sum_{\text{i}' = 1}^{\bm{n}} \sum_{\text{k}_1 = 1}^{\text{i}' - 1} \sum_{\text{k}_2 = 1}^{\text{i}' - 1} \left| \pi_{M,\text{i},\text{j}_1}(s) \pi_{M,\text{i},\text{j}_2}(s) \pi_{M,\text{i}',\text{k}_1}(s) \pi_{M,\text{i}',\text{k}_2}(s) \right| \\ 
    & \quad \times \left| \bbE\{ (\varepsilon_{\text{j}_1}(s_l)  \varepsilon_{\text{j}_2}(s_{l'}) - \bbE[\varepsilon_{\text{j}_1}(s_l)  \varepsilon_{\text{j}_2}(s_{l'})]) (\varepsilon_{\text{k}_1}(s_l)  \varepsilon_{\text{k}_2}(s_{l'}) - \bbE[\varepsilon_{\text{k}_1}(s_l)  \varepsilon_{\text{k}_2}(s_{l'})])\} \right| \\
    & \lesssim  \frac{1}{N^2} \sum_{\text{i} = 1}^{\bm{n}} \sum_{\text{i}' = 1}^{\bm{n}} \sum_{\text{j}_1 = 1}^{\text{i} - 1} \sum_{\text{j}_2 = 1}^{\text{i} - 1}  \left| \pi_{M,\text{i},\text{j}_1}(s) \pi_{M,\text{i},\text{j}_2}(s) \pi_{M,\text{i}',\text{j}_1}(s) \pi_{M,\text{i}',\text{j}_2}(s) \right| \\
    & \le \frac{1}{N^2} \sum_{\text{i}' = 1}^{\bm{n}} || \Pi_M(s) ||_1^2 \cdot || \Pi_M(s) ||_\infty^2 \lesssim 1/(nT). 
\end{align}
Consequently, $4 \sum_{\text{i} = 1}^{\bm{n}} \sum_{\text{j}_1 = 1}^{\text{i} - 1} \sum_{\text{j}_2 = 1}^{\text{i} - 1} \bbE[ b_{\text{i}, \text{j}_1}(s) b_{\text{i}, \text{j}_2}(s) \mid \mathscr{F}_{nT}(\text{i} - 1)] \overset{p}{\to} \sum_{a=1}^M \sum_{b=1}^M \Lambda_{a, nT}(s) \Lambda_{b, nT}(s) \mathcal{V}_{ab,nT}$ holds from Chebyshev's inequality.

For the third term, for $\text{j} \le \text{i} - 1$,
\begin{align}
    \bbE[ a_\text{i}(s) b_{\text{i}, \text{j}}(s) \mid \mathscr{F}_{nT}(\text{i} - 1)]
    & = \frac{1}{L^2 N}\sum_{l=1}^L \sum_{l'=1}^L  \Lambda_{z, nT}(s) z_\text{i}^\dagger(s_l) \pi_{M,\text{i},\text{j}}(s) \bbE[ \varepsilon_\text{i}(s_l) \varepsilon_\text{i}(s_{l'}) ] \varepsilon_\text{j}(s_{l'}) \\
    & = \frac{1}{L N}\sum_{l=1}^L \pi_{M,\text{i},\text{j}}(s) h_\text{i}(s, s_l) \varepsilon_\text{j}(s_l),
\end{align}
where $h_\text{i}(s, s_{l'}) \coloneqq L^{-1} \sum_{l = 1}^L \Lambda_{z, nT}(s) z_\text{i}^\dagger(s_l) \bbE[ \varepsilon_\text{i}(s_l) \varepsilon_\text{i}(s_{l'}) ]$.
Hence, we can write
\begin{align}
    4 \sum_{\text{i} = 1}^{\bm{n}} \sum_{\text{j} = 1}^{\text{i} - 1} \bbE[ a_\text{i}(s) b_{\text{i}, \text{j}}(s) \mid \mathscr{F}_{nT}(\text{i} - 1)] 
    & = \frac{4}{L}\sum_{l=1}^L \left( \frac{1}{N}  \sum_{\text{j} = 1}^{\bm{n}} \sum_{\text{i} = \text{j} + 1}^{\bm{n}} \pi_{M,\text{i},\text{j}}(s) h_\text{i}(s, s_l) \varepsilon_\text{j}(s_l) \right).
\end{align}
Noting that $|h_\text{i}(s, s_{l'})| \lesssim K/|| \phi^K(s)||$,
\begin{align}
    \bbE \left| \frac{1}{N} \sum_{\text{j} = 1}^{\bm{n}} \sum_{\text{i} = \text{j} + 1}^{\bm{n}} \pi_{M,\text{i},\text{j}}(s) h_\text{i}(s,s_l) \varepsilon_\text{j}(s_l) \right|^2
    & \lesssim \frac{K^2}{N^2 \left\| \phi^K(s) \right\|^2} \sum_{\text{j} = 1}^{\bm{n}} \sum_{\text{i} = \text{j} + 1}^{\bm{n}} \sum_{\text{i}' = \text{j} + 1}^{\bm{n}} |\pi_{M,\text{i},\text{j}}(s)| \cdot |\pi_{M,\text{i}',\text{j}}(s)| \\
    & \le \frac{K^2}{N^2 \left\| \phi^K(s) \right\|^2} \sum_{\text{j} = 1}^{\bm{n}} \left\| \Pi_M(s) \right\|_1^2 \lesssim \frac{K^2}{nT \left\| \phi^K(s) \right\|^2}.
\end{align}
Then, by Markov's inequality, we obtain $4 \sum_{\text{i} = 1}^{\bm{n}} \sum_{\text{j} = 1}^{\text{i} - 1} \bbE[ a_\text{i}(s) b_{\text{i}, \text{j}}(s) \mid \mathscr{F}_{nT}(\text{i} - 1)] \overset{p}{\to} 0$.

Finally, combining the above results gives
\small\begin{align}
    & \sum_{\text{i} = 1}^{\bm{n}} \bbE[(\gamma_{\text{i}}(s))^2 \mid \mathscr{F}_{nT}(\text{i} - 1)] \overset{p}{\to} \Lambda_{z, nT}(s) \mathcal{V}_{z,nT} \Lambda_{z, nT}(s)^\top + \sum_{a=1}^M \sum_{b=1}^M \Lambda_{a, nT}(s) \Lambda_{b, nT}(s) \mathcal{V}_{ab,nT} \\
    & = \left(\begin{array}{cccc}
     \Lambda_{z, nT}(s) & \Lambda_{1, nT}(s) & \cdots & \Lambda_{M, nT}(s) 
    \end{array}\right)  \\
    & \quad \times
    \left( \begin{array}{cccc}
    \mathcal{V}_{z, nT} & \bm{0}_{(d_q + d_x)K \times 1} & \cdots & \bm{0}_{(d_q + d_x)K \times 1} \\
     \bm{0}_{1 \times (d_q + d_x)K} &     \mathcal{V}_{11, nT} & \cdots & \mathcal{V}_{1M, nT} \\
     \vdots &  \vdots & \ddots & \vdots \\
    \bm{0}_{1 \times (d_q + d_x)K} &     \mathcal{V}_{M1, nT} & \cdots & \mathcal{V}_{MM, nT}
    \end{array} \right) \left(\begin{array}{c}
     \Lambda_{z, nT}(s)^\top \\
     \Lambda_{1, nT}(s) \\
     \vdots \\
     \Lambda_{M, nT}(s) 
    \end{array}\right) \\
    & =  \frac{\phi^K(s)^\top \mathbb{S}_\alpha \Sigma_{nT} \mathbb{S}_\alpha^\top \phi^K(s)}{[\sigma_{nT, \alpha}(s)]^2} = 1,
\end{align}\normalsize
as desired.

\paragraph{Verification of condition (2)}
To verify condition (2), it is sufficient to show that $\sum_{\text{i} = 1}^{\bm{n}} \bbE[|\gamma_{\text{i}}(s)|^4 \mid \mathscr{F}_{nT}(\text{i} - 1)] \overset{p}{\to} 0$.
Moreover, by the $c_r$ inequality,
\begin{align}
    \sum_{\text{i} = 1}^{\bm{n}} \bbE[ |\gamma_{\text{i}}(s)|^4 \mid \mathscr{F}_{nT}(\text{i} - 1)]
    \le 8 \sum_{\text{i} = 1}^{\bm{n}} \bbE[ |a_\text{i}(s)|^4 ] + 128 \sum_{\text{i} = 1}^{\bm{n}} \bbE\left[ \left| \sum_{\text{j} = 1}^{\text{i} - 1} b_{\text{i}, \text{j}}(s) \right|^4 \mid \mathscr{F}_{nT}(\text{i} - 1) \right].
\end{align}
For the first term on the right-hand side, noting that Assumption \ref{as:error}(ii) implies $\bbE|\prod_{k = 1}^4 \varepsilon_{it}(s_k) | < \infty$ by H\"{o}lder's inequality,
\begin{align}
    \sum_{\text{i} = 1}^{\bm{n}} \bbE[ |a_\text{i}(s)|^4 ] 
    & \le \frac{1}{L^4 N^2} \sum_{\text{i} = 1}^{\bm{n}} \sum_{1 \le l_1, l_2, l_3, l_4 \le L}  \left| \prod_{j = 1}^4 \Lambda_{z, nT}(s) z_\text{i}^\dagger(s_{l_j}) \right| \cdot 
 \bbE\left| \prod_{j = 1}^4 \varepsilon_\text{i}(s_{l_j}) \right| \\
    & \lesssim \frac{1}{N^2} \sum_{\text{i} = 1}^{\bm{n}} \frac{1}{L^4} \sum_{1 \le l_1, l_2, l_3, l_4 \le L}  \left| \prod_{j = 1}^4 \Lambda_{z, nT}(s) z_\text{i}^\dagger(s_{l_j}) \right| \lesssim \frac{K^4}{nT \left\| \phi^K(s) \right\|^4}.
\end{align}

For the second term, observe that
\small\begin{align}
    & \sum_{\text{i} = 1}^{\bm{n}} \bbE\left[ \left| \sum_{\text{j} = 1}^{\text{i} - 1} b_{\text{i}, \text{j}}(s) \right|^4 \mid \mathscr{F}_{nT}(\text{i} - 1) \right] \\
    &  = \frac{1}{L^4 N^2} \sum_{\text{i} = 1}^{\bm{n}} \sum_{1 \le l_1, l_2, l_3, l_4 \le L} \left( \sum_{\text{j}_1 = 1}^{\text{i} - 1}  \pi_{M,\text{i},\text{j}_1}(s) \varepsilon_{\text{j}_1}(s_{l_1}) \right)  \cdots  \left( \sum_{\text{j}_4 = 1}^{\text{i} - 1}  \pi_{M,\text{i},\text{j}_4}(s) \varepsilon_{\text{j}_4}(s_{l_4}) \right) \bbE\left( \prod_{k = 1}^4 \varepsilon_\text{i}(s_{l_k})  \right).
\end{align}\normalsize
Further, it is easy to see that $\sum_{\text{j} = 1}^{\text{i} - 1}  \pi_{M,\text{i},\text{j}}(s) \varepsilon_{\text{j}}(s_l) \lesssim_p 1$ by Markov's inequality. 
Hence, we have $\sum_{\text{i} = 1}^{\bm{n}} \bbE\left[ \left| \sum_{\text{j} = 1}^{\text{i} - 1} b_{\text{i}, \text{j}}(s) \right|^4 \mid \mathscr{F}_{nT}(\text{i} - 1) \right] \lesssim_p 1/(nT)$, and combining this with the previous result implies condition (2).

\qed

\subsection{Proof of Corollary \ref{cor:normality}.}

In view of Proposition \ref{prop:var}(iii) -- (v), the results will follow from Slutsky's theorem.

\qed


\bigskip

\subsection{Proof of Proposition \ref{prop:impulse}}

Since the proofs of (i) and (ii) are almost identical, we only prove (i).
By the triangle inequality,
\begin{align}
    \left\| \hat M^S_{nT}(i,j,s) - M(i,j,s) \right\|_\infty
    & \le \left\| \hat M^S_{nT}(i,j,s) - M^S(i,j,s) \right\|_\infty + \left\| M^S(i,j,s) - M(i,j,s) \right\|_\infty,
\end{align}
where $M^S(i,j,s) \coloneqq \sum_{\ell = 0}^S W_n^\ell \bm{e}_i \tau^\ell(\beta_{0j},s)$.
For the first term on the right-hand side, observe that
\begin{align}
    \left| \{\hat M^S_{nT}(i,j,s)\}_k - \{M^S(i,j,s)\}_k \right|
    & \le \sum_{\ell = 0}^S \left| \{W_n^\ell \bm{e}_i\}_k \right| \cdot \left| \hat\tau_{nT}^\ell(\hat\beta_{nT,j},s) - \tau^\ell(\beta_{0j},s) \right| \\
    & \le \sum_{\ell = 0}^S \left| \hat\tau_{nT}^\ell(\hat\beta_{nT,j},s) - \tau^\ell(\beta_{0j},s) \right|
\end{align}
for all $k \in [n]$.
By definition, when $\ell = 0$, 
\begin{align}
    \left| \hat \tau_{nT}^0(\hat \beta_{nT,j}, s) - \tau^0(\beta_{0j}, s) \right|
    & = \left| \hat \beta_{nT,j}(s) - \beta_{0j}(s) \right| \lesssim_p c_n
\end{align}
uniformly in $s \in [0,1]$, where $c_n \coloneqq \sqrt{K}/\sqrt{nT} + K^{1/2 - \varsigma}$.
When $\ell = 1$, by Assumption \ref{as:A},
\begin{align}
    \left| \hat \tau_{nT}^1(\hat \beta_{nT,j}, s) - \tau^1(\beta_{0j}, s) \right|
    & = \left| \hat \alpha_{nT}(s) A_1(\hat \beta_{nT,j}, s) - \alpha_0(s) A_1(\beta_{0j}, s) \right| \\
    & \le \left| \hat \alpha_{nT}(s) \right| \cdot \left| A_1(\hat \beta_{nT,j} - \beta_{0j}, s) \right| + \left| \hat \alpha_{nT}(s) - \alpha_0(s) \right| \cdot \left| A_1(\beta_{0j}, s) \right| \\
    & \lesssim_p (\bar \alpha_0 + c_n) \cdot \sup_{s \in [0,1]} \left| \hat \beta_{nT,j}(s) - \beta_{0j}(s) \right| + c_n \\
    & \lesssim_p \bar \alpha_0 c_n + c_n
\end{align}
uniformly in $s \in [0,1]$.
Similarly, when $\ell = 2$, we have
\begin{align}
    \left| \hat \tau_{nT}^2(\hat \beta_{nT,j}, s) - \tau^2(\beta_{0j}, s) \right|
    & = \left| \hat \alpha_{nT}(s) A_1(\hat \tau_{nT}^1(\hat \beta_{nT,j}, \cdot), s) - \alpha_0(s) A_1(\tau^1(\beta_{0j}, \cdot), s) \right| \\
    & \le \left| \hat \alpha_{nT}(s) \right| \cdot \left| A_1(\hat \tau_{nT}^1(\hat \beta_{nT,j}, \cdot) - \tau^1(\beta_{0j}, \cdot), s) \right| \\
    & \quad + \left| \hat \alpha_{nT}(s) - \alpha_0(s) \right| \cdot \left| A_1(\tau^1(\beta_{0j}, \cdot), s) \right| \\
    & \lesssim_p (\bar \alpha_0 + c_n)(\bar \alpha_0 c_n + c_n) + \bar \alpha_0 c_n \\
    & \lesssim_p \bar \alpha_0^2 c_n + \bar \alpha_0 c_n + c_n^2.
\end{align}
Thus, when $\bar \alpha_0 > 0$, repeating the same computation recursively, we can obtain
\begin{align}
    \left| \hat \tau_{nT}^\ell(\hat \beta_{nT,j}, s) - \tau^\ell(\beta_{0j}, s) \right| \lesssim_p \bar \alpha_0^{\ell - 1} c_n
\end{align}
for general $\ell \ge 1$.
From a straightforward calculation, we have $\sum_{\ell = 1}^S \bar \alpha_0^{\ell - 1} c_n = c_n (1 - \bar \alpha_0^S)/(1 - \bar \alpha_0)$.

When $\bar \alpha_0 = 0$, $\tau^\ell(\beta_{0j}, s) = 0$ for all $\ell \ge 1$, while recursively applying Assumption \ref{as:A} gives
\begin{align}
    \left| \hat \tau_{nT}^\ell(\hat \beta_{nT,j}, s) \right| \lesssim_p c_n^\ell
\end{align}
uniformly in $s \in [0,1]$.
Since $c_n \to 0$, we have $\sum_{\ell = 1}^S c_n^\ell \lesssim c_n$ for all sufficiently large $nT$.
Hence, in either case, $||\hat M^S_{nT}(i,j,s) - M^S(i,j,s)||_\infty \lesssim_p c_n$.

Next, observe that $M^S(i,j,s) - M(i,j,s) = -\sum_{\ell = S + 1}^\infty W_n^\ell \bm{e}_i \tau^\ell(\beta_{0j}, s)$ and that
\begin{align}
    |\tau^0(\beta_{0j}, s)| \le \bar \beta_{0j}, \;\;
    |\tau^1(\beta_{0j}, s)| \le \bar \alpha_0 \bar \beta_{0j}, \;\;
    \ldots, \;\; |\tau^\ell(\beta_{0j}, s)| \le \bar \alpha_0^\ell \bar \beta_{0j}
\end{align}
by repeatedly applying Assumption \ref{as:A}, where $\bar \beta_{0j} \coloneqq \sup_{s \in [0,1]}|\beta_{0j}(s)|$.
Hence, we have
\begin{align}
    \left\| M^S(i,j,s) - M(i,j,s) \right\|_\infty \lesssim \sum_{\ell = S + 1}^\infty |\tau^\ell(\beta_{0j}, s)| \le \frac{\bar \beta_{0j}}{1 - \bar \alpha_0} \cdot \bar \alpha_0^{S+1}.
\end{align}
Combining these results completes the proof.

\qed

\section{Consistent variance estimation}\label{app:covmat}

First, observe the following alternative representations of $\mathcal{V}_{z, nT}$ and $\mathcal{V}_{ab, nT}$:
\begin{align}
    \mathcal{V}_{z, nT}
    & = \frac{1}{L^2 N} \sum_{l = 1}^L \sum_{l' = 1}^L \sum_{t = 1}^{T-1} \sum_{t': \: |t' - t| \le 1} \sum_{i = 1}^n \vec z_{it}(s_l) \vec z_{it'}(s_{l'})^\top  \bbE[ \vec \varepsilon_{i t}(s_l) \vec \varepsilon_{i t'}(s_{l'}) ] \\
    \mathcal{V}_{ab, nT}
    & = \frac{2}{L^2 N} \sum_{l = 1}^L \sum_{l' = 1}^L \sum_{t = 1}^{T-1} \sum_{t': \: |t' - t| \le 1} \sum_{1 \le i, j \le n} p_{a, i, j} p_{b, i, j} \bbE[ \vec \varepsilon_{i t}(s_l) \vec \varepsilon_{i t'}(s_{l'})] \bbE[ \vec \varepsilon_{j t}(s_l) \vec \varepsilon_{j t'}(s_{l'}) ].
\end{align}
Define $\vec{\hat e}_{it}(s) \coloneqq e_{i,t+1}(s; \hat \theta_{nT}) - e_{it}(s; \hat \theta_{nT})$,
\begin{align}
    \hat{\mathcal{V}}_{z, nT}
    & \coloneqq \frac{1}{L^2 N} \sum_{l = 1}^L \sum_{l' = 1}^L \sum_{t = 1}^{T-1} \sum_{t': \: |t' - t| \le 1} \sum_{i = 1}^n \vec z_{it}(s_l) \vec z_{it'}(s_{l'})^\top \vec{\hat e}_{it}(s_l) \vec{\hat e}_{it'}(s_{l'}) \\
    \hat{\mathcal{V}}_{ab, nT}
    & \coloneqq \frac{2}{L^2 N} \sum_{l = 1}^L \sum_{l' = 1}^L \sum_{t = 1}^{T-1} \sum_{t': \: |t' - t| \le 1} \sum_{1 \le i, j \le n} p_{a, i, j} p_{b, i, j} \vec{\hat e}_{it}(s_l) \vec{\hat e}_{it'}(s_{l'}) \vec{\hat e}_{jt}(s_l) \vec{\hat e}_{jt'}(s_{l'}) \\
    \hat{\mathcal{V}}_{nT}
    & \coloneqq \left( \begin{array}{cccc}
    \hat{\mathcal{V}}_{z, nT} & \bm{0}_{(d_q + d_x)K \times 1} & \cdots & \bm{0}_{(d_q + d_x)K \times 1} \\
     \bm{0}_{1 \times (d_q + d_x)K} & \hat{\mathcal{V}}_{11, nT} & \cdots & \hat{\mathcal{V}}_{1M, nT} \\
     \vdots &  \vdots & \ddots & \vdots \\
    \bm{0}_{1 \times (d_q + d_x)K} &  \hat{\mathcal{V}}_{M1, nT} & \cdots & \hat{\mathcal{V}}_{MM, nT}
    \end{array} \right) \\
    \hat \Sigma_{nT}
    & \coloneqq \left( \bar J_{nT}(\hat \theta_{nT})^\top \Omega_{nT}  \bar J_{nT}(\hat \theta_{nT}) \right)^{-1}  \bar J_{nT}(\hat \theta_{nT})^\top \Omega_{nT} \hat{\mathcal{V}}_{nT} \Omega_{nT}  \bar J_{nT}(\hat \theta_{nT}) \left(  \bar J_{nT}(\hat \theta_{nT})^\top \Omega_{nT}  \bar J_{nT}(\hat \theta_{nT}) \right)^{-1}.
\end{align}
Then, our standard error estimators for $\sigma_{nT,\alpha}(s)$, $\sigma_{nT,\gamma}(s)$, and $\sigma_{nT,j}(s)$ are given as
\begin{align}\label{eq:se}
    \begin{split}
    \hat \sigma_{nT,\alpha}(s) 
    & \coloneqq \sqrt{\phi^K(s)^\top \mathbb{S}_\alpha \hat \Sigma_{nT} \mathbb{S}_\alpha^\top \phi^K(s)} \\
    \hat \sigma_{nT,\gamma}(s) 
    & \coloneqq \sqrt{\phi^K(s)^\top \mathbb{S}_\gamma \hat \Sigma_{nT} \mathbb{S}_\gamma^\top \phi^K(s)} \\
    \hat \sigma_{nT,j}(s)
    & \coloneqq \sqrt{\phi^K(s)^\top \mathbb{S}_j \hat \Sigma_{nT} \mathbb{S}_j^\top \phi^K(s)},    
    \end{split}
\end{align}
respectively.

\begin{proposition}[Consistent variance estimation]\label{prop:var}
    Suppose that the assumptions in Theorem \ref{thm:normality} are satisfied.
    In addition, assume that $K^3/(nT) \to 0$ and $K^{3/2 - \varsigma} \to 0$ as $nT \to \infty$. 
    Then, 
    \begin{itemize}
        \item[(i)] $\left\|\hat{\mathcal{V}}_{z, nT} - \mathcal{V}_{z, nT} \right\| \lesssim_p K^{3/2}/\sqrt{nT} + K^{3/2 - \varsigma}$
        \item[(ii)] $\left\|\hat{\mathcal{V}}_{ab, nT} - \mathcal{V}_{ab, nT} \right\| \lesssim_p \sqrt{K} / \sqrt{nT} + K^{1/2 - \varsigma}$ for all $1 \le a,b \le M$
        \item[(iii)] $\hat \sigma_{nT,\alpha}(s)/ \sigma_{nT,\alpha}(s) \overset{p}{\to} 1$
        \item[(iv)] $\hat \sigma_{nT,\gamma}(s)/ \sigma_{nT,\gamma}(s) \overset{p}{\to} 1$
        \item[(v)] $\hat \sigma_{nT,j}(s) / \sigma_{nT,j}(s) \overset{p}{\to} 1$ for all $j \in [d_x]$.        
    \end{itemize}
\end{proposition}

\begin{proof}
(i) Decompose
\begin{align}
    \hat{\mathcal{V}}_{z, nT} - \mathcal{V}_{z, nT} = \left( \hat{\mathcal{V}}_{z, nT} - \tilde{\mathcal{V}}_{z, nT} \right) + \left( \tilde{{\mathcal{V}}}_{z, nT} - \mathcal{V}_{z, nT} \right),
\end{align}
where
\begin{align}
    \tilde{\mathcal{V}}_{z, nT}
    & = \frac{1}{L^2 N} \sum_{l = 1}^L \sum_{l' = 1}^L \sum_{t = 1}^{T-1} \sum_{t': \: |t' - t| \le 1} \sum_{i = 1}^n \vec z_{it}(s_l) \vec z_{it'}(s_{l'})^\top \vec \varepsilon_{it}(s_l) \vec \varepsilon_{it'}(s_{l'}).
\end{align}
Since
\begin{align}
    \vec{\hat e}_{it}(s)
    = \underbracket{\vec a_{1, it}(s) [\alpha_0(s) - \hat \alpha_{nT}(s)] + \vec a_{2, it}(s) [\gamma_0(s) - \hat \gamma_{nT}(s)] + \sum_{j = 1}^{d_x} \vec X_{it}^j [\beta_{0j}(s) - \hat \beta_{nT,j}(s)]}_{\eqqcolon b_{it}(s)} + \vec \varepsilon_{it}(s),
\end{align}
we have
\begin{align}
    \vec{\hat e}_{it}(s_l) \vec{\hat e}_{it'}(s_{l'})
    & = \vec \varepsilon_{it}(s_l) \vec \varepsilon_{it'}(s_{l'}) + b_{it}(s_l) b_{it'}(s_{l'}) + b_{it}(s_l) \vec \varepsilon_{it'}(s_{l'}) + \vec \varepsilon_{it}(s_l) b_{it'}(s_{l'}).
\end{align}
Then, we can write
\begin{align}
    \hat{\mathcal{V}}_{z, nT} - \tilde{\mathcal{V}}_{z, nT}
    & = \frac{1}{L^2 N} \sum_{l = 1}^L \sum_{l' = 1}^L \sum_{t = 1}^{T-1} \sum_{t': \: |t' - t| \le 1} \sum_{i = 1}^n \vec z_{it}(s_l) \vec z_{it'}(s_{l'})^\top [\vec{\hat e}_{it}(s_l) \vec{\hat e}_{it'}(s_{l'}) - \vec \varepsilon_{it}(s_l) \vec \varepsilon_{it'}(s_{l'})] \\
    & = \frac{1}{L^2 N} \sum_{l = 1}^L \sum_{l' = 1}^L \sum_{t = 1}^{T-1} \sum_{t': \: |t' - t| \le 1} \sum_{i = 1}^n \vec z_{it}(s_l) \vec z_{it'}(s_{l'})^\top b_{it}(s_l) b_{it'}(s_{l'}) \\
    & \quad + \frac{1}{L^2 N} \sum_{l = 1}^L \sum_{l' = 1}^L \sum_{t = 1}^{T-1} \sum_{t': \: |t' - t| \le 1} \sum_{i = 1}^n \vec z_{it}(s_l) \vec z_{it'}(s_{l'})^\top b_{it}(s_l) \vec \varepsilon_{it'}(s_{l'}) \\
    & \quad + \frac{1}{L^2 N} \sum_{l = 1}^L \sum_{l' = 1}^L \sum_{t = 1}^{T-1} \sum_{t': \: |t' - t| \le 1} \sum_{i = 1}^n \vec z_{it}(s_l) \vec z_{it'}(s_{l'})^\top \vec \varepsilon_{it}(s_l) b_{it'}(s_{l'}).
\end{align}
In view of Theorem \ref{thm:roc}(ii)--(iv), we can easily find that $|b_{it}(s)| \lesssim_p \sqrt{K}/\sqrt{nT} + K^{1/2 - \varsigma}$ uniformly in $s$ and $(i,t)$.
In addition, it is straightforward to see that $||\vec z_{it}(s_l) \vec z_{it'}(s_{l'})^\top|| \lesssim K$, and $||\vec z_{it}(s_l) \vec z_{it'}(s_{l'})^\top \vec \varepsilon_{it}(s_l)||$, $||\vec z_{it}(s_l) \vec z_{it'}(s_{l'})^\top \vec \varepsilon_{it'}(s_{l'})|| \lesssim_p K$ by Markov's inequality.
Hence, we have $||\hat{\mathcal{V}}_{z, nT} - \tilde{\mathcal{V}}_{z, nT}|| \lesssim_p K^{3/2}/\sqrt{nT} + K^{3/2 - \varsigma}$.

Meanwhile, it is not difficult to see that $||\tilde{\mathcal{V}}_{z, nT} - \mathcal{V}_{z, nT}|| \lesssim_p K/\sqrt{nT}$ by Markov's inequality.
Then, the result follows from the triangle inequality.

\bigskip

(ii) Similar to the above, we decompose
\begin{align}
    \hat{\mathcal{V}}_{ab, nT} - \mathcal{V}_{ab, nT} = \left( \hat{\mathcal{V}}_{ab, nT} - \tilde{\mathcal{V}}_{ab, nT} \right) + \left( \tilde{{\mathcal{V}}}_{ab, nT} - \mathcal{V}_{ab, nT} \right),
\end{align}
where
\begin{align}
    \tilde{\mathcal{V}}_{ab, nT}
    & = \frac{2}{L^2 N} \sum_{l = 1}^L \sum_{l' = 1}^L \sum_{t = 1}^{T-1} \sum_{t': \: |t' - t| \le 1} \sum_{1 \le i, j \le n} p_{a, i, j} p_{b, i, j} \vec \varepsilon_{it}(s_l) \vec \varepsilon_{it'}(s_{l'}) \vec \varepsilon_{jt}(s_l) \vec \varepsilon_{jt'}(s_{l'}).
\end{align}
For the first term on the right-hand side, noting that
\begin{align}
    \vec{\hat e}_{it}(s_l) \vec{\hat e}_{it'}(s_{l'}) \vec{\hat e}_{jt}(s_l) \vec{\hat e}_{jt'}(s_{l'})
    & = \{\vec \varepsilon_{it}(s_l) \vec \varepsilon_{it'}(s_{l'}) + b_{it}(s_l) b_{it'}(s_{l'}) + b_{it}(s_l) \vec \varepsilon_{it'}(s_{l'}) + \vec \varepsilon_{it}(s_l) b_{it'}(s_{l'})\} \\
    & \quad \times \{\vec \varepsilon_{jt}(s_l) \vec \varepsilon_{jt'}(s_{l'}) + b_{jt}(s_l) b_{jt'}(s_{l'}) + b_{jt}(s_l) \vec \varepsilon_{jt'}(s_{l'}) + \vec \varepsilon_{jt}(s_l) b_{jt'}(s_{l'})\} \\
    & = \vec \varepsilon_{it}(s_l) \vec \varepsilon_{it'}(s_{l'}) \vec \varepsilon_{jt}(s_l) \vec \varepsilon_{jt'}(s_{l'}) \\
    & \quad + \vec \varepsilon_{it}(s_l) \vec \varepsilon_{it'}(s_{l'}) \vec \varepsilon_{jt'}(s_{l'}) b_{jt}(s_l) + \;\; \cdots \text{(three $\vec \varepsilon$'s * one $b$)} \\
    & \quad + \vec \varepsilon_{it}(s_l) \vec \varepsilon_{it'}(s_{l'}) b_{jt}(s_l) b_{jt'}(s_{l'}) + \;\; \cdots \text{(two $\vec \varepsilon$'s * two $b$'s)} \\
    & \quad + \vec \varepsilon_{it'}(s_{l'}) b_{it}(s_l) b_{jt}(s_l) b_{jt'}(s_{l'}) + \;\; \cdots \text{(one $\vec \varepsilon$ * three $b$'s)} \\
    & \quad + b_{it}(s_l) b_{it'}(s_{l'}) b_{jt}(s_l) b_{jt'}(s_{l'}),
\end{align}
it is not difficult to show that $|\hat{\mathcal{V}}_{ab, nT} - \tilde{\mathcal{V}}_{ab, nT}| \lesssim_p \sqrt{K}/\sqrt{nT} + K^{1/2 - \varsigma}$.

For the second term, following the analogous argument as in the proof of Proposition 2 \cite{lin2010gmm}, $|\tilde{{\mathcal{V}}}_{ab, nT} - \mathcal{V}_{ab, nT}| \lesssim_p 1/\sqrt{nT}$ holds.
Hence, we obtain the desired result.

\bigskip

(iii) To prove the result, it suffices to show that
\begin{align}
    \left| \frac{[\hat \sigma_{nT,\alpha}(s)]^2 - [\sigma_{nT,\alpha}(s)]^2}{[\sigma_{nT,\alpha}(s)]^2} \right| \overset{p}{\to} 0.
\end{align}
As shown in the proof of Theorem \ref{thm:normality}, $[\sigma_{nT,\alpha}(s)]^2$ is bounded below from $c||\phi^K(s)||^2$.
On the other hand, writing $R_{nT} \coloneqq \Omega_{nT} \bar J_{nT} \left( \bar J^\top_{nT} \Omega_{nT} \bar J_{nT} \right)^{-1}$ and its estimator counterpart as $\hat R_{nT}$, by the triangle inequality,
\begin{align}
    \left| [\hat \sigma_{nT,\alpha}(s)]^2 - [\sigma_{nT,\alpha}(s)]^2 \right|
    & = \left| \phi^K(s)^\top \mathbb{S}_\alpha \left[\hat \Sigma_{nT} - \Sigma_{nT}\right] \mathbb{S}_\alpha^\top \phi^K(s) \right| \\
    & \le \left| \phi^K(s)^\top \mathbb{S}_\alpha \left[\hat R_{nT}^\top \hat{\mathcal{V}}_{nT} \hat R_{nT} - \hat R_{nT}^\top \mathcal{V}_{nT} \hat R_{nT}\right] \mathbb{S}_\alpha^\top \phi^K(s) \right| \\
    & \quad + \left| \phi^K(s)^\top \mathbb{S}_\alpha \left[\hat R_{nT}^\top \mathcal{V}_{nT} \hat R_{nT} - R_{nT}^\top \mathcal{V}_{nT} R_{nT}\right] \mathbb{S}_\alpha^\top \phi^K(s) \right| \\
    & \le \left\|\hat{\mathcal{V}}_{nT} - \mathcal{V}_{nT}\right\| \cdot \left\|\hat R_{nT} \mathbb{S}_\alpha^\top \phi^K(s)\right\|^2 + \lambda_{\max}(\mathcal{V}_{nT}) \cdot \left\|\left\{\hat R_{nT} - R_{nT}\right\} \mathbb{S}_\alpha^\top \phi^K(s)\right\|^2 \\
    & \quad + 2 \left| \phi^K(s)^\top \mathbb{S}_\alpha R_{nT}^\top \mathcal{V}_{nT} \left\{\hat R_{nT} - R_{nT}\right\} \mathbb{S}_\alpha^\top \phi^K(s) \right| \\
    & \lesssim_p \left(K^{3/2}/\sqrt{nT} + K^{3/2 - \varsigma}\right) \left\|\phi^K(s)\right\|^2,
\end{align}
where the last inequality is from $||\hat R_{nT} - R_{nT}|| \lesssim_p K/\sqrt{nT} + K^{-\varsigma}$ by Lemma \ref{lem:matLLN2}(i) and (iv) and $||\hat{\mathcal{V}}_{nT} - \mathcal{V}_{nT}|| \lesssim_p K^{3/2}/\sqrt{nT} + K^{3/2 - \varsigma}$ by results (i) and (ii).
Thus, we have
\begin{align}
    \left| \frac{[\hat \sigma_{nT,\alpha}(s)]^2 - [\sigma_{nT,\alpha}(s)]^2}{[\sigma_{nT,\alpha}(s)]^2} \right| \lesssim_p K^{3/2}/\sqrt{nT} + K^{3/2 - \varsigma} \to 0.
\end{align}

\bigskip

(iv) -- (v) Analogous to the proof of (iii).
\end{proof}


\section{Details of the Monte Carlo experiments}\label{app:MC}

\subsection{Setup}
Throughout the experiments, we consider the following data-generating process (DGP):
\begin{align}
    Y_{it}(s) = \alpha_0(s) \int_0^1 \bar Y_{it}(u) \nu(u,s)\text{d}u + \gamma_0(s) Y_{i,t-1}(s) + X_{it}\beta_0(s) + f_{0i}(s) + \varepsilon_{it}(s),
\end{align}
where $X_{it} \sim \mathcal{N}(0,1)$, $\alpha_0(s) = \phi(s; 0.4, 0.6^2) + 0.1s - 0.4s^2$, $\phi(\cdot ; \mu, \sigma^2)$ denotes the standard normal density function with mean $\mu$ and variance $\sigma^2$, $\nu(u,s) = 0.75(1 - (u-s)^2)$ (i.e., the Epanechnikov kernel function), $\gamma_0(s) = 0.4 - 0.3 \Phi(3s - 2)$, $\Phi$ denotes the standard normal distribution function, and the individual fixed effects are given by $f_{0i}(s) = 1 + \cos(is)$.
The coefficient function is given by 
\begin{align}
    \beta_0(s) =r  (0.5 + s + 0.5\sin(\pi s)),
\end{align}
where $r \in \{0.4, 1\}$. 
We use $Q_{it} = (\bar X_{it}, \bar{\bar X}_{it}, X_{i,t-1})$ as the IVs for $\left(\int_0^1 \bar Y_{it}(u) \nu(u,s)\text{d}u, Y_{i,t-1}(s) \right)$, and, thus, the magnitude of $r$ determines the strength of these IVs. 
For the error term, we generate $\varepsilon_{it}(s) = \sqrt{1 + \text{deg}_i}(e_{1,it},e_{2,it},e_{3,it})^\top (1, s, s^2)$, where $\text{deg}_i$ denotes the number of units connected to $i$ (i.e., $i$'s degree), and $e_{j,it} \sim \mathcal{N}(0,0.3^2)$ for $j = 1,2,3$. 
The weight matrix $W_n$ is a row-normalized adjacency matrix, which is constructed by randomly placing $n$ units on a $[\sqrt{2n}] \times [\sqrt{2n}]$ lattice, where $[a]$ denotes the nearest integer to $a$. Any two units are connected if the Euclidean distance between them is exactly one.
The sample size $n$ is chosen from $n \in \{50, 100\}$, and $T$ is from $T \in \{5, 10\}$.

The outcome functions are generated recursively using the Neumann series approximation: $Y_t \approx Y_t^{(S)} \coloneqq \sum_{\ell = 0}^S \mathcal{A}_1^\ell [\gamma_0 Y_{t-1} + X_t \beta_0 + F_0 + \mathcal{E}_t]$ for $t \in 0 \cup [T]$ with $Y_{-1} = \bm 0$, where $S$ is increased iteratively until $\max_{i \in [n]} |Y_{it}^{(S)}(s) - Y_{it}^{(S-1)}(s)| < 0.001$ is satisfied for all $t$ and $s$.
Throughout the simulations, integrals over $[0,1]$ are approximated by finite summations over 99 equally-spaced grid points.

For the choice of basis functions, we use the cubic B-splines orthonormalized via the Gram--Schmidt procedure.
The number of inner knots for the B-splines is selected using cross-validation.
Specifically, we use the first $T_{\text{train}}$ periods as the training data and the remaining $T - T_{\text{train}}$ periods as the validation data.
When $T = 5$, we set $T_{\text{train}} = 3$, and when $T = 10$, we set $T_{\text{train}} = 6$.
Once the parameters are estimated, we compute the predicted outcome difference as
\begin{align}
    \hat Y_{t + 1}^{(S)} - \hat Y_t^{(S)}
    \coloneqq \sum_{\ell = 0}^S \hat{\mathcal{A}}_{nT_{\text{train}}}^\ell \left\{ \hat \gamma_{nT_{\text{train}}} (Y_{t} - Y_{t-1}) + (X_{t+1} - X_t) \hat \beta_{nT_{\text{train}}} \right\},
\end{align}
where $\hat{\mathcal{A}}_{nT_{\text{train}}}$ is obtained by replacing $\alpha_0$ with $\hat \alpha_{nT_{\text{train}}}$.
We then compute the average mean squared prediction error (AMSPE) as
\begin{align}
    \text{AMSPE}
    \coloneqq
    \frac{1}{99} \sum_{l = 1}^{99} \left( \frac{1}{n (T - T_{\text{train}} - 1)} \sum_{i = 1}^n \sum_{t = T_{\text{train}} + 1}^{T - 1} \left\{ [\hat Y_{i,t+1}(s_l) - \hat Y_{it}(s_l)] - [Y_{i,t+1}(s_l) - Y_{it}(s_l)] \right\}^2 \right),
\end{align}
where $\hat Y_{i,t+1}(s_l) - \hat Y_{it}(s_l)$ denotes the $i$-th element of $\hat Y_{t + 1}^{(S)} - \hat Y_t^{(S)}$ evaluated at $s_l$.
In view of Remark \ref{rem:K}, we set the number of inner knots to $\lfloor c_K (nT)^{1/5} \rfloor$, where $c_K$ is a constant chosen to minimize the AMSPE. In this numerical study, considering computational costs, we conduct a grid search over
$c_K \in \{0.4, 0.8, 1.2, 1.6\}$.
Moreover, instead of performing the cross-validation in every simulation run, we compute the optimal $c_K$ based on 10 pilot simulations under each design.
We then average the selected values of $c_K$ across these pilot runs and use the resulting value throughout the entire simulation analysis.
As a result, we find that $c_K \approx 0.93$ is approximately optimal in our setting.

The number of grid points used to evaluate the moment function is chosen from $L = c_L K$, where $c_L \in \{2, 3, 4\}$, with the points $0 < s_1 < \dots < s_L < 1$ evenly spaced over $[0,1]$.
For the quadratic moments, we use two weight matrices ($M =2$): $P_{1,1} = W_n$ and $P_{2,1} = W_n^\top W_n - \text{diag}(W_n^\top W_n)$. 
We then compare the performance of three different estimators: \textbf{GMM 1}: the integrated-GMM estimator using the weight matrix given in  \eqref{eq:2slsweight}, \textbf{GMM 2}: the integrated-GMM estimator using the identity weight matrix, and (integrated) \textbf{2SLS}: GMM 1 estimator without utilizing the quadratic moment conditions.

\subsection{Results}
For each setup, we generate the dataset 500 times. 
The performance of the estimators is evaluated based on the average bias (BIAS) and the average root-mean-squared error (RMSE).
Specifically, the BIAS and RMSE of estimating $\alpha_0$ are defined as  
\begin{align}
    \text{BIAS}(\alpha_0) 
    & \coloneqq  \frac{1}{500}\sum_{b = 1}^{500} \left[ \frac{1}{99}\sum_{l = 1}^{99} \left( \hat \alpha_{nT}^{(b)} (s_l) - \alpha_0(s_l) \right) \right] \\
    \text{RMSE}(\alpha_0) 
    & \coloneqq \frac{1}{500}\sum_{b = 1}^{500} \sqrt{ \frac{1}{99}\sum_{l = 1}^{99} \left( \hat \alpha_{nT}^{(b)} (s_l) - \alpha_0(s_l) \right)^2 }
\end{align}
respectively.
Here, $\hat \alpha_{nT}^{(b)}$ denotes the estimator of $\alpha_0$ obtained from the $b$-th replicated dataset. 
The BIAS and RMSE for $\gamma_0$ and $\beta_0$ are defined analogously.

In computing the integrated-GMM estimators, we use the \texttt{nloptr} function in \texttt{R}, with the 2SLS estimator used to obtain the initial parameter value for optimization.
The simulation results for the estimation of $\alpha_0$, $\gamma_0$, and $\beta_0$ are summarized in Tables \ref{MC:alpha}, \ref{MC:gamma}, and \ref{MC:beta}, respectively.

\begin{table}[!ht]
\centering
\caption{Simulation result: $\alpha_0$}\label{MC:alpha}
\small
\begin{tabular}{llll|cccccc}
\hline\hline
 & & & & \multicolumn{2}{c}{GMM 1} & \multicolumn{2}{c}{GMM 2} & \multicolumn{2}{c}{2SLS} \\
$n$ & $T$ & $c_L$ & $r$ & BIAS & RMSE & BIAS & RMSE & BIAS & RMSE \\
\hline
50 & 5 & 2 & 0.4 & -0.0204 & 0.1912 & -0.0233 & 0.2030 & -0.0143 & 0.2022 \\
 &  &  & 1 & -0.0088 & 0.0767 & -0.0161 & 0.1013 & -0.0072 & 0.0794 \\
 &  & 3 & 0.4 & -0.0210 & 0.1917 & -0.0236 & 0.2020 & -0.0143 & 0.2021 \\
 &  &  & 1 & -0.0092 & 0.0768 & -0.0156 & 0.0967 & -0.0072 & 0.0794 \\
 &  & 4 & 0.4 & -0.0208 & 0.1924 & -0.0242 & 0.2021 & -0.0143 & 0.2020 \\
 &  &  & 1 & -0.0094 & 0.0769 & -0.0153 & 0.0967 & -0.0072 & 0.0794 \\
 & 10 & 2 & 0.4 & 0.0049 & 0.1279 & 0.0023 & 0.1350 & 0.0053 & 0.1294 \\
 &  &  & 1 & -0.0008 & 0.0499 & -0.0011 & 0.0624 & 0.0011 & 0.0515 \\
 &  & 3 & 0.4 & 0.0049 & 0.1279 & 0.0027 & 0.1338 & 0.0053 & 0.1294 \\
 &  &  & 1 & -0.0005 & 0.0501 & -0.0010 & 0.0619 & 0.0011 & 0.0515 \\
 &  & 4 & 0.4 & 0.0049 & 0.1279 & 0.0028 & 0.1335 & 0.0053 & 0.1294 \\
 &  &  & 1 & -0.0001 & 0.0504 & -0.0009 & 0.0613 & 0.0011 & 0.0515 \\
100 & 5 & 2 & 0.4 & 0.0147 & 0.1449 & 0.0125 & 0.1503 & 0.0168 & 0.1460 \\
 &  &  & 1 & 0.0030 & 0.0561 & 0.0023 & 0.0649 & 0.0055 & 0.0577 \\
 &  & 3 & 0.4 & 0.0152 & 0.1449 & 0.0132 & 0.1491 & 0.0168 & 0.1459 \\
 &  &  & 1 & 0.0033 & 0.0563 & 0.0022 & 0.0647 & 0.0055 & 0.0576 \\
 &  & 4 & 0.4 & 0.0152 & 0.1449 & 0.0134 & 0.1488 & 0.0168 & 0.1459 \\
 &  &  & 1 & 0.0033 & 0.0563 & 0.0022 & 0.0641 & 0.0055 & 0.0576 \\
 & 10 & 2 & 0.4 & -0.0009 & 0.0904 & -0.0009 & 0.0916 & -0.0005 & 0.0906 \\
 &  &  & 1 & -0.0014 & 0.0358 & -0.0007 & 0.0385 & -0.0007 & 0.0363 \\
 &  & 3 & 0.4 & -0.0009 & 0.0904 & -0.0009 & 0.0914 & -0.0005 & 0.0906 \\
 &  &  & 1 & -0.0014 & 0.0360 & -0.0009 & 0.0387 & -0.0007 & 0.0363 \\
 &  & 4 & 0.4 & -0.0009 & 0.0904 & -0.0010 & 0.0915 & -0.0005 & 0.0906 \\
 &  &  & 1 & -0.0013 & 0.0361 & -0.0007 & 0.0378 & -0.0007 & 0.0363 \\
\hline
\end{tabular}

\small
\flushleft
GMM 1: integrated-GMM estimator with the weight matrix in \eqref{eq:2slsweight}.
GMM 2: integrated-GMM estimator with the identity weight matrix.
\normalsize
\end{table}

\begin{table}[!ht]
\centering
\caption{Simulation result: $\gamma_0$}\label{MC:gamma}
\small
\begin{tabular}{llll|cccccc}
\hline\hline
 & & & & \multicolumn{2}{c}{GMM 1} & \multicolumn{2}{c}{GMM 2} & \multicolumn{2}{c}{2SLS} \\
$n$ & $T$ & $c_L$ & $r$ & BIAS & RMSE & BIAS & RMSE & BIAS & RMSE \\
\hline
50 & 5 & 2 & 0.4 & 0.0031 & 0.1140 & 0.0195 & 0.1268 & -0.0048 & 0.1161 \\
 &  &  & 1 & 0.0004 & 0.0441 & 0.0111 & 0.0575 & -0.0022 & 0.0446 \\
 &  & 3 & 0.4 & 0.0023 & 0.1142 & 0.0164 & 0.1267 & -0.0048 & 0.1159 \\
 &  &  & 1 & 0.0002 & 0.0440 & 0.0105 & 0.0563 & -0.0022 & 0.0446 \\
 &  & 4 & 0.4 & 0.0020 & 0.1144 & 0.0159 & 0.1275 & -0.0048 & 0.1158 \\
 &  &  & 1 & 0.0001 & 0.0439 & 0.0105 & 0.0564 & -0.0022 & 0.0445 \\
 & 10 & 2 & 0.4 & 0.0015 & 0.0750 & 0.0032 & 0.0800 & 0.0012 & 0.0749 \\
 &  &  & 1 & 0.0012 & 0.0295 & 0.0033 & 0.0336 & 0.0004 & 0.0295 \\
 &  & 3 & 0.4 & 0.0015 & 0.0749 & 0.0037 & 0.0800 & 0.0012 & 0.0749 \\
 &  &  & 1 & 0.0011 & 0.0295 & 0.0033 & 0.0337 & 0.0004 & 0.0295 \\
 &  & 4 & 0.4 & 0.0016 & 0.0749 & 0.0038 & 0.0803 & 0.0012 & 0.0749 \\
 &  &  & 1 & 0.0010 & 0.0295 & 0.0032 & 0.0338 & 0.0004 & 0.0295 \\
100 & 5 & 2 & 0.4 & -0.0080 & 0.0840 & -0.0035 & 0.0887 & -0.0088 & 0.0844 \\
 &  &  & 1 & -0.0025 & 0.0331 & 0.0018 & 0.0383 & -0.0035 & 0.0332 \\
 &  & 3 & 0.4 & -0.0082 & 0.0842 & -0.0041 & 0.0884 & -0.0088 & 0.0843 \\
 &  &  & 1 & -0.0026 & 0.0330 & 0.0015 & 0.0386 & -0.0035 & 0.0332 \\
 &  & 4 & 0.4 & -0.0082 & 0.0842 & -0.0042 & 0.0883 & -0.0088 & 0.0842 \\
 &  &  & 1 & -0.0026 & 0.0330 & 0.0014 & 0.0389 & -0.0035 & 0.0332 \\
 & 10 & 2 & 0.4 & 0.0058 & 0.0522 & 0.0070 & 0.0540 & 0.0056 & 0.0523 \\
 &  &  & 1 & 0.0023 & 0.0208 & 0.0037 & 0.0222 & 0.0021 & 0.0209 \\
 &  & 3 & 0.4 & 0.0057 & 0.0522 & 0.0070 & 0.0539 & 0.0056 & 0.0522 \\
 &  &  & 1 & 0.0022 & 0.0208 & 0.0035 & 0.0220 & 0.0021 & 0.0208 \\
 &  & 4 & 0.4 & 0.0057 & 0.0522 & 0.0066 & 0.0539 & 0.0056 & 0.0522 \\
 &  &  & 1 & 0.0022 & 0.0208 & 0.0032 & 0.0216 & 0.0021 & 0.0208 \\
\hline
\end{tabular}

\small
\flushleft
GMM 1: integrated-GMM estimator with the weight matrix in \eqref{eq:2slsweight}.
GMM 2: integrated-GMM estimator with the identity weight matrix.
\normalsize
\end{table}

\begin{table}[!ht]
\centering
\caption{Simulation result: $\beta_0$}\label{MC:beta}
\small
\begin{tabular}{llll|cccccc}
\hline\hline
 & & & & \multicolumn{2}{c}{GMM 1} & \multicolumn{2}{c}{GMM 2} & \multicolumn{2}{c}{2SLS} \\
$n$ & $T$ & $c_L$ & $r$ & BIAS & RMSE & BIAS & RMSE & BIAS & RMSE \\
\hline
50 & 5 & 2 & 0.4 & 0.0057 & 0.0479 & 0.0131 & 0.0550 & 0.0039 & 0.0485 \\
 &  &  & 1 & 0.0045 & 0.0477 & 0.0180 & 0.0619 & 0.0021 & 0.0479 \\
 &  & 3 & 0.4 & 0.0055 & 0.0479 & 0.0122 & 0.0542 & 0.0039 & 0.0485 \\
 &  &  & 1 & 0.0044 & 0.0476 & 0.0172 & 0.0602 & 0.0021 & 0.0479 \\
 &  & 4 & 0.4 & 0.0055 & 0.0479 & 0.0121 & 0.0543 & 0.0039 & 0.0485 \\
 &  &  & 1 & 0.0043 & 0.0476 & 0.0170 & 0.0600 & 0.0021 & 0.0479 \\
 & 10 & 2 & 0.4 & 0.0009 & 0.0304 & 0.0019 & 0.0317 & 0.0008 & 0.0303 \\
 &  &  & 1 & 0.0008 & 0.0304 & 0.0030 & 0.0342 & -0.0001 & 0.0300 \\
 &  & 3 & 0.4 & 0.0009 & 0.0304 & 0.0019 & 0.0317 & 0.0008 & 0.0303 \\
 &  &  & 1 & 0.0007 & 0.0303 & 0.0029 & 0.0342 & -0.0001 & 0.0300 \\
 &  & 4 & 0.4 & 0.0009 & 0.0304 & 0.0019 & 0.0316 & 0.0008 & 0.0303 \\
 &  &  & 1 & 0.0006 & 0.0303 & 0.0028 & 0.0341 & -0.0001 & 0.0300 \\
100 & 5 & 2 & 0.4 & -0.0027 & 0.0350 & -0.0008 & 0.0361 & -0.0030 & 0.0352 \\
 &  &  & 1 & -0.0027 & 0.0347 & 0.0017 & 0.0391 & -0.0038 & 0.0348 \\
 &  & 3 & 0.4 & -0.0028 & 0.0351 & -0.0011 & 0.0360 & -0.0030 & 0.0352 \\
 &  &  & 1 & -0.0028 & 0.0347 & 0.0014 & 0.0388 & -0.0038 & 0.0348 \\
 &  & 4 & 0.4 & -0.0028 & 0.0351 & -0.0011 & 0.0359 & -0.0030 & 0.0352 \\
 &  &  & 1 & -0.0028 & 0.0347 & 0.0013 & 0.0387 & -0.0038 & 0.0348 \\
 & 10 & 2 & 0.4 & 0.0031 & 0.0221 & 0.0037 & 0.0226 & 0.0031 & 0.0221 \\
 &  &  & 1 & 0.0028 & 0.0220 & 0.0041 & 0.0234 & 0.0025 & 0.0220 \\
 &  & 3 & 0.4 & 0.0031 & 0.0221 & 0.0036 & 0.0226 & 0.0031 & 0.0221 \\
 &  &  & 1 & 0.0028 & 0.0220 & 0.0039 & 0.0232 & 0.0026 & 0.0220 \\
 &  & 4 & 0.4 & 0.0031 & 0.0221 & 0.0034 & 0.0225 & 0.0031 & 0.0221 \\
 &  &  & 1 & 0.0028 & 0.0220 & 0.0037 & 0.0230 & 0.0026 & 0.0220 \\
\hline
\end{tabular}

\small
\flushleft
GMM 1: integrated-GMM estimator with the weight matrix in \eqref{eq:2slsweight}.
GMM 2: integrated-GMM estimator with the identity weight matrix.
\normalsize
\end{table}


\clearpage

\section{Supplementary Material for the Empirical Analysis}\label{app:emp}

\begin{figure}[ht!]
    \centering
    \includegraphics[width = 12cm]{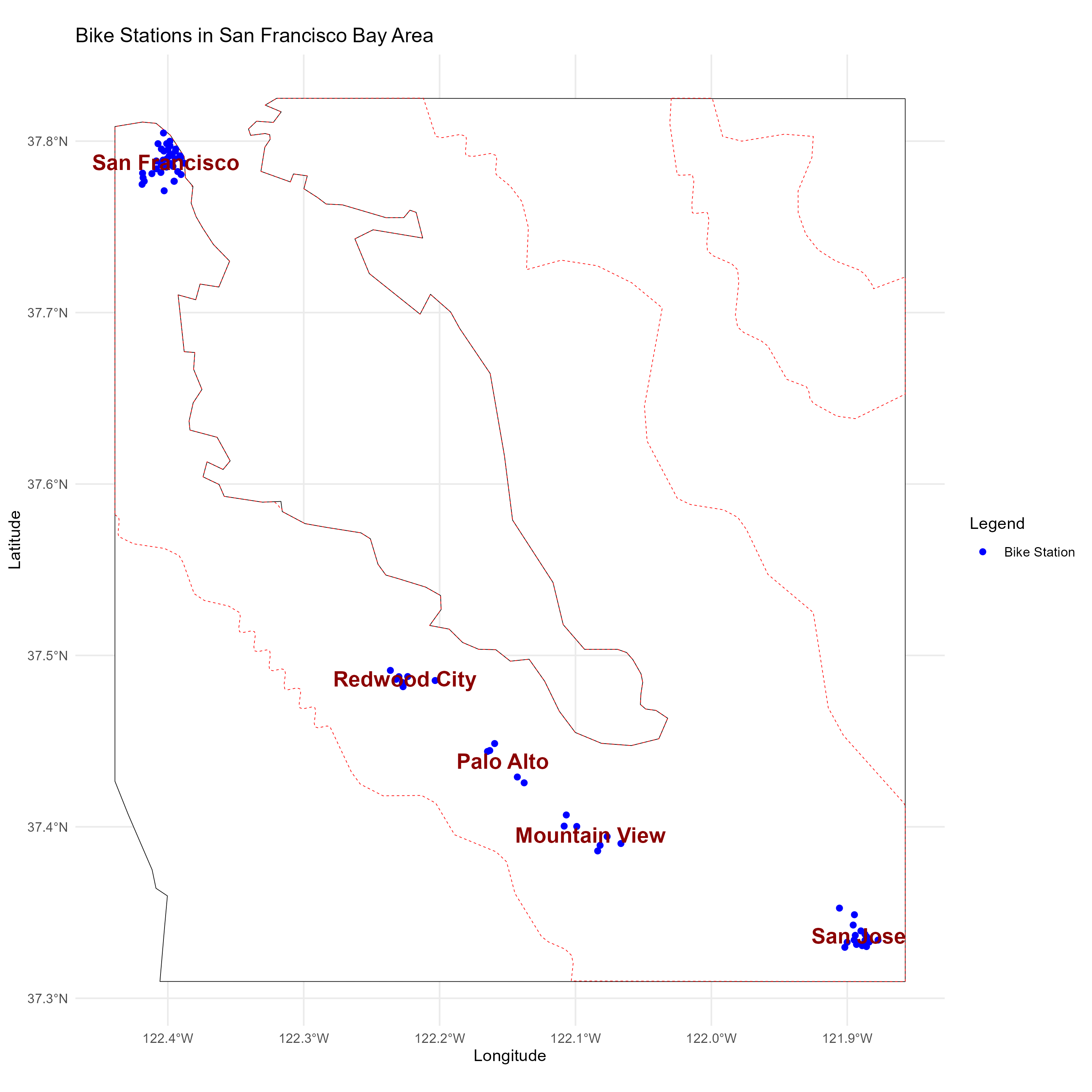}
    \caption{Locations of bike stations}
    \label{fig:map}
\end{figure}

\begin{figure}[ht!]
    \centering
    \includegraphics[width = 15cm]{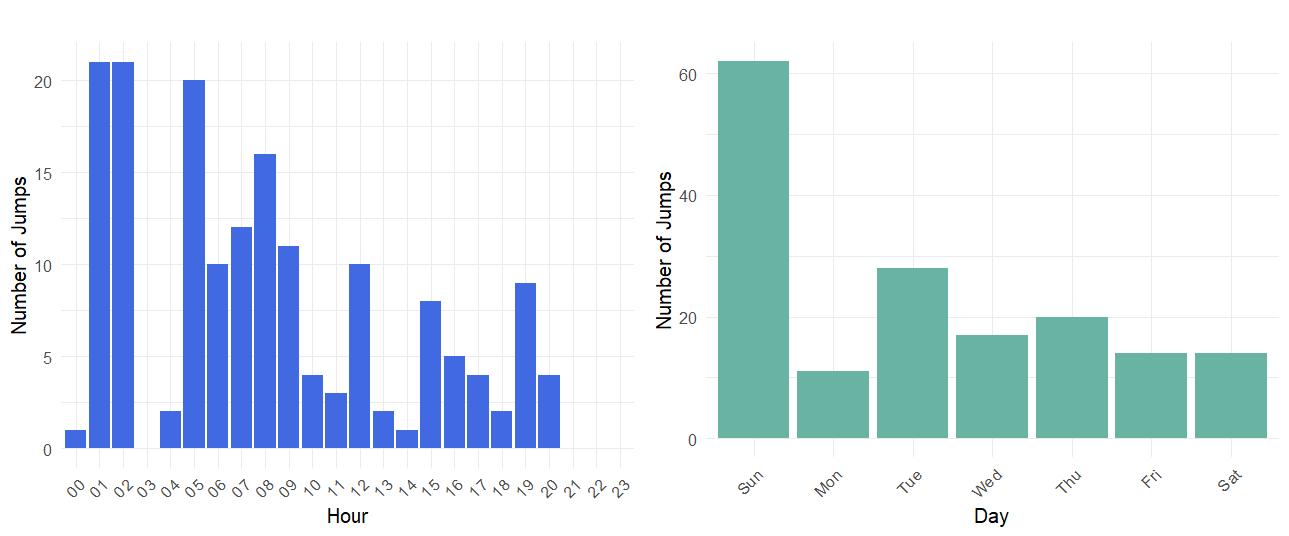}
    \caption{Distribution of potential bike relocation events}
    \label{fig:reloc}
\end{figure}

\begin{figure}
    \centering
    \begin{subfigure}{0.45\textwidth}
        \includegraphics[width=\textwidth]{alpha.png}
        \caption{Spatial interaction parameter}
    \end{subfigure}
    \hfill
    \begin{subfigure}{0.45\textwidth}
        \includegraphics[width=\textwidth]{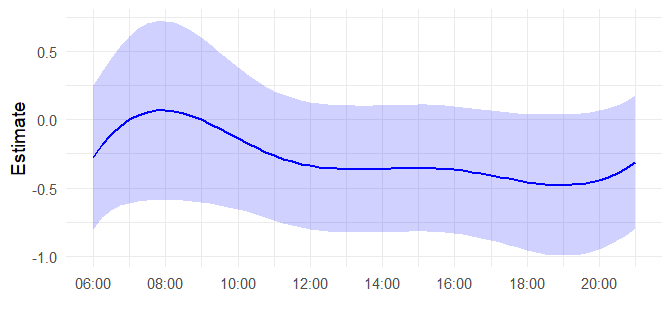}
        \caption{Dynamic interaction parameter}
    \end{subfigure}
    
    \medskip

    \begin{subfigure}{0.45\textwidth}
        \includegraphics[width=\textwidth]{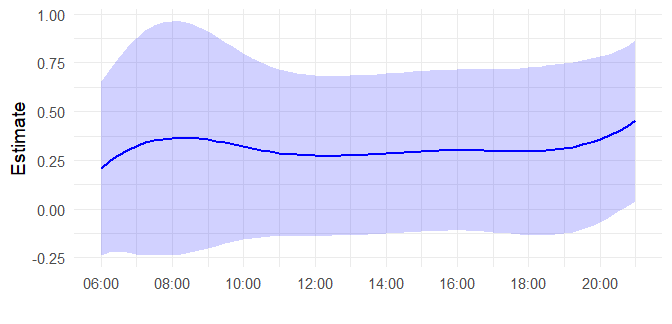}
        \caption{$X^1_{it}$: ratio of round trips}
    \end{subfigure}
    \hfill
    \begin{subfigure}{0.45\textwidth}
        \includegraphics[width=\textwidth]{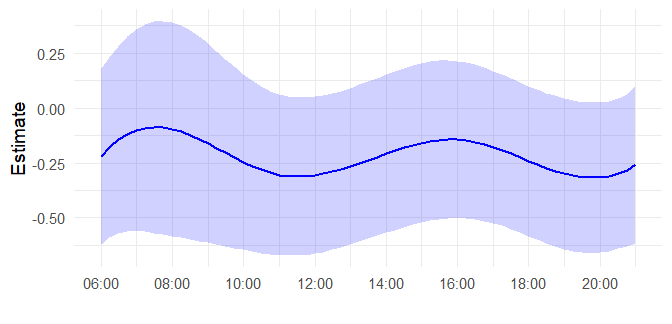}
        \caption{$X^2_{it}$: ratio of departing subscribers}
    \end{subfigure}
    
    \medskip

    \begin{subfigure}{0.45\textwidth}
        \includegraphics[width=\textwidth]{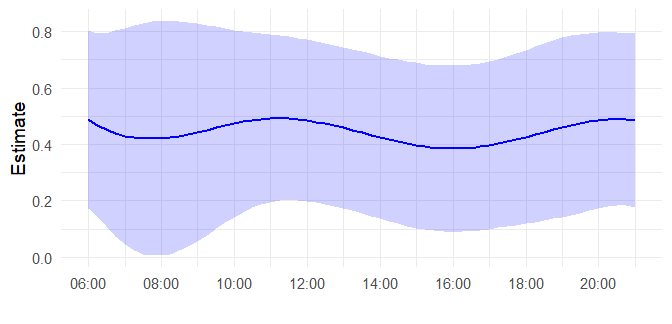}
        \caption{$X^3_{it}$: ratio of arriving subscribers}
    \end{subfigure}
    \hfill
    \begin{subfigure}{0.45\textwidth}
        \includegraphics[width=\textwidth]{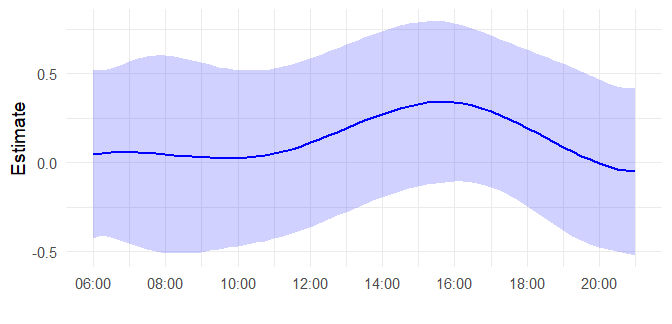}
        \caption{$X^4_{it}$: rainy day}
    \end{subfigure}
    
    \medskip
    
    \begin{subfigure}{0.45\textwidth}
        \includegraphics[width=\textwidth]{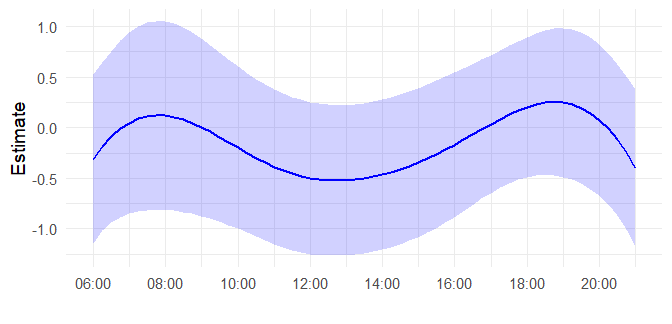}
        \caption{$X^5_{it}$: holiday}
    \end{subfigure}
    \caption{Estimated coefficient functions}
    \label{fig:beta}
\end{figure}

\begin{table}[!ht]
\centering
\caption{Estimation result: non-functional dynamic SAR model with fixed effects}\label{tab:panelNAR}
\small
\begin{tabular}{l|rrr}
\hline\hline
 & Estimate & Standard error & $t$-value \\
\hline
Spatial interaction parameter & 0.1849 & 0.4857 & 0.3807 \\
Dynamic interaction parameter & -0.2645 & 0.2426 & -1.0905 \\
Ratio of round trips & 0.3435 & 0.2295 & 1.4969 \\
Ratio of departing subscribers & -0.1961 & 0.1923 & -1.0199 \\
Ratio of arriving subscribers & 0.4247 & 0.1586 & 2.6775 \\
Rainy day & 0.0556 & 0.2289 & 0.2431 \\
Holiday share & -0.0804 & 0.3546 & -0.2268 \\
Feb & 0.1670 & 0.2368 & 0.7051 \\
Mar & 0.2323 & 0.3123 & 0.7438 \\
Apr & 0.2951 & 0.3775 & 0.7817 \\
May & 0.1398 & 0.3788 & 0.3692 \\
Jun & 0.0417 & 0.3778 & 0.1103 \\
Jul & 0.0067 & 0.3840 & 0.0175 \\
Aug & -0.0128 & 0.3909 & -0.0327 \\
Sep & 0.1440 & 0.4479 & 0.3215 \\
Oct & -0.0745 & 0.3779 & -0.1971 \\
Nov & 0.0678 & 0.3601 & 0.1884 \\
Dec & 0.0973 & 0.2573 & 0.3780 \\
\hline
\end{tabular}

\small
\flushleft
Estimation is performed by 2SLS without quadratic moment conditions.
\normalsize
\end{table}

\clearpage

\section{Proof of Theorem \ref{thm:twoway}}\label{app:twoway}

This appendix outlines the proof of asymptotic normality of the integrated-GMM estimator for the model in \eqref{eq:model_te}.
For simplicity of presentation and ease of comparison, we reuse the same notation as in the main body.
As discussed in Subsection \ref{sub:timeeffect}, our estimator can be defined as follows:
\begin{align}
    \hat \theta_{nT} = (\hat \theta_{nT, \alpha}^\top, \hat \theta_{nT, \gamma}^\top, \hat \theta_{nT, 1}^\top, \ldots, \hat \theta_{nT, d_x}^\top)^\top 
    & \coloneqq \argmin_{\theta \in \Theta_K} \mathcal{Q}_{nT}(\theta), 
\end{align}
where $\mathcal{Q}_{nT}(\theta) \coloneqq \bar g_{nT}(\theta)^\top \Omega_{nT} \bar g_{nT}(\theta)$,
\begin{align}
    g_{nT}(s; \theta) \coloneqq
    \frac{1}{N}\left(\begin{array}{c}
        \bm{Z}(s)^\top \bm D^\top \bm{\calR}^\top \bm{\calR} \bm D \bm{E}(s; \theta) \\
        \bm{E}(s; \theta)^\top \bm D^\top \bm{\calR}^\top P_1 \bm{\calR} \bm D \bm{E}(s; \theta) \\
        \vdots \\
        \bm{E}(s; \theta)^\top \bm D^\top \bm{\calR}^\top P_M \bm{\calR} \bm D \bm{E}(s; \theta)
    \end{array}
    \right),
\end{align}
and $\bar{g}_{nT}(\theta) \coloneqq L^{-1} \sum_{l=1}^L g_{nT}(s_l; \theta)$.
The estimators of $\alpha_0$, $\gamma_0$, and $\beta_0$ can be defined as $\hat \alpha_{nT}(s) \coloneqq \phi^K(s)^\top \hat \theta_{nT, \alpha}$, $\hat \gamma_{nT}(s) \coloneqq \phi^K(s)^\top \hat \theta_{nT, \gamma}$, and $\hat \beta_{nT,j}(s) \coloneqq \phi^K(s)^\top \hat \theta_{nT, j}$, $j \in [d_x]$, respectively.

Then, under the assumptions in Theorem \ref{thm:normality} and Assumption \ref{as:twoway}, we first obtain that, for any given $s \in [0,1]$,
\begin{align}\label{eq:normal_twoway}
    \text{(i)} \;\; & \frac{\sqrt{N}\left( \hat \alpha_{nT}(s) - \alpha_0(s) \right)}{\sigma_{nT,\alpha}(s)} \overset{d}{\to} \mathcal{N}(0,1), \;\; 
    \text{(ii)} \;\; \frac{\sqrt{N}\left( \hat \gamma_{nT}(s) - \gamma_0(s) \right)}{\sigma_{nT,\gamma}(s)} \overset{d}{\to} \mathcal{N}(0,1) \\
    \text{(iii)} \;\; & \frac{\sqrt{N}\left( \hat \beta_{nT, j}(s) - \beta_{0j}(s) \right)}{\sigma_{nT,j}(s)}  \overset{d}{\to} \mathcal{N}(0,1),
\end{align}
where $[\sigma_{nT,\alpha}(s)]^2 \coloneqq \phi^K(s)^\top \mathbb{S}_\alpha \Sigma_{nT} \mathbb{S}_\alpha^\top \phi^K(s)$, $[\sigma_{nT,\gamma}(s)]^2 \coloneqq \phi^K(s)^\top \mathbb{S}_\gamma \Sigma_{nT} \mathbb{S}_\gamma^\top \phi^K(s)$, and $[\sigma_{nT,j}(s)]^2 \coloneqq \phi^K(s)^\top \mathbb{S}_j \Sigma_{nT} \mathbb{S}_j^\top \phi^K(s)$, 
\begin{align}
    \Sigma_{nT} 
    & \coloneqq \left( \bar J^\top_{nT} \Omega_{nT} \bar J_{nT} \right)^{-1} \bar J^\top_{nT} \Omega_{nT} \mathcal{V}_{nT} \Omega_{nT} \bar J_{nT} \left( \bar J^\top_{nT} \Omega_{nT} \bar J_{nT} \right)^{-1}, \\
    \mathcal{V}_{nT} 
    & \coloneqq 
    \left( \begin{array}{cccc}
    \mathcal{V}_{z, nT} & \bm{0}_{(d_q + d_x)K \times 1} & \cdots & \bm{0}_{(d_q + d_x)K \times 1} \\
     \bm{0}_{1 \times (d_q + d_x)K} & \mathcal{V}_{11, nT} & \cdots & \mathcal{V}_{1M, nT} \\
     \vdots &  \vdots & \ddots & \vdots \\
    \bm{0}_{1 \times (d_q + d_x)K} & \mathcal{V}_{M1, nT} & \cdots & \mathcal{V}_{MM, nT}.
    \end{array} \right), \\
    \mathcal{V}_{z, nT}
    & \coloneqq \frac{N}{(NL)^2} \sum_{l = 1}^L \sum_{l' = 1}^L  \bm{Z}(s_l)^\top \bm{D}^\top \bm{ \calR}^\top \bm{ \calR} \bm{D} \bbE[\bm{\mathcal{E}}(s_l) \bm{\mathcal{E}}(s_{l'})^\top]  \bm{D}^\top \bm{\calR}^\top \bm{\calR} \bm{D} \bm{Z}(s_{l'}) \\
    & = \frac{1}{L^2 N} \sum_{l = 1}^L \sum_{l' = 1}^L \sum_{\text{i} = 1}^{\bm n} z_{\text{i}}^\star(s_l) z_{\text{i}}^\star(s_{l'})^\top \Gamma_{\text{i}}(s_l, s_{l'}), \\
    \mathcal{V}_{ab, nT}
    & \coloneqq \frac{N}{(NL)^2} \sum_{l=1}^L \sum_{l'=1}^L \bbE[ \bm{\mathcal{E}}(s_l)^\top \bm{D}^\top \bm{\calR}^\top P_a \bm{\calR} \bm{D} \bm{\mathcal{E}}(s_l) \bm{\mathcal{E}}(s_{l'})^\top \bm{D}^\top \bm{\calR}^\top P_b \bm{\calR} \bm{D} \bm{\mathcal{E}}(s_{l'}) ] \\
    & = \frac{2}{N} \sum_{1 \le \text{i}, \text{k} \le \bm n} \bar p_{a, \text{i}, \text{k}} \bar p_{b, \text{i}, \text{k}} \frac{1}{L^2} \sum_{l = 1}^L \sum_{l' = 1}^L \Gamma_{\text{i}}(s_l, s_{l'}) \Gamma_{\text{k}}(s_l, s_{l'}) ,
\end{align}
$z_\text{i}^\star(s)$ denotes the $\text{i}$-th column of $\bm{Z}(s)^\top \bm{D}^\top \bm{ \calR}^\top \bm{ \calR} \bm{D}$, and $\bar p_{m, \text{i}, \text{j}}$ is the $(\text{i}, \text{j})$-th element of $\bm{D}^\top \bm{\calR}^\top P_m \bm{\calR} \bm{D}$.
The definition of $\bar J_{nT}$ should be clear from the context.

Here, Assumption \ref{as:matrix1} is to be understood with $\Pi_{nT}$ re-defined as 
\begin{align}
    \Pi_{nT} \coloneqq (NL)^{-1} \sum_{l = 1}^L \bm{Z}(s_l)^\top \bm{D}^\top \bm{ \calR}^\top \bm{ \calR} \bm{D} \bbE[\bm{H}(s_l)],
\end{align}
which serves as the main identification condition in this model.
Similarly, Assumption \ref{as:matrix2}(ii) and (iii) are to be understood in terms of $\bar J_{nT}$ and $\mathcal{V}_{nT}$ re-defined here.

\bigskip

The proof of \eqref{eq:normal_twoway} is almost analogous to that of Theorem \ref{thm:normality}.
For example, by treating the matrix $\bm{\calR}^\top P_m \bm{\calR}$ in the same way as the matrix $P_m$ in the original model, the stochastic order evaluations regarding the quadratic moments can be reused without any changes.
That said, several lemmas require additional analysis or adjustment.
In what follows, rather than reproducing the entire proof of \eqref{eq:normal_twoway}, we focus only on the parts that require revision.

\paragraph{Lemma \ref{lem:cov}'}:
Suppose that Assumption \ref{as:error}(i) holds.
Let $\{\xi_{it}: (i,t) \in \mathcal{D}_{nT}^{\mathrm{ST}}\}$ be a uniformly and geometrically $L^2$-NED random field on $\{\varepsilon_{it}: (i,t) \in \mathcal{D}_{nT}^{\mathrm{ST}}\}$.
Denote $\bar \xi_{it} \coloneqq \sum_{k = 1}^n w_{i,k}\xi_{kt}$, $\ddot \xi_{it} \coloneqq \xi_{i,t+1}-\bar \xi_{i,t+1}-(\xi_{it}-\bar \xi_{it})$, and $C_\xi \coloneqq \sup_{i,t}\left\|\xi_{it}\right\|_2$.
Then, for all $i,j \in \mathcal{D}_n$ and $t,u \in [T-1]$,
\begin{align}
    \left|\text{Cov}\left(\ddot \xi_{it},\ddot \xi_{ju}\right)\right| \lesssim C_\xi^2 \rho\left(\Delta_{\mathrm{ST}}\left((i,t),(j,u)\right)/3\right)
\end{align}
with some geometric NED coefficient $\rho$.

\begin{proof}
    Decompose $\ddot \xi_{it} = \ddot \xi_{1,it}^{(\delta)}+\ddot \xi_{2,it}^{(\delta)}$, where
    \begin{align}
        \ddot \xi_{1,it}^{(\delta)} \coloneqq \bbE\left[\ddot \xi_{it}\mid\mathcal{F}_{nT,it}^{+}(\delta)\right], \;\; \text{and} \;\; \ddot \xi_{2,it}^{(\delta)} \coloneqq \ddot \xi_{it}-\bbE\left[\ddot \xi_{it}\mid\mathcal{F}_{nT,it}^{+}(\delta)\right],
    \end{align}
    where $\mathcal{F}_{nT,it}^{+}(\delta) \coloneqq \mathcal{F}_{nT,it}^{\mathrm{ST}}(\delta)\lor\mathcal{F}_{nT,i,t+1}^{\mathrm{ST}}(\delta)$.

    For each pair $\ddot \xi_{it}$ and $\ddot \xi_{ju}$, denote $\delta_{it,ju} \coloneqq \Delta_{\mathrm{ST}}((i,t),(j,u))/3$.
    If $\Delta_{\mathrm{ST}}((i,t),(j,u)) \le 3\bar \Delta_{\mathrm{ST}}$, the desired result follows immediately from the Cauchy--Schwarz inequality and $\left\|\ddot \xi_{it}\right\|_2 \lesssim C_\xi$.
    Hence, suppose that $\Delta_{\mathrm{ST}}((i,t),(j,u))>3\bar \Delta_{\mathrm{ST}}$.

    The minimum distance between a point in $\{(i,t),(i,t+1)\}$ and a point in $\{(j,u),(j,u+1)\}$ is at least $\Delta_{\mathrm{ST}}((i,t),(j,u))-1$.
    Since $\bar \Delta_{\mathrm{ST}}\ge 1$,
    \begin{align}
        2\delta_{it,ju} = \frac{2}{3}\Delta_{\mathrm{ST}}\left((i,t),(j,u)\right)<\Delta_{\mathrm{ST}}\left((i,t),(j,u)\right)-1.
    \end{align}
    Thus, $\mathcal{F}_{nT,it}^{+}(\delta_{it,ju})$ and $\mathcal{F}_{nT,ju}^{+}(\delta_{it,ju})$ are generated by disjoint collections of errors and are therefore independent under Assumption \ref{as:error}(i).

    It follows that
    \begin{align}
        \left|\text{Cov}\left(\ddot \xi_{it},\ddot \xi_{ju}\right)\right|
        & = \left|\text{Cov}\left(\ddot \xi_{1,it}^{(\delta_{it,ju})}+\ddot \xi_{2,it}^{(\delta_{it,ju})},\ddot \xi_{1,ju}^{(\delta_{it,ju})}+\ddot \xi_{2,ju}^{(\delta_{it,ju})}\right)\right| \\
        & \le \left|\text{Cov}\left(\ddot \xi_{1,it}^{(\delta_{it,ju})},\ddot \xi_{1,ju}^{(\delta_{it,ju})}\right)\right|+\left|\text{Cov}\left(\ddot \xi_{1,it}^{(\delta_{it,ju})},\ddot \xi_{2,ju}^{(\delta_{it,ju})}\right)\right| \\
        & \quad+\left|\text{Cov}\left(\ddot \xi_{2,it}^{(\delta_{it,ju})},\ddot \xi_{1,ju}^{(\delta_{it,ju})}\right)\right|+\left|\text{Cov}\left(\ddot \xi_{2,it}^{(\delta_{it,ju})},\ddot \xi_{2,ju}^{(\delta_{it,ju})}\right)\right|.
    \end{align}
    The first term on the right-hand side is zero by Assumption \ref{as:error}(i).
    By Jensen's and triangle inequalities,
    \begin{align}
        \left\|\ddot \xi_{1,it}^{(\delta_{it,ju})}\right\|_2
        & \le \left\|\ddot \xi_{it}\right\|_2 \\
        & \le \left\|\xi_{i,t+1}\right\|_2+\sum_{k = 1}^n |w_{i,k}|\left\|\xi_{k,t+1}\right\|_2+\left\|\xi_{it}\right\|_2+\sum_{k = 1}^n |w_{i,k}|\left\|\xi_{kt}\right\|_2 \lesssim C_\xi.
    \end{align}
    In addition, since $\{\xi_{it}\}$ is uniformly and geometrically $L^2$-NED,
    \begin{align}
        \left\|\ddot \xi_{2,it}^{(\delta_{it,ju})}\right\|_2
        & = \left\|\ddot \xi_{it}-\bbE\left[\ddot \xi_{it}\mid\mathcal{F}_{nT,it}^{+}(\delta_{it,ju})\right]\right\|_2 \\
        & \le \left\|\xi_{i,t+1}-\bbE\left[\xi_{i,t+1}\mid\mathcal{F}_{nT,i,t+1}^{\mathrm{ST}}(\delta_{it,ju})\right]\right\|_2+\left\|\xi_{it}-\bbE\left[\xi_{it}\mid\mathcal{F}_{nT,it}^{\mathrm{ST}}(\delta_{it,ju})\right]\right\|_2 \\
        & \quad+\sum_{k = 1}^n |w_{i,k}|\left\|\xi_{k,t+1}-\bbE\left[\xi_{k,t+1}\mid\mathcal{F}_{nT,k,t+1}^{\mathrm{ST}}(\delta_{it,ju}-\bar \Delta_{\mathrm{ST}})\right]\right\|_2 \\
        & \quad+\sum_{k = 1}^n |w_{i,k}|\left\|\xi_{kt}-\bbE\left[\xi_{kt}\mid\mathcal{F}_{nT,kt}^{\mathrm{ST}}(\delta_{it,ju}-\bar \Delta_{\mathrm{ST}})\right]\right\|_2 \\
        & \lesssim C_\xi\rho(\delta_{it,ju}).
    \end{align}
     Then, the rest of the proof is analogous to Lemma \ref{lem:cov}.
\end{proof}

\paragraph{Lemma \ref{lem:matLLN}'}: 
\small\begin{itemize}
        \item[(i')] $\left\| \sum_{l = 1}^L \bm{Z}(s_l)^\top \bm{D}^\top \bm{ \calR}^\top \bm{ \calR} \bm{D} (\bm{H}(s_l) - \bbE[\bm{H}(s_l)]) (\theta_0 - \theta) / (NL) \right\| \lesssim_p \sqrt{K}/\sqrt{nT}$
        \item[(iv')] $\left\| \sum_{l = 1}^L \bm{Z}(s_l)^\top \bm{D}^\top \bm{ \calR}^\top \bm{\calR} \bm{D} \bm{\mathcal{E}}(s_l) / (NL) \right\| \lesssim_p 1/\sqrt{nT}$
        \item[(v')] $\left| \sum_{l = 1}^L \left\{ \bm{E}(s_l; \theta)^\top \bm{D}^\top  \bm{\calR}^\top P_m \bm{\calR} \bm{D} \bm{E}(s_l; \theta) - \bbE[\bm{E}(s_l; \theta)^\top \bm{D}^\top \bm{\calR}^\top P_m \bm{\calR} \bm{D} \bm{E}(s_l; \theta)] \right\} / (NL) \right| \lesssim_p 1/\sqrt{nT}$.
\end{itemize}\normalsize

\begin{proof}

(i') Observe that for each $s_l$,
\begin{align}
    \left\| \bm{Z}(s_l)^\top \bm{D}^\top \bm{\calR}^\top \bm{\calR} \bm{D} (\bm{H}(s_l) - \bbE[\bm{H}(s_l)]) (\theta_0 - \theta) \right\|
    & \le \left\| \sum_{i = 1}^n \sum_{t = 1}^{T - 1} \ddot B_{it} \otimes \phi^K(s_l) \left[\ddot a_{1,it}(s_l) - \bbE[\ddot a_{1,it}(s_l)]\right] \alpha(s_l; \theta_0 - \theta) \right\| \\
    & \quad + \left\| \sum_{i = 1}^n \sum_{t = 1}^{T - 1} \ddot B_{it} \otimes \phi^K(s_l) \left[\ddot a_{2,it}(s_l) - \bbE[\ddot a_{2,it}(s_l)]\right] \gamma(s_l; \theta_0 - \theta) \right\|.
\end{align}
As a typical element, the variance of the first term involving $a_{1,it}$ is bounded as
\small\begin{align}
    \bbE\left( \sum_{i = 1}^n \sum_{t = 1}^{T - 1} \ddot Q^1_{it} \phi_1(s_l) \left(\ddot a_{1,it} - \bbE[\ddot a_{1,it}]\right) \alpha(s_l; \theta_0 - \theta) \right)^2 \lesssim \sum_{t = 1}^{T - 1} \sum_{t' = 1}^{T - 1} \sum_{i = 1}^n \sum_{i' = 1}^n \left| \text{Cov}\left(\ddot a_{1,it}, \ddot a_{1,i't'}\right) \right| \lesssim nT,
\end{align}\normalsize
where the last inequality follows from Lemmas \ref{lem:NED2} and \ref{lem:cov}' by the same argument as in the proof of Lemma \ref{lem:matLLN}(i).
The same argument applies to the terms involving $a_{2,it}$, and the desired result follows from Markov's and triangle inequalities.

\bigskip

(iv') By the triangle inequality,
\begin{align}
    \left\| \sum_{l = 1}^L \bm{Z}(s_l)^\top \bm{D}^\top \bm{ \calR}^\top \bm{ \calR} \bm{D} \bm{\mathcal{E}}(s_l) / (NL) \right\| 
    & = \left\| \sum_{l = 1}^L \sum_{t = 1}^{T-1} \sum_{i = 1}^n \ddot B_{it} \otimes  \phi^K(s_l) \ddot \varepsilon_{it}(s_l) / (NL) \right\| \\
    &  \le \left\| \sum_{l = 1}^L \sum_{t = 1}^{T - 1} \sum_{i = 1}^n \ddot B_{it} \otimes \varepsilon_{i,t+1}(s_l) \phi^K(s_l) / (NL) \right\| \\
    & \quad + \left\| \sum_{l = 1}^L \sum_{t = 1}^{T - 1} \sum_{i = 1}^n \ddot B_{it} \otimes \varepsilon_{it}(s_l) \phi^K(s_l) / (NL) \right\| \\
    &  \quad + \left\| \sum_{l = 1}^L \sum_{t = 1}^{T - 1} \sum_{i = 1}^n \sum_{j = 1}^n \ddot B_{it} w_{i,j} \otimes \varepsilon_{j,t+1}(s_l) \phi^K(s_l) / (NL) \right\| \\
    & \quad + \left\| \sum_{l = 1}^L \sum_{t = 1}^{T - 1} \sum_{i = 1}^n \sum_{j = 1}^n \ddot B_{it} w_{i,j} \otimes \varepsilon_{jt}(s_l) \phi^K(s_l) / (NL) \right\|.
\end{align}
The same argument as in the proof of Lemma \ref{lem:matLLN}(iv) shows that the first two terms on the right hand side are stochastically of order $1/\sqrt{nT}$.
For the third and fourth terms, observe that
\small\begin{align}
    & \bbE  \left\| \sum_{l = 1}^L \sum_{t = 1}^{T - 1} \sum_{i = 1}^n \sum_{j = 1}^n \ddot B_{it} w_{i,j} \otimes \varepsilon_{jt}(s_l) \phi^K(s_l) / (NL) \right\|^2 \\
    &  = \frac{1}{N^2} \sum_{t = 1}^{T-1} \sum_{i = 1}^n \sum_{i' = 1}^n \sum_{j = 1}^n  \text{trace}\left\{ (\ddot B_{it} w_{i,j}) (\ddot B_{i't}w_{i',j})^\top \otimes \frac{1}{L^2} \sum_{l = 1}^L \sum_{l' = 1}^L \Gamma_{jt}(s_l, s_{l'}) \phi^K(s_l)\phi^K(s_{l'})^\top \right\} \\
    & =  \frac{1}{N^2} \sum_{t = 1}^{T -1} \sum_{j = 1}^n \text{trace}\left\{ \left(\sum_{i = 1}^n \ddot B_{it} w_{i,j}\right) \left(\sum_{i' = 1}^n \ddot B_{i't}w_{i',j}\right)^\top \right\}\text{trace}\left\{ \frac{1}{L^2} \sum_{l = 1}^L \sum_{l' = 1}^L \Gamma_{jt}(s_l, s_{l'}) \phi^K(s_l)\phi^K(s_{l'})^\top\right\} \\
    & \lesssim 1/(nT).
\end{align}\normalsize
Then, the result follows from Markov's inequality.

\bigskip

(v') Let $\check P_{m,1} \coloneqq  \calR_n^\top P_{m,1}  \calR_n$ and $\check P_m \coloneqq I_{T - 1} \otimes \check P_{m,1}$.
Then, each $(i,j)$-th element of $\check P_{m,1}$ can be written as
\begin{align}
    \check p_{m,i,j} = \sum_{a = 1}^n \sum_{b = 1}^n r_{a,i} p_{m,a,b} r_{b,j}.
\end{align}
By Assumptions \ref{as:sample_space}(ii) and \ref{as:weights}(i), $\check p_{m,i,j} \neq 0$ requires the existence of at least one pair $(a,b)$ satisfying $\Delta(i,a) \le \bar \Delta$, $\Delta(a,b) \le \bar \Delta_m$, and $\Delta(j, b) \le \bar \Delta$, simultaneously.
In other words, no such pairs exist if $\Delta(i,j) > 2\bar \Delta + \bar \Delta_m$.
Hence, $\check P_{m,1} = (\check p_{m,i,j})$ satisfies the same threshold condition as in Assumption \ref{as:weights}(i), and thus (v') can be proved analogously to Lemma \ref{lem:matLLN}(v).
\end{proof}

Now, the rest of the proof of \eqref{eq:normal_twoway} proceeds in the same way as that of Theorem \ref{thm:normality}.

\bigskip

Consistent variance estimation can be performed as in Proposition \ref{prop:var}.
Specifically, with an abuse of notation, define
\begin{align}
    \hat{\mathcal{V}}_{z, nT}
    & \coloneqq \frac{1}{L^2 N} \sum_{l = 1}^L \sum_{l' = 1}^L \sum_{t = 1}^{T - 1} \sum_{t': \: |t' - t| \le 1} \sum_{i = 1}^n \check z_{it}(s_l) \check z_{it'}(s_{l'})^\top \vec{\hat e}_{it}(s_l) \vec{\hat e}_{it'}(s_{l'}) \\
    \hat{\mathcal{V}}_{ab, nT}
    & \coloneqq \frac{2}{L^2 N} \sum_{l = 1}^L \sum_{l' = 1}^L \sum_{t = 1}^{T - 1} \sum_{t': \: |t' - t| \le 1} \sum_{1 \le i, j \le n} \check p_{a, i, j} \check p_{b, i, j} \vec{\hat e}_{it}(s_l) \vec{\hat e}_{it'}(s_{l'}) \vec{\hat e}_{jt}(s_l) \vec{\hat e}_{jt'}(s_{l'}),
\end{align}
where $\check z_{it}(s)$ denotes the transpose of the corresponding $(i,t)$-th row of $\bm{\calR}^\top \bm{\calR}\bm{D}\bm{Z}(s)$.
Similarly, update the definitions of $\hat \sigma_{nT,\alpha}(s)$, $\hat \sigma_{nT,\gamma}(s)$, and $\hat \sigma_{nT,j}(s)$ using the re-defined $\hat{\mathcal{V}}_{z, nT}$ and $\hat{\mathcal{V}}_{ab, nT}$.
Then, following the same argument as in Proposition \ref{prop:var}, we can show the consistency of $\hat{\mathcal{V}}_{z, nT}$ and $\hat{\mathcal{V}}_{ab, nT}$.
Finally, Theorem \ref{thm:twoway} follows from Slutsky's theorem.

\clearpage 
\small
\bibliography{references.bib}
\end{document}